\def\lncs{0}

\ifnum\lncs=0
  \documentclass{article}
  \usepackage{fullpage}
  \usepackage{amsmath,amsthm}
  \usepackage{authblk}

  \usepackage{libertine}

  \newtheorem{theorem}{Theorem}
  \newtheorem{proposition}{Proposition}
  \newtheorem{definition}{Definition}
  \newtheorem{claim}{Claim}
  \newtheorem{lemma}{Lemma}
  \newtheorem{corollary}{Corollary}
  \newtheorem{observation}{Observation}
  \newtheorem{bound}{Bound}
  \newtheorem{fact}{Fact}

  \theoremstyle{definition}
  
  \newtheorem{remark*}{Remark}

\else
\fi
\usepackage{color}
\usepackage{ifthen}

\long\def\red#1\par{\par\bigskip{\color{red}#1}\par}
\def\showauthnotes{0}

\ifthenelse{\showauthnotes=1}
{
\newcommand{\authnote}[2]{{{[#1: #2]~}}}
}
{ \newcommand{\authnote}[2]{} }

\usepackage{microtype}
\usepackage{amsfonts,amssymb}
\usepackage{hyperref}
\usepackage{caption,subcaption}
\usepackage{tikz}
\usetikzlibrary{automata,arrows,positioning,calc}
\usepackage{pgfplots}
\usepackage{subfiles}
\usepackage[numbers, sectionbib]{natbib}
\usepackage{enumitem}
\usepackage{hhline}
\usepackage{multirow}

\newcommand{\ignore}[1]{}

\DeclareMathOperator{\Exp}{\mathbb{E}}

\newcommand{\gf}[1]{\mathsf{#1}}
\newcommand{\Z}{\mathbb{Z}}
\newcommand{\R}{\mathbb{R}}
\newcommand{\NN}{\mathbb{N}}
\newcommand{\hdepth}{\mathbf{d}}
\DeclareMathOperator{\length}{length}
\DeclareMathOperator{\depth}{depth}
\DeclareMathOperator{\height}{height}
\DeclareMathOperator{\gap}{gap}
\DeclareMathOperator{\reach}{reach}
\DeclareMathOperator{\reserve}{reserve}

\newcommand{\slot}{\textsl{sl}}
\newcommand{\fprefix}{\sqsubseteq}
\newcommand{\Adversary}{\mathcal{A}}
\newcommand{\Challenger}{\mathcal{C}}
\newcommand{\Distribution}{\mathcal{D}}
\newcommand{\dominatedby}{\preceq}
\newcommand{\StationaryRho}{\mathcal{R}_\infty}
\newcommand{\DistRho}{\mathcal{R}}
\newcommand{\Poly}{\mathrm{poly}}
\newcommand{\SuchThat}{\, : \,}

\newcommand{\PrefixEq}{\preceq}
\newcommand{\Prefix}{\prec}
\newcommand{\ForkPrefix}{\fprefix}

\newcommand{\TrimSlot}[1]{^{\lceil {#1}}}
\newcommand{\Fork}{\vdash}
\newcommand{\pinch}[2]{{#2}^{\vartriangleright {#1} \vartriangleleft} }
\newcommand{\cut}[2]{{#2}^{{#1} \vartriangleleft} }
\newcommand{\Chain}{\mathcal{C}}
\newcommand{\Intersect}{\cap}
\newcommand{\SlotCP}{\mathrm{CP}^{\mathsf{slot}}}
\newcommand{\kSlotCP}[1][k]{{#1}\text{-}\SlotCP}
\newcommand{\CP}{\mathrm{CP}}
\newcommand{\kCP}[1][k]{{#1}\text{-}\CP}
\newcommand{\defeq}{\triangleq}
\newcommand{\SlotDivergence}{\mathrm{div}_{\mathsf{slot}}}

\title{%Rigorous
  The combinatorics of the longest-chain rule:\\
  Linear consistency for proof-of-stake blockchains\thanks{
    Erica Blum's work was partly supported by financial assistance award 70NANB19H126 
    from U.S. Department of Commerce, National Institute of Standards and Technology. 
    Aggelos Kiayias' research was partly supported by H2020 Grant \#780477, PRIViLEDGE. 
    Cristopher Moore's research was partly supported by NSF grant BIGDATA-1838251. 
    Alexander Russell's work was partly supported by NSF Grant \#1717432.
  }
  }

\ifnum\lncs=1
\fi

\ifnum\lncs=0
  \author[1]{Erica Blum}
  \author[2,5]{Aggelos Kiayias}
  \author[3]{Cristopher Moore}
  \author[4]{Saad Quader}
  \author[4,5]{Alexander Russell}

  \affil[1]{University of Maryland, College Park}
  \affil[2]{University of Edinburgh}
  \affil[3]{Santa Fe Institute}
  \affil[4]{University of Connecticut}
  \affil[5]{IOHK}

  \else

\fi

\begin{document}
\maketitle

\begin{abstract}
  Blockchain data structures maintained via the longest-chain rule
  have emerged as a powerful algorithmic tool for consensus
  algorithms. The technique---popularized by the Bitcoin
  protocol---has proven to be remarkably flexible and now supports
  consensus algorithms in a wide variety of settings. Despite such
  broad applicability and adoption, current analytic understanding of
  the technique is highly dependent on details of the protocol's
  leader election scheme. A particular challenge appears in the
  proof-of-stake setting, where existing analyses suffer from
  quadratic dependence on suffix length.

  We describe an axiomatic theory of blockchain dynamics that permits
  rigorous reasoning about the longest-chain rule in quite general
  circumstances and establish bounds---optimal to within a
  constant---on the probability of a consistency violation. This settles
  a critical open question in the proof-of-stake setting where we
  achieve linear consistency for the first time.

  Operationally, blockchain consensus protocols achieve consistency by
  instructing parties to remove a suffix of a certain length from
  their local blockchain. While the analysis of Bitcoin guarantees
  consistency with error $2^{-k}$ by removing $O(k)$ blocks, recent
  work on proof-of-stake (PoS) blockchains has suffered from quadratic
  dependence: (PoS) blockchain protocols, exemplified by Ouroboros
  (Crypto 2017), Ouroboros Praos (Eurocrypt 2018) and Sleepy Consensus
  (Asiacrypt 2017), can only establish that the length of this suffix
  should be $\Theta(k^2)$.  This consistency guarantee is a
  fundamental design parameter for these systems, as the length of the
  suffix is a lower bound for the time required to wait for
  transactions to settle.  Whether this gap is an intrinsic limitation
  of PoS---due to issues such as the ``nothing-at-stake''
  problem---has been an urgent open question, as deployed PoS
  blockchains further rely on consistency for protocol correctness: in
  particular, security of the protocol itself relies on this
  parameter. Our general theory directly improves the required suffix
  length from $\Theta(k^2)$ to $\Theta(k)$. Thus we show, for the
  first time, how PoS protocols can match proof-of-work blockchain
  protocols for exponentially decreasing consistency error.

  Our analysis focuses on the articulation of a two-dimensional stochastic
  process that captures the features of interest, an exact recursive
  closed form for the critical functional of the process, and tail
  bounds established for associated generating functions that dominate
  the failure events. Finally, the analysis provides an explicit
  polynomial-time algorithm for exactly computing the
  exponentially-decaying error function which can directly inform
  practice.
\end{abstract}

\section{Introduction }

A blockchain is a data structure consisting of a collection of data
blocks placed in linear order. It further requires that
% It has the characteristic that
each block contains a collision-free hash of the previous block: thus
blocks implicitly commit to the entire prefix of the blockchain
preceding them. This elementary data structure has remarkable
applications in distributed computing, and now appears as an essential
component of consensus protocols in a wide variety of models and
settings; this notably includes both the ``permissionless'' setting
popularized by Bitcoin and the classic ``permissioned'' model.

Such consensus protocols call for players to collaboratively assemble
a blockchain by repeatedly selecting players to add blocks.
%A blockchain is a hash-chain that is collaboratively constructed by a
%set of players.
%
%New blocks are added to the chain by players chosen
Specifically, the protocol determines a stochastic process resembling
a lottery: each ``leader'' selected by the lottery is then responsible
for broadcasting a new block. While the algorithmic details of this
lottery depend heavily on the protocol, the outcome can be privately
determined and provides the winning player a proof of leadership that
can be publicly demonstrated. Assuming that the expected wait time for
some player to win the lottery is constant, the blockchain experiences
steady growth when players follow the protocol.

Network infelicities, adversarial behavior, or the possibility that
two players simultaneously win the lottery can lead to disagreements
among the players about the current blockchain. Thus protocols adopt a
``chain selection rule'' that determines how players should break ties
among the various chains they observe on the network; ideally, the
combination of the chain selection rule and the lottery should
guarantee that the players' blockchains agree, perhaps with the
exception of a short suffix. The emblematic chain selection strategy
among such systems is the \emph{longest-chain rule}, which calls for
players to adopt the longest chain among various contenders.

%Note that each player runs its own
%lottery attempting to extend a certain hash-chain it is aware of.
%Players choose which hash-chain to extend via some greedy criterion
%such as picking the longest hash-chain they have seen thus far.
%
%Implementing the lottery mechanism and the criterion for
%chain selection yields a particular blockchain protocol instantiation. 

The first blockchain protocol was the core of the sensational Bitcoin
system~\cite{Nakamoto2008}; it adopted a lottery mechanism based on a
cryptographic puzzle~\cite{C:DwoNao92,hashcash}---also
known as \emph{proof-of-work or PoW, for short}---and a chain
selection rule favoring chains that represent more work. The system is
particularly notable for its ability to survive in a
permissionless setting---where players may freely join and
depart---even when a portion of the players are actively attacking the
protocol. Unfortunately, the proof-of-work mechanism makes quite
striking energy demands: the system currently consumes as much
electricity as a small country.\footnote{See e.g., 
\url{https://digiconomist.net/bitcoin-energy-consumption} where it is reported
that Bitcoin 
annual energy consumption is on the order of at least 50 Twhr at the time of writing. } This motivated the blockchain
community to exploring alternative lottery mechanisms, e.g.,
proof-of-stake
(PoS)~\cite{DBLP:journals/corr/BentovGM14,DBLP:conf/asiacrypt/PassS17,KRDO17},
proof of
space~\cite{C:DFKP15,DBLP:journals/iacr/ParkPAFG15}
and others~\cite{cryptoeprint:2016:035}. The proof-of-stake mechanism
is particularly attractive from the perspective of efficiency, as it
makes no assumption of external computational resources.

The fundamental consistency property---critical in all these
blockchain systems---is \emph{common-prefix}
(cf.~\cite{DBLP:conf/eurocrypt/GarayKL15}). It precisely captures the
intuition described above: by trimming a $k$-block suffix from the
chain held by any honest player the resulting blockchain is a prefix
of the blockchain possessed by any honest party at any future point of
the execution. A principal goal in the analysis of these systems is a
to guarantee common prefix, for an appropriate value of $k$, even if
some of the players collude to disrupt the protocol. Common prefix is
typically only shown to hold with high probability $1-\varepsilon$,
where $\varepsilon$ is an error term that is a function of
$k$. The exact dependency of $\varepsilon$ on $k$ is critically
important: it determines the length of the suffix that is to be
removed from a blockchain in order to ensure that the remaining prefix
will be retained at any future point of the execution. This directly
imposes a lower bound on how long one has to wait for information in
the blockchain (such as a payment transaction) to ``settle.''
Additionally, many blockchain protocols internally rely on common
prefix for correctness; thus the relationship between $\varepsilon$ and
$k$ is critical to establishing the regime of correctness of the
entire protocol.

%Under mild assumptions one must have
A relatively straightforward lower bound for $\varepsilon$ is 
$\varepsilon \geq \exp(-\alpha k )$ for some $\alpha>0$.  This lower
bound applies when there is a coalition of adversarial players of
constant fraction,
% that may be adversarial which is
the case of primary interest in practice. The result is easy to infer
from the analysis of~\cite{Nakamoto2008}, where a strategy is
demonstrated that violates common prefix with such probability (this
is referred to as a ``double-spending'' attack in that paper). The
tightness of this bound is an important open problem. For the special
case of proof-of-work an upper bound of $\exp(-\Omega(k))$ was shown
first in \cite{DBLP:conf/eurocrypt/GarayKL15} and further verified in
extended security models
by~\cite{DBLP:conf/crypto/GarayKL17,DBLP:conf/eurocrypt/PassSS17}.
%(here the parameter $T$ corresponds to the total run-time of the
%blockchain).
In the proof-of-stake setting, on the other hand, the tightness of the
bound remains open.  While recent proof-of-stake algorithms have been
presented with rigorous analyses that rival proof-of-work in many
regards, they suffer from a quadratic relationship between $k$ and
$\log(\varepsilon)$. For example, the Ouroboros
protocols~\cite{KRDO17,DBLP:conf/eurocrypt/DavidGKR18,DBLP:journals/iacr/BadertscherGKRZ18},
as well as Snow White~\cite{DBLP:journals/iacr/BentovPS16a}, provide
an upper bound on $\varepsilon$ of $\exp(-\Omega(\sqrt{k}))$; this
should be compared with $\varepsilon = \exp(-\Theta(k))$ for
proof-of-work. The significant gap from the known lower bound was
attributed to a notable, general attack that distinguished PoS from
PoW: Known as the \emph{nothing-at-stake} problem, this refers to the
ability of an adversarial coalition of players to strategically reuse
a winning PoS lottery to extend multiple blockchains.

%%%%%%%%%%%%%%%%

\paragraph{Our results.}
Our objective is to control the common-prefix error
$\varepsilon$ as tightly as possible while
making minimal assumptions on the underlying blockchain
protocol.  We work in a general model formulated by a simple family of
\emph{blockchain axioms}. The axioms themselves are easy to interpret
and few in number. This permits us to abstract many features of the
underlying blockchain protocol (e.g., the details of the
leader-election process, the cryptographic security of the relevant
signature schemes and hash functions, and randomness generation),
while still establishing results that are strong enough to directly
incorporate into existing specific analyses.
%%(For example, our techniques can directly improve
%the analysis of Snow White, Ouroboros, Ouroboros Praos, and Ouroboros
%Genesis.)

Our most interesting finding is a quite tight theory of common prefix
that depends \emph{only on the schedule of participants certified to
  add a block}. Under common assumptions about this schedule, we
achieve the optimal relationship $\varepsilon =
\exp(-\Theta(k))$. This directly improves the common prefix guarantees
(and settlement times) of existing proof-of-stake blockchains such as
Snow White~\cite{DBLP:journals/iacr/BentovPS16a},
Ouroboros~\cite{KRDO17}, Ouroboros
Praos~\cite{DBLP:conf/eurocrypt/DavidGKR18}, and Ouroboros
Genesis~\cite{DBLP:journals/iacr/BadertscherGKRZ18}.
%
%analysis of~\cite{KRDO17}, which established that the probability of
%a ``depth-$k$'' settlement failure at time $T$ was no more than
%$T \exp(-\Omega(\sqrt{k}))$. Our new techniques establish
%that the probability of a settlement failure at time $T$ is no more
%than $\exp(-\Omega(k))$. 
%%Note that this is independent of the
%%the time elapsed $T$, 
Specifically, this improves the scaling in the exponent from
$\sqrt{k}$ to $k$ and establishes a tight characterization for
$\varepsilon = \exp(-\Theta(k))$. (In fact, we even obtain reasonable
control of the constants.)  We remark that our assumptions about the
schedule distribution can be weakened---without any effect on the
final bounds---to apply to martingale-style distributions such
as those that arise in the analysis of adaptive
adversaries~\cite{DBLP:conf/eurocrypt/DavidGKR18,DBLP:journals/iacr/BadertscherGKRZ18}.
%While we discuss this in detail later, we remark that
%this is important for applying our techniques to security proofs
%involving adaptive adversaries.% and, in particular, to the
%analysis of the Ourorboros Genesis protocol.

Our new analysis offers an additional, but lower order, improvement for
several of these blockchains. The existing analysis of, e.g.,
Ouroboros Praos~\cite{DBLP:conf/eurocrypt/DavidGKR18}, required a
union bound to be taken over the entire lifetime of the protocol in
order to rule out a common prefix violation at a particular point of
time; thus such events were actually bounded above by a function of
the form $T \exp(-\Omega(\sqrt{k}))$, where $T$ is the lifetime of the
protocol. While this event \emph{does} depend on the entire dynamics
of the protocol, we show how to avoid this pessimistic tail bound to
achieve a ``single shot'' common prefix violation---at a particular
time of interest---of form $\exp(-\Theta(k))$; this removes the
dependence on $T$.

From a technical perspective, we contrast the structure of our proofs
with existing techniques for the PoW case. The PoW results find a
direct connection between common-prefix and the behavior of a biased,
one-dimensional random walk. Interestingly, our results give a tight
relationship between the general (e.g., PoS) case and a pair of
\emph{coupled} biased random walks. A major challenge in the analysis
is to bound the behavior of this richer stochastic process.  Our tools
yield precise, explicit upper bounds on the probability of persistence
violations that can be directly applied to tune the parameters of
deployed PoS systems. See Appendix~\ref{sec:exact-prob} where we
record some concrete results of the general theory. The importance of
these results in the practice of PoS blockchain systems cannot be
overstated: they provide, for the first time, concrete error bounds
for settlement times for PoS blockchains that follow the longest chain
rule.

\paragraph{Further analytic details.} Our approach begins with the
graph-theoretic framework of \emph{forks} and \emph{margin} developed
for the analysis of the Ouroboros~\cite{KRDO17} protocol.  (A fork is
an abstraction of the protocol execution given the outcomes of the
leader-election process.)  We begin by generalizing the notion of
margin to account for local, rather than global, features of a leader
schedule, and provide an exact, recursive closed form for this new
quantity (see Section~\ref{sec:recursion}). This proof identifies an
optimal online adversary (i.e., a fork-building strategy whose current
decisions do not depend on the future) for PoS blockchain algorithms
%when
%divergence is measured in slots. 
%In addition, we present an \emph{optimal} online adversary 
with the remarkable property that the sequence of forks produced by
this adversary \emph{simultaneously} achieve the worst-case (slot)
common-prefix violations associated with all slots (see
Section~\ref{sec:canonical-forks}). We then study the stochastic
process generated when the \emph{characteristic string}---a Boolean
string representing the outcome of the leader election scheme---is
given by a family of i.i.d.\ Bernoulli random variables. In this case,
we identify a generating function that bounds the tail events off
interest, and analytically upper bound the growth of the function. We
then show how to extend the analysis to the setting where the
characteristic string is drawn from a martingale sequence.  As it
happens, this more general distribution arises naturally in the
analyses of PoS protocols that survive adaptive adversaries; e.g.,
Ouroboros Genesis~\cite{DBLP:journals/iacr/BadertscherGKRZ18}.  We
obtain the pleasing result that the common prefix error probability in
the martingale case is no more than that in the i.i.d.\ Bernoulli
case.

\paragraph{Direct consequences.} 
Our results establish consistency bounds in a quite general
setting---see below: In particular, they directly imply
$\exp(-\Theta(k))$ consistency for the Sleepy consensus (Snow
White)~\cite{DBLP:conf/asiacrypt/PassS17}, Ouroboros~\cite{KRDO17},
Ouroboros Praos~\cite{DBLP:conf/eurocrypt/DavidGKR18}, and Ouroboros
Genesis~\cite{DBLP:journals/iacr/BadertscherGKRZ18} blockchain
protocols. (The Ouroboros Praos and Ouroboros Genesis analyses in fact
directly relied on an earlier e-print version of the present article
for their settlement estimates.)

\paragraph{Related work.} 
Blockchain protocol analysis in the PoW-setting was initiated
in~\cite{DBLP:conf/eurocrypt/GarayKL15} and further improved in
\cite{DBLP:conf/eurocrypt/PassSS17,DBLP:conf/crypto/GarayKL17}.  The
established security bounds for consistency are linear in the security
parameter.  Sleepy
consensus~\cite[Theorem~13]{DBLP:conf/asiacrypt/PassS17} provides a
consistency bound of the form
$\exp(-\Omega(\sqrt{k}))$. %\footnote{We note that the term $T$ stems from a union bound applied over all communication rounds (CHECK XXX).}
Note that~\cite{DBLP:conf/asiacrypt/PassS17} is not a PoS protocol per
se, but it is possible to turn it into one (as was demonstrated in
\cite{DBLP:journals/iacr/BentovPS16a}). The analysis of the Ouroboros
blockchain~\cite{KRDO17} achieves $\exp(-\Omega(\sqrt{k}))$. We remark
that the analyses of Ouroboros
Praos~\cite{DBLP:conf/eurocrypt/DavidGKR18} and Ouroboros
Genesis~\cite{DBLP:journals/iacr/BadertscherGKRZ18} developed
significant new machinery for handling other challenges (e.g.,
adaptive adversaries, partial synchrony), but directly referred to a
preliminary version of this article to conclude their guarantees of
$ \exp(-\Omega(k))$.

%which were only experimentally verified. 

Our results complement the recent results of
\cite{DBLP:journals/corr/abs-1809-06528}, which also considers
longest-chain PoS protocols. \cite{DBLP:journals/corr/abs-1809-06528}
focuses on identifying dynamics unique to longest-chain PoS
protocols. In particular, they show that longest-chain PoS protocols
that are \emph{predictable} (i.e., for which some portion of the
schedule of slot leaders is known ahead of time) are necessarily
vulnerable to ``predictable double-spends.''  The conventional defense
against such attacks is to wait for the specified settlement time to
elapse before accepting a transaction, which (until now) has resulted
in slow confirmation times. As such,
\cite{DBLP:journals/corr/abs-1809-06528} raised the question of
whether long confirmation times are a necessary evil in longest-chain
PoS blockchains.  As double-spending attacks imply a consistency
violation, our results show that PoS protocols can safely decrease
settlement times to asymptotically match PoW protocols without
sacrificing security against double-spends.

Because we focus on the longest-chain rule, our
analysis is not applicable to protocols like
Algorand~\cite{DBLP:journals/corr/Micali16} which, in fact, offer
settlement in expected constant time without invoking blockchain
reorganisation or forks; however, Algorand lacks the ability to
operate in the ``sleepy''~\cite{DBLP:conf/asiacrypt/PassS17} or
``dynamic availability''~\cite{DBLP:journals/iacr/BadertscherGKRZ18}
setting.
% (which permits an evolving population of participants).
In our combinatorial analysis, synchronous operation is assumed
against a rushing adversary; this is without loss of generality
vis-a-vis the result of \cite{DBLP:conf/eurocrypt/DavidGKR18} where it
was shown how to reduce the combinatorial analysis in the partially
synchronous setting to the synchronous one.  We note that a number of
works have shown how to use a blockchain protocol to bootstrap a
cryptographic protocol that can offer faster settlement time under
stronger assumptions than honest majority, e.g., Hybrid
Consensus~\cite{DBLP:conf/wdag/PassS17} or
Thunderella~\cite{DBLP:conf/eurocrypt/PassS18}; our results are
orthogonal and synergistic to those since they can be used to improve
the settlement time bounds of the blockchain protocol that operates as
a fallback mechanism.

\paragraph{Outline.} We begin in Section~\ref{sec:model} by describing
a simple general model for blockchain
dynamics. Section~\ref{sec:definitions} builds on this model to set
down a number of basic definitions required for the proofs. The first
part of the main proof is described in Section~\ref{sec:recursion},
which develops a ``relative'' version of the theory of margin
from~\cite{KRDO17}; most details are then relegated to
Section~\ref{sec:margin-proof} in order to move quickly to the
consistency estimates in Section~\ref{sec:estimates}.  In
Section~\ref{sec:canonical-forks}, we present an optimal online
adversary who can simultaneously maximize the relative margins for all
prefixes of the characteristic string.  \ignore{We then provide two
  different settlement estimates in Section~\ref{sec:estimates};
  roughly, the two bounds trade off generality with the strength of
  the final estimates.} Finally, in Appendix~\ref{sec:exact-prob}, we
compute exact upper bounds on $k$-settlement error probabilities for
various values of $k$ and describe a simple $O(k^3)$-time algorithm to
compute these probabilities in general.
%The \texttt{C++} source code is publicly available 
%at~\href{https://github.com/saad0105050/forkable-strings-code}{https://github.com/saad0105050/forkable-strings-code}~\cite{PrForkableCode}.

%%% Local Variables:
%%% mode: latex
%%% TeX-master: "main"
%%% End:

\section{The blockchain axioms and the settlement security model}
\label{sec:model}

Typical blockchain consensus protocols
%---including Bitcoin, Ouroboros, Snow
%White, and Ouroboros Praos---
call for each participant to maintain a \emph{blockchain}; this is a
data structure that organizes transactions and other protocol metadata
into an ordered historical record of ``blocks.'' A basic design goal
of these systems is to guarantee that participants' blockchains always
agree on a common prefix; the differing suffixes of the chains held
by various participants roughly correspond to the possible future
states of the system. Thus the major analytic challenge is to ensure
that---despite evolving adversarial control of some of the
participants---the portion of honest participants' blockchains that
might pairwise disagree is confined to a short suffix. This analysis
in turn supports the fundamental guarantee of \emph{consistency}
for these algorithms, which asserts that data appearing
deep enough in the chain can be considered to be stable, or ``settled.''

\begin{sloppypar}
  We adopt a discrete notion of time organized into a sequence of
  \emph{slots} $\{\slot_0, \slot_1, \ldots\}$ and assume all protocol
  participants have the luxury of synchronized clocks that report the
  current slot number. As discussed above, the protocols we consider
  rely on two algorithmic devices:
\begin{itemize}
\item A \emph{leader election mechanism}, which randomly assigns to each time
  slot a set of ``leaders'' permitted to post a new block in that slot.
\item The \emph{longest-chain rule}, which calls for the leader(s) of
  each slot to add a block to the end of the longest blockchain she
  has yet observed, and broadcast this new chain to other participants.
\end{itemize}
The Bitcoin protocol uses a proof-of-work mechanism to carry out
leader election, which can be modeled using a random oracle
\cite{DBLP:conf/eurocrypt/GarayKL15,DBLP:conf/eurocrypt/PassSS17,DBLP:conf/crypto/GarayKL17}.
Proof-of-stake systems typically require more intricate leader
election mechanisms; for example, the Ouroboros
protocol~\cite{KRDO17} uses a full multi-party private computation to distribute clean randomness,
while Snow White~\cite{DBLP:journals/iacr/BentovPS16a}, Algorand
\cite{DBLP:journals/corr/Micali16}, and Ouroboros Praos
\cite{DBLP:conf/eurocrypt/DavidGKR18} use hashing and a family of
values determined on-the-fly. Despite these differences, all existing
analyses show that the leader election mechanism suitably approximates
an ideal distribution, which is also the approach we will adopt for our analysis.

%Given a blockchain protocol that uses these two mechanisms, our
%analysis allows us to draw meaningful conclusions about its consistency
%properties.
\end{sloppypar}

\subsection{The blockchain axioms and forks}

To simplify our analysis, we assume a synchronous communication
network in the presence of a \emph{rushing} adversary: in particular,
any message broadcast by an honest participant at the beginning of a
particular slot is received by the adversary first, who may decide
strategically and individually for each recipient in the network
whether to inject additional messages and in what order all messages
are to be delivered prior to the conclusion of the slot. (See
\S\ref{sec:model-comments} below for comments on this network
assumption.)

Given this, the behavior of the protocol when carried out by a group
of honest participants (who follow the protocol in the presence of an
adversary who may only reorganize messages) is clear. Assuming that
the system is initialized with a common ``genesis block''
corresponding to $\slot_0$ and the leader election process in fact
elects a single leader per slot, the players observe a common,
linearly growing blockchain:
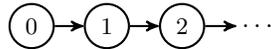
\begin{center}
  \begin{tikzpicture}[>=stealth', auto, semithick,
    flat/.style={circle,draw=black,thick,text=black,font=\small}]
    \node[flat]    at (0,0)  (base) {$0$};
    \node[flat]    at (1,0)  (n1) {$1$};
    \node[flat]    at (2,0)  (n2) {$2$};
    \node[flat,white]    at (3,0)  (n3) {$\ \ \  $};
    \node at (3,0) {$\ldots$};
    \draw[thick,->] (base) to (n1);
    \draw[thick,->] (n1) to (n2);
    \draw[thick,->] (n2) to (n3);
  \end{tikzpicture}
\end{center}
\noindent
Here node $i$ represents the block broadcast by the leader of slot $i$
and the arrows represent the direction of increasing time. (Note that
the requirement of a single leader per slot is important in this
simple picture; it is possible for a network adversary to induce
divergent views between the players by taking advantage of slots where
more than a single honest participant is elected a leader.)

\paragraph{The blockchain axioms: Informal discussion.}
The introduction of adversarial participants or multiple slot leaders
complicates the family of possible blockchains that could emerge from
this process. To explore this in the context of our protocols, we work
with an abstract notion of a blockchain which
% (as informally suggested above)
ignores all internal structure. We consider a fixed assignment of
leaders to time slots, and assume that the blockchain uses a proof
mechanism to ensure that any block labeled with slot $\slot_t$ was
indeed produced by a leader of slot $\slot_t$; this is guaranteed in
practice by appropriate use of a secure digital signature scheme.

Specifically, we treat a \emph{blockchain} as
a sequence of abstract blocks, each labeled with a slot number, so
that:
\begin{enumerate}[label={\textbf{A\arabic*}}., series=axiom]
\item The blockchain begins with a fixed ``genesis'' block, assigned to slot $\slot_0$.
\item The (slot) labels of the blocks are in strictly increasing order.
\end{enumerate}
It is further convenient to introduce the structure of a directed
graph on our presentation, where each block is treated as a vertex; in
light of the first two axioms above, a blockchain is a path beginning
with a special ``genesis'' vertex, labeled $0$, followed by vertices
with strictly increasing labels that indicate which slot is associated
with the block. %(See the example below.)
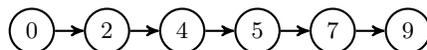
\begin{center}
  \begin{tikzpicture}[>=stealth', auto, semithick,
    flat/.style={circle,draw=black,thick,text=black,font=\small}]
    \node[flat]    at (0,0)  (base) {$0$};
    \node[flat]    at (1,0)  (n1) {$2$};
    \node[flat] at (2,0)  (n2) {$4$};
    \node[flat] at (3,0)  (n3) {$5$};
    \node[flat] at (4,0)  (n4) {$7$};
    \node[flat] at (5,0)  (n5) {$9$};
    \draw[thick,->] (base) to (n1);
    \draw[thick,->] (n1) to (n2);
    \draw[thick,->] (n2) to (n3);
    \draw[thick,->] (n3) -- (n4);
    \draw[thick,->] (n4) -- (n5);
  \end{tikzpicture}
\end{center}
The protocols of interest call for honest players to add a
\emph{single} block %(to a single previous chain in its local state)
during any slot. In particular:
\begin{enumerate}[label={\textbf{A\arabic*}}., resume=axiom]
\item If a slot $\slot_t$ was assigned to a single honest player, 
then a single block is created---during the entire protocol---with the label $\slot_t$.
\end{enumerate}
Recall that blockchains are \emph{immutable} in the sense that any
block in the chain commits to the entire previous history of the
chain; this is achieved in practice by including with each block a
collision-free hash of the previous block. These properties imply that
if a specific slot $\slot_t$ was assigned to a unique honest player,
then any chain that includes the unique block from $\slot_t$ must also include
that block's associated prefix in its entirety.

As we analyze the dynamics of blockchain algorithms, it is convenient
to maintain an entire family of blockchains at once. As a matter of
bookkeeping, when two blockchains agree on a common prefix, we can
glue together the associated paths to reflect this, as indicated
below.
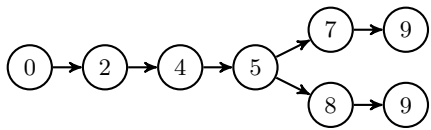
\begin{center}
  \begin{tikzpicture}[>=stealth', auto, semithick,
    flat/.style={circle,draw=black,thick,text=black,font=\small}]
    \node[flat]    at (0,0)  (base) {$0$};
    \node[flat]    at (1,0)  (n1) {$2$};
    \node[flat] at (2,0)  (n2) {$4$};
    \node[flat] at (3,0)  (n3) {$5$};
    \node[flat] at (4,.5)  (n4a) {$7$};
    \node[flat] at (5,.5)  (n5a) {$9$};
    \node[flat] at (4,-.5)  (n4b) {$8$};
    \node[flat] at (5,-.5)  (n5b) {$9$};
    \draw[thick,->] (base) to (n1);
    \draw[thick,->] (n1) to (n2);
    \draw[thick,->] (n2) to (n3);
    \draw[thick,->] (n3) to (n4a);
    \draw[thick,->] (n4a) to (n5a);
    \draw[thick,->] (n3) to (n4b);
    \draw[thick,->] (n4b) to (n5b);
  \end{tikzpicture}
  \end{center}
  When we glue together many chains to form such a diagram, we call it
  a ``fork''---the precise definition appears below. Observe that
  while these two blockchains agree through the vertex (block) labeled
  5, they contain (distinct) vertices labeled 9; this reflects two
  distinct blocks associated with slot 9 which, in light of the axiom
  above, must have been produced by an adversarial participant.
  
  Finally, as we assume that messages from honest players are
  delivered without delay, we note a direct consequence of the longest
  chain rule:
\begin{enumerate}[label={\textbf{A\arabic*}}., resume=axiom]
\item If two honestly generated blocks $B_1$ and $B_2$ are labeled
  with slots $\slot_1$ and $\slot_2$ for which $\slot_1 < \slot_2$,
  then the length of the unique blockchain terminating at $B_1$ is
  strictly less than the length of the unique blockchain terminating at $B_2$.
\end{enumerate}
Recall that the honest participant assigned to slot
$\slot_2$ will be aware of the blockchain terminating at $B_1$ that
was broadcast by the honest player in slot $\slot_1$ as a result of
synchronicity; according to the longest-chain rule, it must have
placed $B_2$ on a chain that was at least this long. In contrast, not
all participants are necessarily aware of all blocks generated by
dishonest players, and indeed dishonest players may often want to
delay the delivery of an adversarial block to a participant or show
one block to some participants and show a completely different block
to others.

\paragraph{Characteristic strings, forks, and the formal axioms.}
Note that with the axioms we have discussed above, whether or not a
particular fork diagram (such as the one just above) corresponds to a valid
execution of the protocol depends on how the slots have been awarded to the parties by the
leader election mechanism. We introduce the notion of a ``characteristic'' string as a convenient
means of representing information about slot leaders in a given execution.
\begin{definition}[Characteristic string]
  Let $\slot_1, \ldots, \slot_{n}$ be a sequence of slots. A \emph{characteristic string} $w$ is an element of $\{0,1\}^n$ defined for a particular execution of a blockchain protocol so that
  \[
    w_t =   \begin{cases}
    0 & \text{if $\slot_{t}$ was assigned to a single honest participant},\\
%    1 & \text{if $\slot_{t}$ was assigned to an adversarial participant},
    1 & \text{otherwise.}
  \end{cases}
\]
\end{definition}
For two Boolean strings $x$ and $w$, 
we write $x \Prefix w$ iff $x$ is a strict prefix of $w$. 
Similarly, 
we write $x \PrefixEq w$ iff either $x = w$ or $x \Prefix w$. 
The empty string $\varepsilon$ is a prefix to any string. 
With this discussion behind us, we set down the formal object we use
to reflect the various blockchains adopted by honest players during
the execution of a blockchain protocol. This definition formalizes the blockchains axioms discussed above.

%%%%Forks

\begin{definition}[Fork; \cite{KRDO17}]
  Let $w\in \{0,1\}^n$ and let $H = \{ i \mid w_i = 0\}$. A
  \emph{fork} for the string $w$ consists of a directed and rooted
  tree $F=(V,E)$ with a labeling $\ell:V\to\{0,1,\dots,n\}$. We insist
  that each edge of $F$ is directed away from the root vertex and
  further require that
  \begin{enumerate}[label=(F{\arabic*}.)]
  \item\label{fork:root} the root vertex $r$ has label $\ell(r)=0$;
  \item\label{fork:monotone} the labels of vertices along any directed path are strictly increasing;
  \item\label{fork:unique-honest} each index $i\in H$ is the label for exactly one vertex of $F$;
  \item\label{fork:honest-depth} for any  vertices $i,j\in H$, if $i<j$, then the depth of vertex $i$ in $F$ is strictly less than the depth of vertex $j$ in $F$.
  \end{enumerate}
\end{definition}

If $F$ is a fork for the characteristic string $w$, we write
$F\vdash w$.  Note that the
conditions~\ref{fork:root}--\ref{fork:honest-depth} are direct
analogues of the axioms A1--A4 above. See Fig.~\ref{fig:fork} for an
example fork. A final notational convention: If $F \vdash x$ and
$\hat{F} \vdash w$, we say that $F$ is a \emph{prefix} of $\hat{F}$,
written $F \fprefix \hat{F}$, if 
% the string $x \in \{0,1\}^\ell$ is a prefix of the string $w \in \{0,1\}^{\ell + m}$ 
$x \PrefixEq w$
and $F$ appears as a
consistently-labeled subgraph of $\hat{F}$. 
(Specifically, each path of $F$ appears, with identical labels, in $\hat{F}$.)

% Observe that any string of the form $0^k$ has a unique
% fork consisting of a single path:
% \begin{center}
%   \begin{tikzpicture}[scale=1,>=stealth', auto, semithick,
%     flat/.style={circle,draw=black,thick,text=black,font=\small}]
%     \node[flat] at (0,0)  (base) {$0$};
%     \node[flat] at (1,0)  (n1) {$1$};
%     \node[flat] at (2,0)  (n2) {$2$};
%     \node[flat] at (5,0)  (n5) {$k$};
%     \draw[thick,->] (base) to (n1);
%     \draw[thick,->] (n1) to (n2);
%     \draw[thick,->,dashed] (n2) to (n5);
%   \end{tikzpicture}
% \end{center}
% On the other hand, there are in fact an infinite number of forks for
% any string with at least one ``1,'' as slots for which $w_i = 1$ can
% be associated with any number of vertices. Axiom~\ref{fork:monotone}
% reflects that any legal chain must consist of blocks with increasing
% slot labels. Axiom~\ref{fork:unique-honest} reflects the fact that honest plays produce a single block. Axiom~\ref{fork:honest-depth} reflects that each new
% honest vertex is always placed at a depth strictly greater than all
% previous honest vertices, because honest users always choose to add
% their block to the longest visible chain, and we assume honest blocks
% can be seen by all users. (By contrast, adversaries may play on
% shorter tines, or may ``hide'' dishonest blocks from other users until
% later slots.)

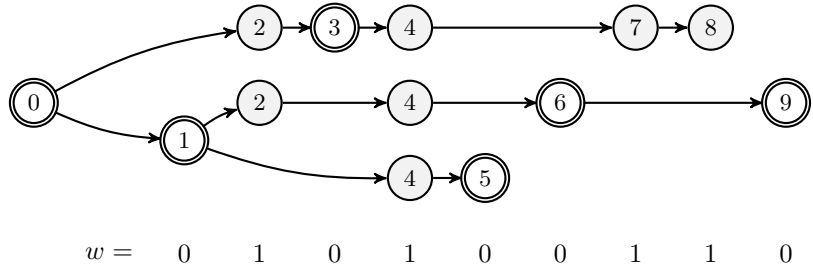
\begin{figure}[th]
\centering
\begin{tikzpicture}[>=stealth', auto, semithick,
  honest/.style={circle,draw=black,thick,text=black,double,font=\small},
  malicious/.style={fill=gray!10,circle,draw=black,thick,text=black,font=\small}]
    \node at (0,-2) {$w =$};
    \node at (1,-2) {$0$};
    \node[honest]    at (1,-.5)  (ab1) {$1$};
    \node at (2,-2) {$1$};
    \node[malicious] at (2,0)  (b2) {$2$}; \node[malicious] at (2,1) (c1) {$2$};
    \node at (3,-2) {$0$};    \node[honest]    at (3,1)  (c2) {$3$};
    \node at (4,-2) {$1$};    \node[malicious] at (4,0)  (b3) {$4$}; \node[malicious] at (4,-1) (a2) {$4$};
\node[malicious] at (4,1) (c3) {$4$};
    \node at (5,-2) {$0$};    \node[honest]    at (5,-1) (a3) {$5$};
    \node at (6,-2) {$0$};    \node[honest]    at (6,0)  (b4) {$6$};
    \node at (7,-2) {$1$};    \node[malicious] at (7,1)  (c4) {$7$};
    \node at (8,-2) {$1$};    \node[malicious] at (8,1)  (c5) {$8$};
    \node at (9,-2) {$0$};    \node[honest]    at (9,0)  (b5) {$9$};
    \node[honest] at (-1,0) (base) {$0$};
    % \node[state,honest] at (3,-1) (bottom) {};
    % \node[state,honest] at (7,1) (top) {$H$};
    \draw[thick,->] (base) to[bend left=10] (c1);
    \draw[thick,->] (base) to[bend right=10] (ab1);
    \draw[thick,->] (ab1) to[bend right=10] (a2);
    \draw[thick,->] (a2) -- (a3);
    \draw[thick,->] (ab1) to[bend left=10] (b2);
    \draw[thick,->] (b2) -- (b3);
    \draw[thick,->] (b3) -- (b4);
    \draw[thick,->] (b4) -- (b5);
    \draw[thick,->] (c1) -- (c2);
    \draw[thick,->] (c2) -- (c3);
    \draw[thick,->] (c3) -- (c4);
    \draw[thick,->] (c4) -- (c5);
    % \draw[thick,<->] (3,0) -- (7,0) node[pos=.5] {$\gap(f)$};
    \end{tikzpicture}
    \caption{A fork $F$ for the characteristic string $w = 010100110$;
      vertices appear with their labels and honest vertices are
      highlighted with double borders. Note that the depths of the
      (honest) vertices associated with the honest indices of $w$ are
      strictly increasing. Note, also, that this fork has two disjoint
      paths of maximum depth.}
    \label{fig:fork}
  \end{figure}

  Let $w$ be a characteristic string.  The directed paths in the fork
  $F \Fork w$ originating from the root are called \emph{tines}; these
  are abstract representations of blockchains. (Note that a tine might
  not terminate at a leaf of the fork.)  We naturally extend the label
  function $\ell$ for tines: i.e., $\ell(t) \triangleq \ell(v)$ where
  the tine $t$ terminates at vertex $v$. The length of a tine $t$ is
  denoted by $\length(t)$.

 \paragraph{Viable tines.}
 The longest-chain rule dictates that honest players build on chains
 that are at least as long as all previously broadcast honest
 chains. It is convenient to distinguish such tines in the analysis:
 specifically, a tine $t$ of $F$ is called \emph{viable} if its length
 is at least the depth of any honest vertex $v$ for which
 $\ell(v) \leq \ell(t)$. A tine $t$ is \emph{viable at slot $s$} if
 the portion of $t$ appearing over slots $0,\ldots, s$ has length at
 least that of any honest vertices labeled from this set. (As noted,
 the properties~\ref{fork:unique-honest} and~\ref{fork:honest-depth}
 together imply that an honest observer at slot $s$ will only adopt a
 viable tine.)  The \emph{honest depth} function
 $\hdepth : H \rightarrow [n]$ gives the depth of the (unique) vertex
 associated with an honest slot; by~\ref{fork:honest-depth},
 $\hdepth(\cdot)$ is strictly increasing.
  
  \subsection{Settlement and the common prefix
    property}\label{sec:cp-settlement}
  
  We are now ready to explore the power of an adversary in this
  setting who has corrupted a (perhaps evolving) coalition of the
  players. We focus on the possibility that such an adversary can
  blatantly confound consistency of the honest player's
  blockchains. In particular, we consider the possibility that, at
  some time $t$, the adversary conspires to produce two blockchains of
  maximum length that diverge prior to a previous slot $s \leq t$; in
  this case honest players adopting the longest-chain rule may clearly
  disagree about the history of the blockchain after slot $s$. We call
  such a circumstance a \emph{settlement violation}.

  To reflect this in our abstract language, let $F \Fork w$ be a fork
  corresponding to an execution with characteristic string $w$. Such a
  settlement violation induces two viable tines $t_1, t_2$ with the
  same length that diverge prior to a particular slot of interest. We
  record this below.
  
\begin{definition}[Settlement with parameters $s,k \in \NN$]\label{def:settlement}
  Let $w \in \{0,1\}^n$ be a characteristic string. Let
  $F \Fork w_1 \ldots w_t$ be a fork for a prefix of $w$ with
  $s + k \leq t \leq n$.  We say that a slot $s$ is \emph{not
    $k$-settled} in $F$ if the fork contains two tines $t_1, t_2$ of
  maximum length that ``diverge prior to $s$,'' i.e., they either
  contain different vertices labeled with $s$, or one contains a
  vertex labeled with $s$ while the other does not. Note that such
  tines are viable by definition.
  % admits a $\kSlotCP$ violation 
  % witnessed by two tines $t_1, t_2$ 
  % such that $s + k \leq \ell(t_1) \leq \ell(t_2)$; 
  Otherwise, \emph{slot $s$ is $k$-settled in $F$}. We say that a slot
  $s$ is \emph{$k$-settled} (for the characteristic string $w$) if it
  is $k$-settled in every fork $F \Fork w_1, \ldots w_t$, for each
  $t \geq s+k$.
  % containing two viable tines $t_1, t_2$ which 
  % (i.) satisfy $s + k \leq \ell(t_1) \leq \ell(t_2)$ and 
  % (ii.) the pair $(t_1, t_2)$ witness a $\kSlotCP$ violation.
\end{definition}

\paragraph{Common prefix.} Settlement violations are a convenient and
intuitive proxy for the notion of common prefix discussed in the
introduction. Indeed, as we show in Section~\ref{sec:cp-forks}, the
two notions are equivalent, so we have the luxury of discussing
settlement violations which have the advantage of a more ready
interpretation. Concretely, we will simultaneously upper bound---using
the same analytic techniques---the probability of settlement
violations and common prefix violations.

Recall that the common prefix property with parameter $k$ asserts
that, for any slot index $s$, if an honest observer at slot $s + k$
adopts a blockchain $\Chain$, the prefix $\Chain[0 : s]$ will be
present in every honestly-held blockchain at or after slot $s + k$.
(Here, $\Chain[0 : s]$ denotes the prefix of the blockchain $\Chain$
containing only the blocks issued from slots $0, 1, \ldots, s$.)

We translate this property into the framework of forks.  Consider a
tine $t$ of a fork $F \vdash w$.  The \emph{trimmed} tine
$t\TrimSlot{k}$ is defined as the portion of $t$ labeled with slots
$\{ 0, \ldots, \ell(t) - k\}$. For two tines, we use the notation
$t_1 \PrefixEq t_2$ to indicate that the tine $t_1$ is a
prefix of tine $t_2$.

\begin{definition}[Common Prefix Property with parameter $k \in \NN$]\label{def:cp-slot}
  Let $w$ be a characteristic string. A fork $F \vdash w$ satisfies
  $\kSlotCP$ if, for all pairs $(t_1, t_2)$ of viable tines $F$ for
  which $\ell(t_1) \leq \ell(t_2)$, we have $t_1\TrimSlot{k} \PrefixEq t_2$. 
  Otherwise, we say that the tine-pair $(t_1, t_2)$ is a witness to a $\kSlotCP$ violation.
  Finally, \emph{$w$ satisfies $\kSlotCP$} if every fork $F \vdash w$ satisfies $\kSlotCP$.
  %We
  %denote this property by $\kSlotCP$.
  % Let $\Chain_1$ and $\Chain_2$ be two blockchains adopted by 
  % two (not necessarily distinct) honest players 
  % at the onset of slots $r_1$ and $r_2$, respectively, with $r_1 \leq r_2$. 
  % Then $\Chain_1\TrimSlot{k} \PrefixEq \Chain_2$. 
  % We denote this property by $\kSlotCP$.
\end{definition} 
If a string $w$ does not possess the $\kSlotCP$ property, 
we say that \emph{$w$ violates $\kSlotCP$}.
Observe that we defined the common prefix property in terms of
deleting any blocks associated with the \emph{last $k$ trailing slots}
from a local blockchain $\Chain$.  Traditionally
(cf. \cite{C:GarKiaLeo17}), this property has been defined in terms of
deleting a suffix of (block-)length $k$ from $\Chain$.  We denote the
block-deletion-based version of the common prefix property as the
$\kCP$ property.  Note, however, that a $\kCP$ violation immediately
implies a $\kSlotCP$ violation, so bounding the probability of a
$\kSlotCP$ violation is sufficient to rule out both events.

\subsection{Adversarial attacks on settlement time; the settlement game}\label{sec:game} 

To clarify the relationship between forks and the chains at play in a
canonical blockchain protocol, we define a game-based model below that
explicitly describes the relationship between forks and executions.
By design, the probability that the adversary wins this game is at
most the probability that a slot $s$ is not $k$-settled. We remark
that while we focus on settlement violations for clarity, one could
equally well have designed the game around common prefix violations.

Consider the \emph{$(\Distribution,T;s,k)$-settlement game}, played
between an adversary $\Adversary$ and a challenger $\Challenger$ with
a leader election mechanism modeled by an ideal distribution
$\Distribution$. Intuitively, the game should reflect the ability of
the adversary to achieve a settlement violation; that is, to present
two maximally-long viable blockchains to a future honest observer,
thus forcing them to choose between two alternate histories which
disagree on slot $s$.
%Specifically, the $\Adversary$ will furnish a fork with
%two viable tines such that (i.) they diverge prior to slot $s$ and
%(ii.) both have (equal and) maximal length among all chains produced
%by the protocol at some later time $t \geq s+k$.
The challenger plays the role(s) of the honest players during the
protocol.

Note that in typical PoS settings the distribution $\Distribution$ is
determined by the combined stake held by the adversarial players, the
leader election mechanism, and the dynamics of the protocol. The most
common case (as seen in Snow White~\cite{DBLP:conf/asiacrypt/PassS17}
and Ouroboros~\cite{KRDO17}) guarantees that the characteristic string
$w = w_1 \ldots w_T$ is drawn from an i.i.d.\ distribution for which
$\Pr[w_i = 1] \leq (1 - \epsilon)/2$; here the constant
$(1-\epsilon)/2$ is directly related to the stake held by the
adversary. Settings involving adaptive adversaries (e.g., Ouroboros
Praos~\cite{DBLP:conf/eurocrypt/DavidGKR18} and Ouroboros
Genesis~\cite{DBLP:journals/iacr/BadertscherGKRZ18}) yield the weaker
martingale-type guarantee that
$\Pr[w_i = 1 \mid w_1, \ldots, w_{i-1}] \leq (1 - \epsilon)/2$.

\begin{center}
  \fbox{
    \begin{minipage}{.9 \textwidth}
      \begin{center}
        \textbf{The $(\Distribution,T;s,k)$-settlement game}
      \end{center}
      \begin{enumerate}

      \item A characteristic string $w \in \{0,1\}^T$ is drawn from
        $\mathcal{D}$ and provided to $\Adversary$. (This reflects the results of the leader
        election mechanism.)

      \item Let $A_0 \vdash \varepsilon$ denote the initial fork for
        the empty string $\varepsilon$ consisting of a single node
        corresponding to the genesis block.

      \item For each slot $t = 1, \ldots, T$ in increasing order:
        \begin{enumerate}

        \item If $w_t = 0$, this is an honest slot. In this case, the
          challenger is given the fork
          $A_{t-1} \vdash w_1 \ldots w_{t-1}$ and must determine a new
          fork $F_{t} \vdash w_1 \ldots w_t$ by adding a single vertex
          (labeled with $t$) to the end of a longest path in
          $A_{t-1}$.  (If there are ties, $\Adversary$ may choose
          which path the challenger adopts.)

        \item If $w_t = 1$, this is an adversarial slot. $\Adversary$
          may set $F_t \vdash w_1\ldots w_t$ to be an arbitrary fork
          for which $A_{t-1} \fprefix F_t$.
          
        \item (Adversarial augmentation.) $\Adversary$ determines an
          arbitrary fork $A_t \vdash w_1 \ldots, w_{t}$ for
          which $F_{t} \fprefix A_{t}$.
        \end{enumerate}
         Recall that $F \fprefix F'$ indicates that $F'$
          contains, as a consistently-labeled subgraph, the fork $F$.
      \end{enumerate}
      % With each slot $\slot_i$, we can associate a subset $\mathcal{V}_i$ of paths in the fork $A_t$ 
      % such that the paths in $\mathcal{V}_i$ are \emph{viable at $\slot_i$}, i.e., 
      %   they are at least as long as the depth of the last honest vertex at or prior to $\slot_i$. 
      % (Note that the association $\slot_i \mapsto \mathcal{V}_i$ is not unique: 
      % multiple subsets can be associated with the same slot.)
      
      % We say that 
      $\Adversary$ \emph{wins the settlement game} if slot $s$ is not
      $k$-settled in some fork $A_t$ (with $t \geq s+k$).
    %   We say that $\Adversary$ \emph{wins} the settlement game if 
    %   some $t, a, b$ satisfying $s + k \leq a \leq b \leq t$, 
    %   the fork $A_t$ contains two tines, 
    %   $\Chain_a$ and $\Chain_b$
    %   % $\Chain_a \in \mathcal{V}_a$ and $\Chain_b \in \mathcal{V}_b$, 
    %   such that 
    %   (i.) $\Chain_a$ (resp. $\Chain_b$) is viable at slot $a$ (resp. slot $b$); 
    %   (i.) $\Chain_a$ (resp. $\Chain_b$) has the maximal length among all viable tines at slot $a$ (resp. slot $b$); and 
    %   (ii.) $\Chain_a$ and $\Chain_b$ \emph{diverge prior to $\slot_s$}---specifically, 
    %   they either
    %   contain different vertices labeled with $s$, or one contains a
    %   vertex labeled with $s$ while the other does not. 
    % %   We say that $\Adversary$ \emph{wins} the settlement game if, for
    % %   some $t \geq s+k$, there are two paths in the fork $A_t$ where both paths 
    % %   (i.) have the maximal length among all paths in the fork and
    % %   (ii.) \emph{diverge prior to $\slot_s$}---specifically, they either
    % %   contain different vertices labeled with $s$ or one contains a
    % %   vertex labeled $s$ and the other does not.
    \end{minipage}
  }
\end{center}

\begin{definition}
  Let $\Distribution$ be a distribution on $\{0,1\}^T$. Then define
  the \emph{$(s,k)$-settlement insecurity} of $\Distribution$ to be
  \[
    \mathbf{S}^{s,k}[\Distribution] \triangleq \max_{\Adversary}\, \Pr[\text{$\Adversary$ wins the $(\Distribution, T; s, k)$-settlement game}]
    \,,
  \]
  this maximum taken over all adversaries $\Adversary$.
\end{definition}

\paragraph{Remarks.}
% A subset $\mathcal{V}_i$ associated with the slot $\slot_i$ represents the ``local, viable chains'' shown to an honest observer by the system; he updates his state by adopting a maximally-long chain from this subset where ties are broken by the adversary. It is likely that different observers will have a different view of the current state of the system. 
% Next, observe that the adversary wins 
% the $(\Distribution, T; s,k)$-settlement game if 
% he can create a fork containing two tines that 
% (i.) witness a $\kSlotCP$ violation while
% (ii.) ``disagreeing'' about slot $s$. 
Observe that the adversarial augmentation step permits the adversary to
``suddenly'' inject new paths in the fork between two honest
players at adjacent slots; this corresponds to circumstances
when the adversary chooses to deliver a new blockchain to an honest
participant which may consist of an earlier honest chain with some
adversarial blocks appended to the end. Observe, additionally, that
the behavior of the challenger in the game is entirely deterministic,
as it simply plays according to the longest-chain rule (even
permitting the adversary to break ties). Thus the result of the game
is entirely determined by the characteristic string $w$ drawn from
$\Distribution$ and the choices of the adversary $\Adversary$. 
% For a Boolean string $w$, let $w[i : j], 1 \leq i \leq j$ denote the substring 
% $w_i w_{i+1} \ldots w_j$. 
We record the
following immediate conclusion:
\begin{lemma}\label{lem:main-forks}
  Let 
  $s, k, T \in \NN$. 
  Let $\Distribution$ be a distribution on $\{0,1\}^T$. Then
  \[
    \mathbf{S}^{s,k}[\Distribution] \leq 
      \Pr_{w \sim \Distribution}[\text{slot $s$ is not $k$-settled for $w$}]
    \,.
  %     \Pr_w\left[\parbox{70mm}{
  %       there exists an integer $t \geq s + k$ 
  %       and a fork $F \Fork w[1 : t]$ such that 
  %       $F$ contains two viable tines that 
  %       violate $\kSlotCP$ and 
  %       diverge prior to slot $s$}
  %       % there exists a prefix of $w$ 
  %       % of length at least $s + k$ 
  %       % and a fork $F$ on the prefix such that 
  %       % $F$ contains two tines that 
  %       % violate $\kSlotCP$ and 
  %       % diverge prior to slot $s$}
  %     \right]\,,
  \]
  % where the string $w$ is drawn from the distribution $\Distribution$.
\end{lemma}

In the subsequent sections, we will develop some further notation and
tools to analyze this event.  We will investigate two different
families of distributions, those with i.i.d.\ coordinates and those
with martingale-type conditioning guarantees. For $T \in \NN$ and
$\epsilon \in (0, 1)$, let $B_\epsilon = (B_1, \ldots, B_n)$ denote
the random variable taking values in $\{0,1\}^n$ so that the $B_i$ are
independent and $\Pr[B_i = 1] = (1 - \epsilon)/2$; we let
$\mathcal{B}_\epsilon$ denote the distribution on $\{0,1\}^n$
associated with $B_\epsilon$. When $\epsilon$ can be inferred from
context, we simply write $B$ and $\mathcal{B}$.
% More generally, we say
% that a random variable $W = (W_1, \ldots, W_n)$ taking values in
% $\{0,1\}^T$ is \emph{$\epsilon$-conditional} if, for each $i$,
% \begin{equation}\label{eq:conditioning}
%   \Pr[w_i=1 \mid w_1, \ldots, w_{i-1}] \leq (1 - \epsilon)/2\,.
% \end{equation}
% (This indicates arbitrary conditioning on the previous $i-1$
% variables.) 
% Note that $\mathcal{B}_\epsilon$ is $\epsilon$-conditional.  

We also study a more general family of distributions, defined next.

\begin{definition}[$\epsilon$-martingale condition]\label{def:eps-martingale}
  Let $W = (W_1, \ldots, W_n)$ be a random variable taking values in
  $\{0,1\}^n$.  We say that $W$ satisfies the
  \emph{$\epsilon$-martingale condition} if for each
  $t \in \{1, \ldots, n\}$,
  \[
    \Exp[W_t \mid W_1, \cdots, W_{t-1}] \leq (1-\epsilon)/2\,.
  \]
  Equivalently,
  $\Pr[W_t = 1\mid W_1, \ldots, W_{t-1}] \leq (1-\epsilon)/2$. 
  The conditioning on the variables
  $W_1, \cdots, W_{t-1}$ is arbitrary in both cases; as a consequence,
  $\Pr[W_t = 1] \leq (1-\epsilon)/2$. As a matter of notation, we let
  $\mathcal{W}$ denote the distribution associated with the random
  variable $W$.  We use the term ``$\epsilon$-martingale condition''
  to qualify both a random variable and its distribution.
\end{definition}
There are settings, such as
Genesis~\cite{DBLP:journals/iacr/BadertscherGKRZ18}, where this
martingale-type conditioning is important.  Note that
$\mathcal{B}_\epsilon$ satisfies the $\epsilon$-martingale condition.
Now we are ready to state our main theorem.

\begin{theorem}[Main theorem]\label{thm:main}
  Let $\epsilon \in (0, 1), s, k, T \in \NN$.  Let
  $\mathcal{W}$ and $\mathcal{B}_\epsilon$ be two distributions on
  $\{0,1\}^T$ where $\mathcal{B}_\epsilon$ is defined above and
  $\mathcal{W}$ satisfies the $\epsilon$-martingale condition.
  Then 
  % both $\mathbf{S}^{s,k}[\mathcal{W}]$ and $\mathbf{S}^{s,k}[\mathcal{B}_\epsilon]$ 
  % are at most 
  % \[
  %   \exp\bigl(-\Omega(\epsilon^3 (1 - O(\epsilon))k)\bigr)
  %   \,.
  % \]
  \[
    \mathbf{S}^{s,k}[\mathcal{W}] 
      \leq \mathbf{S}^{s,k}[\mathcal{B}_\epsilon] 
      \leq \exp\bigl(-\Omega(\epsilon^3 (1 - O(\epsilon))k)\bigr)
    \,.
  \]
  (Here, the asymptotic notation hides constants that do not depend on $\epsilon$ or $k$.)
\end{theorem}

By techniques similar to the ones used to prove this result, 
we obtain the following theorem pertaining
directly to $\kSlotCP$ (and $\kCP$).
\begin{theorem}[Main theorem; $\kCP$ version] \label{thm:main-CP} Let
  $\epsilon \in (0,1)$ and $T \in \NN$. Let $w \in \{0,1\}^T$ be a
  random variable satisfying the $\epsilon$-martingale condition.
  Then
  \[
    \Pr[\text{$w$ violates $\kCP$}] 
      \leq \Pr[\text{$w$ violates $\kSlotCP$}] 
      \leq T \cdot \exp\bigl(-\Omega(\epsilon^3 (1 - O(\epsilon))k)\bigr)
      \,.
  \]
\end{theorem}
The proofs of these theorems are presented in Section~\ref{sec:thm-proofs}.
Additionally, we provide a $O(k^3)$-time algorithm for 
computing an explicit upper bound on these probabilities; cf. Appendix~\ref{sec:exact-prob}.

\subsection{Survey of the proofs of the main theorems}\label{sec:args-survey}
A central object in our combinatorial analysis is an ``$x$-balanced fork'' 
for a characteristic string $w = xy$. 
Such a fork contains two distinct, %(viable,)
maximum-length tines 
that are disjoint over $y$; 
see Definition~\ref{def:balanced-fork} for details. 
% An $x$-balanced fork is important because 
% In Observation~\ref{obs:settlement-balanced-fork}, we observe that a
A settlement violation for the slot $|x| + 1$ implies an $x$-balanced fork for the string $xy$; 
see Observation~\ref{obs:settlement-balanced-fork}. 
In particular, for any distribution on characteristic strings in $\{0,1\}^n$ and
$s + k \leq n$,
\[
  \Pr_{w}[\text{slot $s$ is not $k$-settled} ] 
    \leq
  % \Pr_{w}\left[\parbox{75mm}{exists a decomposition $w = xyz$ with $|x| = s - 1, |y| \geq k+1, |z| \geq 0$ and an $x$-balanced fork for $xy$}\right] \,.
    \Pr_w\left[\parbox{55mm}{
      there is a decomposition $w = xyz$ and a fork $F \Fork xy$, 
      where $|x| = s - 1$ and $|y| \geq k + 1$, 
      so that $F$ is $x$-balanced
    }\right] 
    \,.
\]
(This is a variant of Lemma~\ref{lemma:settlement-margin} from Section~\ref{sec:thm-proofs}.)

As promised above, common prefix violations can be handled the same
way: we likewise establish (see Section~\ref{sec:cp-forks}; Theorem~\ref{thm:cp-fork}) that a common
prefix violation implies that there exists a balanced fork for some prefix of
$w$.  Specifically, for any distribution of characteristic strings, 
\begin{equation}\label{eq:pr-cp-fork}
  \Pr_w[\text{$w$ violates $\kSlotCP$}] 
    \leq 
    \Pr_{w}\left[\parbox{55mm}{
      there is a decomposition $w = xyz$ and a fork $F \Fork xy$, 
      where $|y| \geq k + 1$, 
      so that $F$ is $x$-balanced
    }\right] 
    \,. 
\end{equation}

Next, in Section~\ref{sec:recursion}, 
we give a recursive expression for the combinatorial quantity 
``relative margin,'' written
$\mu_x(y)$ (see Definition~\ref{def:margin} in Section~\ref{sec:definitions}). 
We establish that,
for an arbitrary decomposition of the characteristic string $w = xy$, 
the event ``there is an $x$-balanced fork for $xy$'' 
is equivalent to the event 
``the relative margin $\mu_x(y)$ is non-negative;'' 
this is Fact~\ref{fact:margin-balance}. 
In Lemma~\ref{lem:relative-margin}, we develop an exact recursive presentation for $\mu_x(y)$; hence we can bound the probability of a common prefix violation
(or a settlement violation) 
by reasoning about the non-negativity of the relative margin 
and, in particular, without reasoning directly about forks. 

In Section~\ref{sec:estimates}, we prove two bounds for the probability 
\[
  \Pr_{\substack{w = xy\\|x| = s}}[\mu_x(y) \geq 0]\,,
\]
for a fixed length $s$.  The first bound pertains to the setting where
$w = xy$ is drawn from $\mathcal{B}_\epsilon$. The second pertains to
any distribution $\mathcal{W}$ satisfying the $\epsilon$-martingale
condition.  For characteristic strings with distribution
$\mathcal{B}_\epsilon$, we identify a random variable which
stochastically dominates $\mu_x(y)$ and is amenable to exact analysis
via generating functions; this yields the bound
\[
  \Pr_{w=xy}[\mu_x(y) \geq 0] \leq \exp(-\Omega(|y|))
  \,.
\]
Notice that this bound does not depend on $s$, the length of $x$.  The
result for distributions satisfying the $\epsilon$-martingale
condition then follows from stochastic dominance
(Lemma~\ref{lemma:rho-stationary}).  See Section~\ref{sec:estimates}
for details.

%; the uncertainty comes
%entirely from the characteristic string.
It immediately follows that 
an $(s,k)$-settlement violation (or a $\kSlotCP$ violation) is a rare event 
for distributions of interest. 
The multiplicative factor $T$ in Theorem~\ref{thm:main-CP} comes from a union bound 
taken over all prefixes of $w$.

\subsection{Comments on the model}
\label{sec:model-comments}

\paragraph{Analysis in the $\Delta$-synchronous setting.} The security
game above most naturally models a blockchain protocol over a
synchronous network with immediate delivery (because each ``honest''
play of the challenger always builds on a fork that contains the fork
generated by previous honest plays). However, the model can be easily
adapted to protocols in the $\Delta$-synchronous model adopted by
the Snow White and Ouroboros Praos protocols and analyses. In
particular, \citet{DBLP:conf/eurocrypt/DavidGKR18} developed a
``$\Delta$-reduction'' mapping on the space of characteristic strings
that permits analyses of forks (and the related statistics of
interest, cf. \S\ref{sec:definitions}) in the $\Delta$-synchronous
setting by a direct appeal to the synchronous setting.

\paragraph{Public leader schedules.} One attractive feature of this
model is that it gives the adversary full information about the future
schedule of leaders. The analysis of some protocols indeed demand this
(e.g., Ouroboros, Snow White). Other protocols---especially those
designed to offer security against adaptive adversaries (Praos,
Genesis)---in fact contrive to keep the leader schedule private. Of
course, as our analysis is in the more difficult ``full information''
model, it applies to all of these systems.

\paragraph{Bootstrapping multi-phase algorithms; stake shift.} We remark that
several existing proof-of-stake blockchain protocols proceed in
phases, each of which is obligated to generate the randomness (for
leader election, say) for the next phase based on the current stake
distribution. The blockchain security properties of each phase are
then individually analyzed---assuming clean randomness---which yields
a recursive security argument; in this context the game outlined above
precisely reflects the single phase analysis.

%%% Local Variables:
%%% mode: latex
%%% TeX-master: "main"
%%% End:

\section{Definitions}
\label{sec:definitions}
We rely on the elementary framework of forks and margin
from~\citet{KRDO17}. We restate and briefly discuss the pertinent
definitions below. With these basic notions behind us, we then define
a new ``relative'' notion of margin, which will allow us to
significantly improve the efficacy of these tools for reasoning about
settlement times.
%In particular, these tools will allow us to reason
%about the possibility that an adversary can produce two alternate
%histories of the blockchain that diverge prior to a particular block.

Recall that for a given execution of the protocol, we record the
result of the leader election process via a \emph{characteristic
  string} $w \in \{0,1\}^T$, defined such that $w_i = 0$ when a unique
and honest party is assigned to slot $i$ and $w_i = 1$ otherwise.
A vertex of a fork is said to be \emph{honest}
  if it is labeled with an index $i$ such that $w_i=0$.

\begin{definition}[Tines, length, and height]
  Let $F \vdash w$ be a fork for a characteristic string.  A
  \emph{tine} of $F$ is a directed path starting from the root. For
  any tine $t$ we define its \emph{length} to be the number of edges
  in the path, and for any vertex $v$ we define its \emph{depth} to be
  the length of the unique tine that ends at $v$. 
  If a tine $t_1$ is a strict prefix of another tine $t_2$, we write $t_1 \Prefix t_2$. 
  Similarly, if $t_1$ is a non-strict prefix of $t_2$, we write $t_1 \PrefixEq t_2$.
  The longest common prefix of two tines $t_1, t_2$ is denoted by $t_1 \Intersect t_2$. 
  That is, $\ell(t_1 \Intersect t_2) = \max\{\ell(u) \SuchThat \text{$u \PrefixEq t_1$ and $u \PrefixEq t_2$} \}$. 
  The \emph{height} of
  a fork (as usual for a tree) is the length of the longest tine,
  denoted $\height(F)$. 
\end{definition}

\begin{definition}[The $\sim_x$ relations]
  For two tines $t_1$ and $t_2$ of a fork $F$, we write $t_1 \sim t_2$
  when $t_1$ and $t_2$ share an edge; otherwise we write
  $t_1 \nsim t_2$. We generalize this equivalence relation to reflect
  whether tines share an edge over a particular suffix of $w$: for
  $w = xy$ we define $t_1 \sim_x t_2$ if $t_1$ and $t_2$ share an edge
  that terminates at some node labeled with an index in $y$;
  otherwise, we write $t_1 \nsim_x t_2$ (observe that in this case the
  paths share no vertex labeled by a slot associated with $y$).  We
  sometimes call such pairs of tines \emph{disjoint} (or, if
  $t_1 \nsim_x t_2$ for a string $w = xy$, \emph{disjoint over
    $y$}). Note that $\sim$ and $\sim_\varepsilon$ are the same
  relation.
\end{definition}

%Informally, $t_1\sim_x t_2$ indicates that when we restrict our view of history to only blocks \emph{after}
%the prefix $x$, $t_1$ and $t_2$ share an edge (and thus agree on at least one block after that point).

The basic structure we use to use to reason about settlement times is
that of a ``balanced fork.''

\begin{definition}[Balanced fork; cf.\ ``flat'' in \cite{KRDO17}]\label{def:balanced-fork} A
  fork $F$ is \emph{balanced} if it contains a pair of tines $t_1$ and
  $t_2$ for which $t_1\nsim t_2$ and
  $\length(t_1)=\length(t_2)=\height(F)$. We define a relative notion
  of balance as follows: a fork $F \vdash xy$ is \emph{$x$-balanced}
  if it contains a pair of tines $t_1$ and $t_2$ for which
  $t_1 \not\sim_x t_2$ and $\length(t_1) = \length(t_2) = \height(F)$.
\end{definition}

Thus, balanced forks contain two completely disjoint, maximum-length
tines, while $x$-balanced forks contain two maximum-length tines that
may share edges in $x$ but must be disjoint over the rest of the
string. 
See Figures~\ref{fig:balanced} and~\ref{fig:x-balanced} 
for examples of balanced forks.
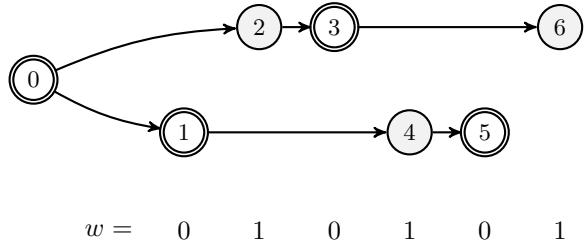
\begin{figure}[ht]
  \centering
  \begin{tikzpicture}[>=stealth', auto, semithick,
    honest/.style={circle,draw=black,thick,text=black,double,font=\small},
   malicious/.style={fill=gray!10,circle,draw=black,thick,text=black,font=\small}]
    \node at (0,-2) {$w =$};
  \node at (1,-2) {$0$}; \node[honest] at (1,-.7) (b1) {$1$};
  \node at (2,-2) {$1$}; \node[malicious] at (2,.7) (a1) {$2$};
  \node at (3,-2) {$0$}; \node[honest] at (3,.7) (a2) {$3$};
  \node at (4,-2) {$1$}; \node[malicious] at (4,-.7) (b2) {$4$};
  \node at (5,-2) {$0$}; \node[honest] at (5,-.7) (b3) {$5$};
  \node at (6,-2) {$1$}; \node[malicious] at (6,.7) (a3) {$6$};
    \node[honest] at (-1,0) (base) {$0$};
  \draw[thick,->] (base) to[bend left=10] (a1);
      \draw[thick,->] (a1) -- (a2);
      \draw[thick,->] (a2) -- (a3);
  \draw[thick,->] (base) to[bend right=10] (b1);
      \draw[thick,->] (b1) -- (b2);
      \draw[thick,->] (b2) -- (b3);
    \end{tikzpicture} 
  \caption{A balanced fork}
  \label{fig:balanced}
\end{figure}

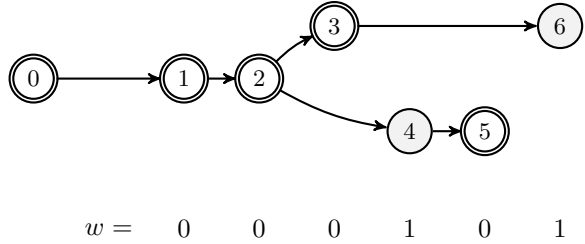
\begin{figure}[ht]
  \centering
  \begin{tikzpicture}[>=stealth', auto, semithick,
    honest/.style={circle,draw=black,thick,text=black,double,font=\small},
   malicious/.style={fill=gray!10,circle,draw=black,thick,text=black,font=\small}]
    \node at (0,-2) {$w =$};
  \node at (1,-2) {$0$}; \node[honest] at (1,0) (ab1) {$1$};
  \node at (2,-2) {$0$}; \node[honest] at (2,0) (ab2) {$2$};
  \node at (3,-2) {$0$}; \node[honest] at (3,.7) (a3) {$3$};
  \node at (4,-2) {$1$}; \node[malicious] at (4,-.7) (b3) {$4$};
  \node at (5,-2) {$0$}; \node[honest] at (5,-.7) (b4) {$5$};
  \node at (6,-2) {$1$}; \node[malicious] at (6,.7) (a4) {$6$};
    \node[honest] at (-1,0) (base) {$0$};
  \draw[thick,->] (base) -- (ab1);
  \draw[thick,->] (ab1) -- (ab2);
  \draw[thick,->] (ab2) to[bend left=10] (a3);
    \draw[thick,->] (a3) -- (a4);
  \draw[thick,->] (ab2) to[bend right=10] (b3);
    \draw[thick,->] (b3) -- (b4);
    \end{tikzpicture} 
  \caption{An $x$-balanced fork, where $x=00$}
  \label{fig:x-balanced}
\end{figure}

\paragraph{Balanced forks and settlement time.}
A fundamental question arising in typical blockchain settings is how
to determine \emph{settlement time}, the delay after which the
contents of a particular block of a blockchain can be considered
stable. The existence of a balanced fork is a precise indicator for
``settlement violations'' in this sense. Specifically, consider a
characteristic string $xy$ and a transaction appearing in a block
associated with the first slot of $y$ (that is, slot $|x| + 1$). One
clear violation of settlement at this point of the execution is the
existence of two chains---each of maximum length---which diverge
\emph{prior to $y$}; in particular, this indicates that there is an
$x$-balanced fork $F$ for $xy$. Let us record this observation below.

\begin{observation}\label{obs:settlement-balanced-fork}
  Let $s, k \in \NN$ be given and 
  let $w$ be a characteristic string. 
  Slot $s$ is not $k$-settled for the characteristic string $w$ 
  if 
  there exist a decomposition $w = xyz$, 
  where $|x| = s - 1$ and $|y| \geq k+1$, 
  and an $x$-balanced fork for $xy$. 
\end{observation}

In fact, every $\kSlotCP$ violation produces a balanced fork as well;
see Theorem~\ref{thm:cp-fork} in Section~\ref{sec:cp-forks}.  In
particular, to provide a rigorous $k$-slot settlement
guarantee---which is to say that the transaction can be considered
settled once $k$ slots have gone by---it suffices to show that with
overwhelming probability in choice of the characteristic string
determined by the leader election process (of a full execution of the
protocol), no such forks are possible. Specifically, if the protocol
runs for a total of $T$ time steps yielding the characteristics string
$w = xy$ (where $w \in \{0,1\}^T$ and the transaction of interest
appears in slot $|x| + 1$ as above) then it suffices to ensure that
there is no $x$-balanced fork for $x\hat{y}$, where $\hat{y}$ is an
arbitrary prefix of $y$ of length at least $k + 1$; see
Corollary~\ref{cor:main} in Section~\ref{sec:estimates}.  Note that
for systems adopting the longest chain rule, this condition must
necessarily involve the \emph{entire future dynamics} of the
blockchain. We remark that our analysis below will in fact let us take
$T = \infty$.

\begin{definition}[Closed fork]
A fork $F$ is \emph{closed} if every leaf is honest. For convenience, we say the trivial fork is closed.
\end{definition}

Closed forks have two nice properties that make them especially useful in reasoning about the view of honest parties.
First, a closed fork must have a unique longest tine (since honest parties are aware of all previous honest blocks, and honest
parties observe the longest chain rule). Second, recalling our description of the
settlement game, closed forks intuitively capture decision points for the adversary.
The adversary can potentially show many tines to many honest parties, but once an honest node has been placed on top of 
a tine, any adversarial blocks beneath it are part of the public record and are visible to all honest parties. For these
reasons, we will often find it easier to reason about closed forks than arbitrary forks. % (without loss of generality).

The next few definitions are the start of a general toolkit for reasoning about an adversary's capacity to build highly diverging paths in forks, based on the underlying characteristic string.
%current state of a fork.

%%%%Reach and margin
\begin{definition}[Gap, reserve, and reach]\label{def:gap-reserve-reach}
For a closed fork $F \vdash w$ and its unique longest tine $\hat{t}$, we define the \emph{gap} of a tine $t$ to be $\gap(t)=\length(\hat{t})-\length(t)$.
Furthermore, we define the \emph{reserve} of $t$, denoted $\reserve(t)$, to be the number of adversarial indices in $w$ that appear after the terminating vertex of $t$. More precisely, if $v$ is the last vertex of $t$, then
\[
  \reserve(t)=|\{\ i \mid w_i=1 \ and \ i > \ell(v)\}|\,.
  \]
These quantities together define the \emph{reach} of a tine: $
\reach(t)=\reserve(t)-\gap(t)$.
\end{definition}

The notion of reach can be intuitively understood as a measurement of
the resources available to our adversary in the settlement
game. Reserve tracks the number of slots in which the adversary has
the right to issue new blocks.  When reserve exceeds gap (or
equivalently, when reach is nonnegative), such a tine could be
extended---using a sequence of dishonest blocks---until it is as long
as the longest tine. Such a tine could be offered to an honest player
who would prefer it over, e.g., the current longest tine in the
fork. In contrast, a tine with negative reach is too far behind to be
directly useful to the adversary at that time.

\begin{definition}[Maximum reach]
For a closed fork $F\vdash w$, we define $\rho(F)$ to be the largest reach attained by any tine of $F$, i.e., 
\[
\rho(F)=\underset{t}\max \ \reach(t)\,.
\]
Note that $\rho(F)$ is never negative (as the longest tine of any fork always has reach at least 0). We overload this notation to denote the maximum reach over all forks for a given characteristic string: 
\[
\rho(w)=\underset{\substack{F\vdash w\\\text{$F$ closed}}}\max\big[\underset{t}\max \ \reach(t)\big]\,.
\]
\end{definition}

\begin{definition}[Margin]\label{def:margin}
The \emph{margin} of a fork $F\vdash w$, denoted $\mu(F)$, is defined as 
\begin{equation}\label{eq:margin-absolute}
\mu(F)=\underset{t_1\nsim t_2}\max \bigl(\min\{\reach(t_1),\reach(t_2)\}\bigr)\,,
\end{equation}
where this maximum is extended over all pairs of disjoint tines of
$F$; thus margin reflects the ``second best'' reach obtained over all
disjoint tines. In order to study splits in the chain over particular portions of a
string, we generalize this to define a ``relative'' notion of margin:
If $w = xy$ for two strings $x$ and $y$ and, as above, $F \vdash w$,
we define
\[
  \mu_x(F)=\underset{t_1\nsim_x t_2}\max \bigl(\min\{\reach(t_1),\reach(t_2)\}\bigr)\,.
\]
Note that $\mu_\varepsilon(F) = \mu(F)$.

For convenience, we once again overload this notation to denote the
margin of a string. $\mu(w)$ refers to the maximum value of $\mu(F)$
over all possible closed forks $F$ for a characteristic string $w$:
\[
\mu(w)=\underset{\substack{F\vdash w,\\ \text{$F$ closed}}}\max \, \mu(F)\,.
\]
Likewise, if $w = xy$ for two strings $x$ and $y$ we define
\[
\mu_x(y)=\underset{\substack{F\vdash w,\\ \text{$F$ closed}}} \max \, \mu_x(F)\,.
\]
%(Cf.~\cite{KRDO17}, which defined and studied the ``absolute'' version
%$\mu(\cdot)$ of this quantity of~\eqref{eq:margin-absolute}.)
\end{definition}
Note that, at least informally, ``second-best'' tines are of natural
interest to an adversary intent on the construction of an $x$-balanced fork, 
which involves two (partially disjoint) long tines.

\paragraph{Balanced forks and relative margin.}
\citet{KRDO17} showed that a balanced fork can be constructed for a
given characteristic string $w$ if and only if there exists some
closed $F\vdash w$ such that $\mu(F)\geq 0$.  We record a relative
version of this theorem below, which will ultimately allow us to
extend the analysis of \cite{KRDO17} to more general class of
disagreement and settlement failures.

\begin{fact}\label{fact:margin-balance}
  Let $xy \in \{0,1\}^n$ be a characteristic string. Then there is an
  $x$-balanced fork $F \vdash xy$ if and only if $\mu_x(y) \geq 0$.
\end{fact}

\begin{proof}
  The proof is immediate from the definitions. We sketch the details for completeness.
  
  Suppose $F$ is an $x$-balanced fork for $xy$. Then $F$ must contain a pair of tines $t_1$ and $t_2$ for which
  $t_1 \not\sim_x t_2$ and $\length(t_1) = \length(t_2) = \height(F)$. We observe that (1) $\gap(t_i)=0$ for both $t_1$ and $t_2$, and (2) reserve is always a nonnegative quantity. Together with the definition of $\reach$, these two facts immediately imply $\reach(t_i) \geq 0$. Because $t_1$ and $t_2$ are edge-disjoint over $y$ and $\min\{\reach(t_1),\reach(t_2)\}\geq0,$ we conclude that $\mu_x(y)\geq 0$, as desired. 

  Suppose $\mu_x(y)\geq 0$. Then there is some closed fork $F$ for $xy$ such that $\mu_x(F)\geq0$. By the definition of
   relative margin, we know that $F$ has two tines $t_1$, $t_2$ such that $t_1\nsim_x t_2$ and 
  $\reach(t_i)\geq0$. Recall that we define reach by $\reach(t)=\reserve(t)-\gap(t)$, and so in this case 
 it follows that $\reserve(t_i) - \gap(t_i)\geq0$. Thus, an $x$-balanced fork $F'\vdash xy$ can be constructed from $F$ by 
 appending a path of $\gap(t_i)$ adversarial vertices to each $t_i$.
\end{proof}

As indicated above, we can define the ``forkability'' of a
characteristic string in terms of its margin.
\begin{definition}[Forkable strings]\label{def:forkable}
  A charactersitic string $w$ is \emph{forkable} if its margin is non-negative, i.e., $\mu(w) \geq 0$.
  Equivalently, $w$ is forkable if there is a balanced fork for $w$.
\end{definition}
Although this definition is not necessary for our presentation, it
reflects the terminology of existing literature.

%%% Local Variables:
%%% mode: latex
%%% TeX-master: "main"
%%% End:

\section{Common prefix violation and balanced forks}
\label{sec:cp-forks}
In this section, we show that a common prefix violation implies the
existence of a balanced fork. This allows us to bound consistency
errors by reasoning about balanced forks.  In particular,
inequality~\eqref{eq:pr-cp-fork} is a direct consequence of the
theorem below.

\begin{theorem}\label{thm:cp-fork}
  Let $k, T \in \NN$.  
  Let $w \in \{0,1\}^T$ be a characteristic string which violates $\kSlotCP$. 
  Then 
  there exist a decomposition $w = xyz$ and a fork $\hat{F} \Fork xy$, 
  where $|y| \geq k + 1$, 
  so that $\hat{F}$ is $x$-balanced. 
  %
  % Then $w$ can be written $w = xyz$, where 
  % % $x,y, z \in \{0,1\}^*$, 
  % $|y| \geq k + 1$, so that 
  % there is an $x$-balanced fork $\hat{F} \Fork xy$. 
\end{theorem}

\newcommand{\Final}[1]{\tilde{#1}}

% \subsection{Proof of Theorem~\ref{thm:cp-fork}}
\begin{proof}
  
  % \newcommand{\Base}{x}

  % The proof closely follows the argument of \cite[Theorem 4.26]{KRDO17}.  
  Recall that $\ell(t)$ is the slot index of the last vertex of tine
  $t$.  Define $A \triangleq \bigcup_{F \Fork w} A_F$ where, for a
  given fork $F \Fork w$, define
  \[
    A_F \triangleq \left\{
      (\tau_1, \tau_2) \SuchThat \parbox{60mm}{       
      $\tau_1, \tau_2$ are two viable tines in the fork $F$, 
      $\ell(\tau_1) \leq \ell(\tau_2)$, and 
      the pair $(\tau_1, \tau_2)$ is a witness to a $\kSlotCP$ violation}
     \right\}
     \,.
  \]
  % Consider the set of tine-pairs $(\tau_1, \tau_2)$ from all forks on $w$, 
  % such that $\ell(\tau_1) \leq \ell(\tau_2)$ and the pair witnesses a $\kSlotCP$ violation. 
  % Let $T$ be the set of all such tine-pairs.
  % \begin{definition}[Slot divergence]    
  % \end{definition}
  Define the \emph{slot divergence} of two tines as 
  $\SlotDivergence(\tau_1, \tau_2) \defeq \ell(\tau_1) - \ell(\tau_1 \Intersect \tau_2)$ 
  where $\tau_1 \Intersect \tau_2$ denotes the common prefix 
  of the tines $\tau_1$ and $\tau_2$. 
  Recalling the definition of a $\kSlotCP$ violation, it is clear that 
  \begin{equation}\label{eq:divergence}
      \SlotDivergence(\tau_1, \tau_2) \geq k + 1 \quad \text{for all } (\tau_1, \tau_2) \in A
      \,.
  \end{equation}
  % For concreteness and simplicity, we assume that the nodes in $F$ 
  % have been suitably labeled so that the following two conditions are met:
  Notice that there must be a tine-pair $(t_1, t_2) \in A$ which satisfies the following two conditions: 
    \begin{equation}\label{eq:tines}
      \SlotDivergence(t_1, t_2) 
      % = \SlotDivergence(w) 
      % = \max_{(\tau_1, \tau_2) \in T} \SlotDivergence(\tau_1, \tau_2) 
      \text{ is maximal over $A$}
      \,, \text{and}
    \end{equation}
  % and
  \begin{equation}\label{eq:minimality}
    | \ell(t_2) - \ell(t_1) | 
    \text{ is minimal among all tine-pairs in $A$ 
      for which~\eqref{eq:tines} holds.}
  \end{equation}
  The tines $t_1, t_2$ will play a special role in our proof; 
  let $F$ be a fork containing these tines. 

  \paragraph{The prefix $x$, fork $F_x$, and vertex $u$.} 
  Let $u$ denote the last vertex on the tine
  $t_1 \cap t_2$, as shown in the diagram below, and let
  $\alpha \triangleq \ell(u) = \ell(t_1 \cap t_2)$. 
  Let $x \triangleq w_1, \ldots, w_\alpha$ 
  and let $F_x$ be the fork-prefix of $F$ supported on $x$. 
  We will argue that $u$ must be honest and, in addition, that 
  $F_x$ must contain a unique longest tine $t_u$ terminating 
  at the vertex $u$. 
  We will also identify a substring $y, |y| \geq k + 1$ 
  such that $w$ can be written as $w = xyz$. 
  Then we will construct a balanced fork $\tilde{F}_y \Fork y$ by 
  modifying the subgraph of $F$ supported on $y$. 
  We will finish the proof by constructing an $x$-balanced fork by 
  suitably appending $\tilde{F}_y$ to $F_x$.
  % and then appealing to Fact~\ref{fact:margin-balance}.
    
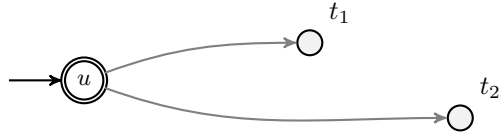
  \begin{center}
      \begin{tikzpicture}[>=stealth', auto, semithick,
        unknown/.style={circle,draw=black,thick,font=\small},
        honest/.style={circle,draw=black,thick,double,font=\small},
        malicious/.style={fill=gray!10,circle,draw=black,thick,font=\small}]
        \node[honest] at (0,0) (u) {$u$};
        \node[malicious] at (3,.5)  (z1) {};
        \node[malicious] at (5,-.5)   (z2) {};
        \path (z1) ++(.4,.4) node {$t_1$};
        \path (z2) ++(.4,.4) node {$t_2$};
        \draw[thick,<-] (u) to (-1,0);
        \draw[thick,<-,gray] (z1) to[out=180,in=20] (u);
        \draw[thick,<-,gray] (z2) to[out=180,in=-20] (u);
      \end{tikzpicture}
    \end{center}
  %  Let $\beta$ denote the smallest honest index of $w$ for which
  %  $\beta \geq \ell(t_2) = \max(\ell(t_1), \ell(t_2))$, with the convention that
  %  $\beta = n+1$ if there is no such index.

    \paragraph{$u$ must be an honest vertex.}
    We observe, first of all, that the vertex $u$ cannot be adversarial:
    otherwise it is easy to construct an alternative fork
    $F^\prime \Fork w$ and a pair of tines in $F^\prime$ that violate~\eqref{eq:tines}. 
    Specifically, construct $F^\prime$ from $F$ by
    adding a new (adversarial) vertex $u^\prime$ to $F$ for which
    $\ell(u^\prime) = \ell(u)$, adding an edge to $u^\prime$ from the
    vertex preceding $u$, and replacing the edge of $t_1$ following $u$
    with one from $u^\prime$; then the other relevant properties of the
    fork are maintained, but the slot divergence of the resulting tines has
    increased by at least one. (See the diagram below.)
    
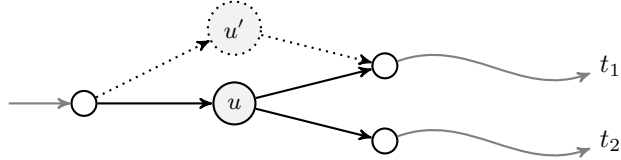
\begin{center}
      \begin{tikzpicture}[>=stealth', auto, semithick,
        unknown/.style={circle,draw=black,thick,font=\small},
        honest/.style={circle,draw=black,thick,double,font=\small},
        malicious/.style={fill=gray!10,circle,draw=black,thick,font=\small}]
        \node[malicious] at (2,0) (v) {$u$};
        \node[malicious,dotted] at (2,1) (u) {$u^\prime$};
        \node[unknown] at (4,-.5)  (b1) {};
        \node[unknown] at (4,.5)  (a1) {};
        \node[unknown] at (0,0) (base) {};
        \node at (7,.5) (t1) {$t_1$};
        \node at (7,-.5) (t2) {$t_2$};
        % \node[state,honest] at (3,-1) (bottom) {};
        % \node[state,honest] at (7,1) (top) {$H$};
        \draw[thick,->] (base) -- (v);
        \draw[thick,->] (v) -- (a1);
        \draw[thick,->] (v) -- (b1);
        \draw[thick,->,dotted] (u) -- (a1);
        \draw[thick,->,dotted] (base) -- (u);
        \draw[thick,<-,gray] (t1) to[in=20,out=200] (a1);
        \draw[thick,<-,gray] (t2) to[in=20,out=200] (b1);
        \draw[thick,<-,gray] (base) to (-1,0);
        % \draw[thick,<->] (3,0) -- (7,0) node[pos=.5] {$\gap(f)$};
      \end{tikzpicture}
    \end{center}
    
    \paragraph{$F_x$ has a unique, longest (and honest) tine $t_u$.}
    A similar argument implies that the fork
    $F_x$ has a unique vertex of depth $\depth(u)$: namely, $u$ itself. In
    the presence of another vertex $u^\prime$ (of $F_x$) with depth
    $\depth(u)$, ``redirecting'' $t_1$ through $u^\prime$ (as in the
    argument above) would likewise result in a fork with 
    a larger slot divergence. 
    To see this, notice that $\ell(u^\prime)$ must be strictly less than $\ell(u)$ 
    since $\ell(u)$ is an honest slot (which means $u$ is the only vertex at that slot).
    Thus $\ell(\cdot)$ would indeed be increasing along
    this new tine (resulting from redirecting $t_1$).
    As $\alpha$ is the last index of the string $x$, this additionally
    implies that $F_x$ has no vertices of depth exceeding $\depth(u)$. 
    Let $t_u \in F_x$ be the tine with $\ell(t_u) = \alpha$. 
    \begin{equation}\label{eq:tu}
        \text{The honest tine $t_u$ is the unique longest tine in $F_x$}
        \,.
    \end{equation}

    % Without loss of generality we may assume that $\ell(z_2)$, the
    % honest index labeling $z_2$, is in fact the first honest index in
    % $w$ appearing after $\ell(z_1)$. To justify this, let $\beta$ denote
    % this first honest index of $w$ after $\ell(z_1)$ and let $x$ denote
    % the unique vertex of $F$ for which $\ell(x) = \beta$. Note that
    % $\hdepth(x) > \hdepth(z_1)$. If the tine $t$ ending at $x$ shares an
    % edge with $t_1$ after $u$, then $t$ is disjoint from $t_2$ after $u$
    % and it follows that $\divergence(t,t_2) > \divergence(t_1,t_2)$, a
    % contradiction. Thus $t$ shares no edges with $t_1$ after $u$, and
    % $\length(t) > \length(t_1)$; it follows that
    % $\divergence(t_1,t) = \divergence(t_1,t_2)$ and we may assume
    % $t_2 = t$ in the remainder of the argument.

    \paragraph{Identifying $y$.}
    Let $\beta$ denote the smallest honest index of $w$ for which
    $\beta \geq \ell(t_2)$, with the convention that if there is no such
    index we define $\beta = T + 1$. 
    Observe that $\beta - 1 \geq \ell(t_1)$. 
    (If $\ell(t_2)$ is an honest slot then $\beta = \ell(t_2)$ 
    but $\ell(t_1) < \ell(t_2)$. 
    The case $\ell(t_1) = \ell(t_2)$ is possible if $\ell(t_2)$ is an adversarial slot; 
    but then $\beta > \ell(t_2)$.)
    These indices, $\alpha$ and $\beta$, distinguish the
    substrings $y = w_{\alpha+1} \ldots w_{\beta-1}$ and 
    $z = w_{\beta} \ldots w_T$; 
    we will focus on $y$ in the remainder of the proof. 
    Since the function
    $\ell(\cdot)$ is strictly increasing along any tine, observe that
    \begin{equation*}
        |y| 
        = \beta - \alpha - 1 
        \geq \ell(t_1) - \ell(u) 
        \geq k + 1
        \,.
    \end{equation*}
    Hence $y$ has the desired length and it suffices to establish that it is forkable.
    We can extract from $F$ a balanced fork (for $y$) in
    two steps: (i.) we subject the fork $F$ to some minor
    restructuring to ensure that all ``long'' tines pass through $u$;
    (ii.) we construct a flat fork by treating the vertex $u$ as the
    root of a portion of the subtree of $F$ labeled with the indices of
    $y$. At the conclusion of the construction, the segments of the
    two tines $t_1$ and $t_2$ will yield the required ``long, disjoint, equal-length''
    tines satisfying the definition of a balanced fork.

    \paragraph{Honest indices in $xy$ have low depths.}
    The minimality assumption~\eqref{eq:minimality} implies that any honest
    index $h$ for which $h < \beta$ has depth no more than
    $\min(\length(t_1),\length(t_2))$: specifically,
    \begin{equation}\label{eq:honest-depth}
      h < \beta \quad\Longrightarrow \quad \hdepth(h) \leq \min(\length(t_1), \length(t_2))\,.
    \end{equation}
    To see this, consider an honest index $h,h < \beta$ and a tine $t_h$
    for which $\ell(t_h) = h$. Recall that $t_1$ and $t_2$ are viable and 
    that $h < \ell(t_2)$. (If $\ell(t_2)$ is honest, it is obvious. 
    Otherwise, $h < \ell(t_2) < \beta$ since $\ell(t_2)$ is adversarial.) 
    As $t_2$ is viable, it follows immediately that
    $\hdepth(h)  = \length(t_h) \leq \length(t_2)$. 
    Similarly, if $h \leq \ell(t_1)$
    then $\hdepth(h) \leq \length(t_1)$ since $t_1$ is viable as well. 
    The remaining case, i.e., when $\ell(t_1) < h < \ell(t_2)$, can be ruled out 
    by the argument below.

    \paragraph{There is no honest index between $\ell(t_1)$ and $\ell(t_2)$.}
    We claim that 
    \begin{equation}\label{eq:no-honest-index}
        \text{There is no honest index $h$ satisfying $\ell(t_1) < h < \ell(t_2)$}
        \,.
    \end{equation}
    The claim above is trivially true if $\ell(t_1) = \ell(t_2)$.
    Otherwise, suppose (toward a contradiction) 
    that $h$ is an honest index satisfying $\ell(t_1) < h < \ell(t_2)$. 
    Let $t_h$ be the (honest) tine at slot $h$. 
    The tine-pair $(t_1, t_h)$ may or may not be in $A$. 
    We will show that both cases lead to contradictions.
    \begin{itemize}
      \item If $(t_1, t_h)$ is in $A$ and $\ell(t_1 \Intersect t_h) \leq \ell(u)$, 
      $\SlotDivergence(t_1, t_h)$ is at least $\SlotDivergence(t_1, t_2)$. 
      In fact, due to~\eqref{eq:tines}, this inequality must be an equality. 
      However, the assumption $\ell(t_1) < h < \ell(t_2)$ contradicts~\eqref{eq:minimality}. 

      \item If $(t_1, t_h)$ is in $A$ and $\ell(t_1 \Intersect t_h) > \ell(u)$, 
      it follows that $\SlotDivergence(t_h, t_2) > \SlotDivergence(t_1, t_2)$. 
      As the latter quantity is at least $k + 1$, $(t_h, t_2)$ must be in $A$. 
      The preceding inequality, however, contradicts~\eqref{eq:tines}.

      \item If $(t_1, t_h) \not \in A$, 
      $\SlotDivergence(t_1, t_h)$ is at most $k$.
      As $\SlotDivergence(t_1, t_2)$ is at least $k + 1$, 
      % it follows that $\ell(t_1) - \ell(t_1 \Intersect t_h) > \ell(u)$.
      $t_h$ and $t_1$ must share a vertex after slot $\ell(u)$. 
      Since $\ell(t_1) < h < \ell(t_2)$ by assumption, 
      $\SlotDivergence(t_h, t_2) > \SlotDivergence(t_1, t_2) \geq k + 1$ 
      and, as a result, $(t_h, t_2) \in A$. 
      However, the preceding strict inequality violates condition~\eqref{eq:tines}. 
    \end{itemize}

  \paragraph{A fork $\pinch{u}{F}$ where all long tines go through $u$.}
    In light of the remarks above, we observe that the fork $F$ may be
    ``pinched'' at $u$ to yield an essentially identical fork
    $\pinch{u}{F} \vdash w$ with the exception that all tines of length
    exceeding $\depth(u)$ pass through the vertex $u$. Specifically, the
    fork $\pinch{u}{F} \vdash w$ is defined to be the graph obtained
    from $F$ by changing every edge of $F$ directed towards a vertex of
    depth $\depth(u) + 1$ so that it originates from $u$. To see that
    the resulting tree is a well-defined fork, it suffices to check that
    $\ell(\cdot)$ is still increasing along all tines of
    $\pinch{u}{F}$. For this purpose, consider the effect of this
    pinching on an individual tine $t$ terminating at a particular
    vertex $v$---it is replaced with a tine $\pinch{u}{t}$ defined so
    that:
    \begin{itemize}
    \item If $\length(t) \leq \depth(u)$, the tine $t$ is unchanged:
      $\pinch{u}{t} = t$.
    \item Otherwise, $\length(t) > \depth(u)$ and $t$ has a vertex $v$
      of depth $\depth(u) + 1$; note that $\ell(v) > \ell(u)$ because
      $F_x$ contains no vertices of depth exceeding $\depth(u)$. Then
      $\pinch{u}{t}$ is defined to be the path given by the tine
      terminating at $u$, a (new) edge from $u$ to $v$, and the suffix
      of $t$ beginning at $z$. (As $\ell(v) > \ell(u)$ this has the
      increasing label property.)
    \end{itemize}
    Thus the tree $\pinch{u}{F}$ is a legal fork on the same vertex set;
    note that the depths of vertices in $F$ and $\pinch{u}{F}$ are
    identical.
    
    \paragraph{Constructing a shallow fork $F_y \Fork y$.}
    By excising the tree rooted at $u$ from this pinched fork
    $\pinch{u}{F}$, we may extract a fork for the string
    $w_{\alpha+1} \dots w_T$. Specifically, consider the induced
    subgraph $\cut{u}{F}$ of $\pinch{u}{F}$ given by the vertices
    $\{u\} \cup \{ v \mid \depth(v) > \depth(u)\}$. By treating $u$ as a
    root vertex and suitably defining the labels $\cut{u}{\ell}$ of
    $\cut{u}{F}$ so that $\cut{u}{\ell}(v) = \ell(v) - \ell(u)$, this
    subgraph has the defining properties of a fork for
    $w_{\alpha+1} \ldots w_T$. In particular, considering that
    $\alpha$ is honest it follows that each honest index $h > \alpha$
    has depth $\hdepth(h) > \length(u)$ and hence $h$ labels a vertex in
    $\cut{u}{F}$.  For a tine $t$ of $\pinch{u}{F}$, we let $\cut{u}{t}$
    denote the suffix of this tine beginning at $u$, which forms a tine
    in $\cut{u}{F}$. (If $\length(t) \leq \depth(u)$, we define
    $\cut{u}{t}$ to consist solely of the vertex $u$.)  Note that
    $\cut{u}{t_1}$ and $\cut{u}{t_2}$ share no edges in the fork
    $\cut{u}{F}$.
    
    Finally, let $F_y$ denote the subtree obtained from $\cut{u}{F}$
    as the union of all tines $\cut{u}{t}$ of $\cut{u}{F}$ so that all
    labels of $\cut{u}{t}$ are drawn from $y$ (as it appears as a prefix
    of $w_{\alpha+1} \ldots w_T$), and
    \begin{equation}\label{eq:tines-Fy}
      \length(\cut{u}{t}) \leq \max_{\substack{h \leq |y|\\ \text{$h$ honest} } } \hdepth(h)
      \,.
    \end{equation}
    It is immediate that $F_y \vdash y$. 
    
    \paragraph{Two longest viable tines in $F_y$.}
    Consider the tines $\cut{u}{t_1}$ and $\cut{u}{t_2}$. As mentioned
    above, they share no edges in $\cut{u}{F}$ and hence the prefixes
    $\check{t_1}$ and $\check{t_2}$ (of $\cut{u}{t_1}$ and
    $\cut{u}{t_2}$) appearing in $F_y$ share no edges. 
    % By~\eqref{eq:tines-Fy}, 
    % the lengths of $\check{t}_1, \check{t}_2 \in F_y$ are 
    % at most $\hdepth(h)$ where $h$ is the largest honest index in $y$.
    We wish to
    show that these prefixes have the maximal length in $F_y$, making $F_y$ balanced, as desired. 
    Let $h$ be the largest honest index in $y$. 
    Since the lengths of the tines in $F_y$ 
    are at most $\hdepth(h)$, 
    it suffices to show that the lengths of 
    $\check{t}_i, i \in \{1,2\}$ is at least $\hdepth(h)$. 

    This is immediate for the tine
    $\check{t}_1$ since all labels of $\cut{u}{t_1}$ are drawn from
    $y$ and, considering~\eqref{eq:honest-depth}, its depth is
    at least that of all relevant honest vertices. 
    As for $\check{t_2}$,
    observe that if $\ell(t_2)$ is not honest then $\beta > \ell(t_2)$
    so that, as with $\check{t}_1$, the tine $\check{t}_2$ is labeled by
    $y$ so that the same argument, relying
    on~\eqref{eq:honest-depth}, ensures that the $\length(\check{t}_2)$ 
    is at least 
    the depth of all relevant honest vertices. 
    If $\ell(t_2)$ is
    honest, $\beta = \ell(t_2)$, and the terminal vertex of
    $\cut{u}{t_2}$ does not appear in $F_y$ (as $\ell(\cut{u}{t_2})$ falls outside 
    $y$). In this case, however,
    $\length(\cut{u}{t_2}) > \hdepth(h)$ for any honest index $h$ of
    $y$. 
    It follows that
    $\length(\check{t_2})$, which equals $\length(\cut{u}{t_2}) - 1$, 
    is at least the
    depth of any honest index of $y$, as desired. 
    Thus we have proved
    \begin{equation}\label{eq:two-long-tines}
        \text{$\check{t}_1$ and $\check{t}_2$ are 
        two maximally long viable tines in $F_y \Fork y$}
        \,.
    \end{equation}

    \paragraph{Constructing a flat fork $\tilde{F}_y \Fork y$.}    
    Let us identify the fork prefix $\tilde{F}_y \ForkPrefix F_y$ which 
    is either identical to $F_y$ or differs from $F_y$ 
    in only one of the tines $\check{t}_1, \check{t}_2$. 
    In particular, if $\length(\check{t}_1) = \length(\check{t}_2)$, we set $\tilde{F}_y = F_y$. 
    Otherwise, let $\check{t}_a$ be the longer of the two tines $\check{t}_1, \check{t}_2$; 
    let $\check{t}_b$ be the shorter one. 
    We modify $F_y$ by deleting some trailing adversarial nodes from $\check{t}_a$ 
    until it has the same length as $\check{t}_b$; 
    we set $\tilde{F}_y$ as the resulting fork 
    and, in addition, 
    set $\tilde{t}_b = \check{t}_b$ and 
    $\tilde{t}_a$ as the tine after trimming $\check{t}_a$. 
    
    We claim that $\tilde{F}_y$ is balanced. 
    The claim is obvious if $\length(\check{t}_1) = \length(\check{t}_2)$.
    Otherwise, thanks to~\eqref{eq:two-long-tines}, 
    it remains to show that the longer tine, $\check{t}_a$, 
    has sufficiently many trailing adversarial nodes which, 
    if deleted, yields $\length(\Final{t}_1) = \length(\Final{t}_2)$. 
    To that end, let $h_i$ be the index of the last honest vertex 
    on $\check{t}_i \in F_y, i \in \{1,2\}$. 
  %   Since $t_1, t_2$ were viable tines in $F$, it follows that $\length(\check{t_1}) \geq \hdepth(h_2)$ and 
  %   $\length(\check{t_2}) \geq \hdepth(h_1)$.
    
    Suppose $\length(\check{t}_2) > \length(\check{t}_1)$. 
    By~\eqref{eq:no-honest-index}, we also have $\length(\check{t}_1) \geq \hdepth(h_2)$  
    and hence we can trim some of the trailing adversarial nodes from $\check{t}_2$ 
    to get the tine $\Final{t}_2$ 
    whose length is the same as that of $\check{t}_1$. 
    Otherwise, suppose $\length(\check{t}_1) > \length(\check{t}_2)$. 
    Since $t_2$ is a viable tine in $F$, we also have $\length(\check{t}_2) \geq \hdepth(h_1)$. 
    Thus we can trim some of the trailing adversarial nodes from $\check{t}_1$
    to have a tine $\Final{t}_1$ 
    whose length is the same as that of $\check{t}_2$. 
    In any case, the quantity $\min(\length(\Final{t}_1), \length(\Final{t}_2))$ 
    remains the same as $\min(\length(\check{t}_1), \length(\check{t}_2))$. 
    Thus the fork $\tilde{F}_y$ has at least two tines, $\Final{t}_1$ and $\Final{t}_2$, that achieve the maximum length of all tines in $\tilde{F}_y$; hence $\tilde{F}_y$ is balanced.

    \paragraph{An $x$-balanced fork $\hat{F} \ForkPrefix F$.} 
    Let us identify the root of the fork $\tilde{F}_y$ with the vertex $u$ of $F_x$ and 
    let $\hat{F}$ be the resulting graph (after ``gluing'' the root of $\tilde{F}_y$ to $u$). 
    By~\eqref{eq:tu}, it is easy to see that the fork 
    $\hat{F} \ForkPrefix F$ 
    is indeed a valid fork on the string $x y$. 
    Moreover, $\hat{F}$ is $x$-balanced since $\tilde{F}_y$ is balanced. 
    The claim in Theorem~\ref{thm:cp-fork} follows immediately since $|y| \geq k + 1$.
  \end{proof}

  % \end{proof}
%%% Local Variables:
%%% mode: latex
%%% TeX-master: "main"
%%% End:

\section{A simple recursive formulation of relative margin}
\label{sec:recursion}
A significant finding of~\citet{KRDO17} is that the margin of a
characteristic string $\mu(w)$---the maximum value of a quantity taken
over a (typically) exponentially-large family of forks---can be given
a simple, mutually recursive formulation with the associated quantity
of reach $\rho(w)$. Specifically, they prove the following lemma.

%%% original lemma
\begin{lemma}[{\cite[Lemma~4.19]{KRDO17}}]\label{lem:margin} 
  $\rho(\varepsilon) = 0$ 
  where $\varepsilon$ is the empty string, and, for all nonempty strings $w\in\{0,1\}^*$,
  \begin{equation}
    \rho(w1) = \rho(w)+1\,, \qquad\text{and}\qquad
    \rho(w0) = \begin{cases} 0 & \text{if $\rho(w) = 0$,}\\
      \rho(w)-1 & \text{otherwise.}
    \end{cases}
		\label{eq:rho-recursive}
  \end{equation}
  Furthermore, margin satisfies the mutually recursive relationship
  $\mu(\varepsilon) = 0$ and for all $w \in \{0,1\}^*$,
  \begin{equation}
    \mu(w1) = \mu(w)+1\,,\qquad\text{and}\qquad
    \mu(w0) = \begin{cases}
      0 & \text{if $\rho(w)>\mu(w)=0$,} \\
%      \mu(w)-1 & \text{if $\rho(w)=0$,} \\
      \mu(w)-1 & \text{otherwise.}
    \end{cases}
		\label{eq:mu-recursive}
  \end{equation}
  Additionally, there exists a closed fork $F\vdash w$ such that
  $\rho(F)=\rho(w)$ and $\mu(F)=\mu(w)$.
  %(It is convenient to separate the case $\rho(w) = 0$ from the other case which also yields $\mu(w) - 1$ in the proof, so we reflect that in the statement of the theorem.)
\end{lemma}

We prove an analogous recursive statement for relative margin, recorded below.

\begin{lemma}[Relative margin]\label{lem:relative-margin}
  Given a fixed string $x\in\{0,1\}\text{\emph{*}}$,
  $\mu_x(\varepsilon) =\rho(x)$ 
  where $\varepsilon$ is the empty string, and, for all nonempty strings $w=xy\in\{0,1\}\text{\emph{*}},$
  \begin{equation}
    \mu_x(y1)= \mu_x(y)+1\,,\qquad\text{and}\qquad
    \mu_x(y0)= \begin{cases}
      0 & \text{if } \rho(xy) > \mu_x(y)=0\,, \\
%      \mu_x(y)-1 &  \text{if } \rho(xy)=0\,, \\
      \mu_x(y)-1 & \text{otherwise.}
    \end{cases}
		\label{eq:mu-relative-recursive}
  \end{equation}
  Additionally, there exists a closed fork $F\vdash xy$ such that
  $\rho(F)=\rho(xy)$ and $\mu_x(F)=\mu_x(y)$.
  %(It is convenient to
  %separate the case $\rho(w) = 0$ from the other case which also
  %yields $\mu(w) - 1$ in the proof, so we reflect that in the
  %statement of the lemma.)
\end{lemma}

We delay the proof of Lemma~\ref{lem:relative-margin} to
Section~\ref{sec:margin-proof}, preferring to immediately focus on the
application to settlement times in Section~\ref{sec:estimates}.

\paragraph{Discussion.} The proof of Lemma~\ref{lem:relative-margin}
shares many technical similarities with the proof of
Lemma~\ref{lem:margin} given by~\citet{KRDO17}. However, there is an
important respect in which the proofs differ. Each of the proofs
requires the definition of a particular adversary (which, in effect,
constructs a fork achieving the worst case reach and margin guaranteed
by the lemma). The adversary constructed by~\cite{KRDO17} can create a
balanced fork for $w$ whenever $\mu(w) \geq 0$ (i.e., $w$ is
``forkable''). However, the adversary only focuses on the problem of
producing disjoint tines over the \emph{entire string} $w$ (consistent
with the definition of $\mu(\cdot)$). The ``optimal online adversary,''
developed in Section~\ref{sec:canonical-forks},
% developed during the proof of Lemma~\ref{lem:relative-margin},
%, in
%contrast,
uses a more sophisticated rule for extending chains (tines) of the
fork. 
Notably, this adversary can \emph{simultaneously maximize relative margin
  over all prefixes of the string}. 

% This adversary is a purist: if $w$ is not forkable, he gives up. As an
% illuminating example, consider the characteristic string
% $w=0^n0^k1^k0^k$, for $n\gg k$. The purist adversary will produce a
% single tine containing only the honest nodes, but we can imagine that
% a different adversary could notice that it is possible to build a fork
% with divergence of $2k$, and act differently than the purist
% adversary.
%
% Here, we develop and optimal online adversary
% that makes decisions with the goal of maximizing divergence.
% Crucially, if $w$ is forkable, we want our new adversary to still successfully fork $w$.

%%% Local Variables:
%%% mode: latex
%%% TeX-master: "main"
%%% End:

\section{General settlement guarantees and proof of main theorems}
\label{sec:estimates}
With the recursive formulation for relative margin in hand, 
we study the stochastic process that arises when the
characteristic string $w$ is chosen from a distribution 
satisfying the $\epsilon$-martingale condition. 
Let us write $w = xy$ (where the decomposition is arbitrary) and 
let $E$ be the event that the relative margin $\mu_x(y)$ is non-negative. 
As Fact~\ref{fact:margin-balance} and Observation~\ref{obs:settlement-balanced-fork} point out, 
this event has a direct bearing on the settlement violation on $w$. 

In this section, we prove two bounds on the probability of the event $E$.
The first bound corresponds to the distribution 
$\mathcal{B}_\epsilon$ 
whereas the second bound applies to any distribution that 
satisfies the $\epsilon$-martingale condition. 
(Recall that the distribution $\mathcal{B}_\epsilon$, mentioned in Theorem~\ref{thm:main}, 
satisfies the $\epsilon$-martingale condition with equality.)
Our exposition in this section culminates in the proofs of our main theorems. 

We start with the following theorem 
which is a direct consequences of these bounds; see Section~\ref{sec:bounds} for a proof.
\begin{theorem}\label{thm:plain-main}
  Let $T, k \in \NN$.
  Let $w \in \{0,1\}^T$ be a random variable
  satisfying the $\epsilon$-martingale condition. 
  Consider the decomposition $w = xy, |y| = k$.  
  Then
  \[
    \Pr_{w = xy}[\text{there is an $x$-balanced fork for $xy$}] 
    = \Pr_{w = xy}[\mu_x(y) \geq 0] 
    \leq \exp(-\Omega(k))
    % = \exp({-\epsilon^3 (1 - O(\epsilon)) k/2})
    \,.
  \]
  (The asymptotic notation hides constants that depend only on $\epsilon$.)
\end{theorem}
Notice how the final bound does not depend on $|x|$. 
Indeed, as we show in Lemma~\ref{lemma:rho-stationary}, 
the reach of a Boolean string $x$ 
drawn from the distribution $\mathcal{B}_\epsilon$ 
% mentioned in Theorem~\ref{thm:main} 
converges to a fixed exponential distribution as
$|x| \rightarrow \infty$. 
This limiting distribution ``stochastically dominates'' 
any distribution that satisfies the $\epsilon$-martingale condition; 
see Section~\ref{sec:dominance-rho-stationary}.
The following corollary is immediate.
% An appeal to Fact~\ref{fact:margin-balance} yields the following corollary.
\begin{corollary}\label{cor:main} 
  Let $T, s, k \in \NN$.
  Let $w \in \{0,1\}^T$ be a 
  random variable satisfying the $\epsilon$-martingale condition. 
  Then
  \begin{align}\label{eq:cor-main}
    \Pr_w\left[\parbox{65mm}{
      there is a decomposition $w = xyz$, 
      where $|x| = s - 1$ and $|y| \geq k$, 
      so that $\mu_x(y) \geq 0$ 
    }\right] 
      \leq O(1) \cdot \exp(-\Omega(k))
    \,.
  \end{align}
\end{corollary}
\begin{proof}
  Notice that Theorem~\ref{thm:plain-main} works for \emph{any} prefix $x$ 
  of the characteristic string $w = xy$.
  Thus we can fix the prefix $x$ with length $s - 1$ and 
  sum the bound in Theorem~\ref{thm:plain-main} 
  over all suffixes $y$ with length at least $k$. 
  This would give an upper bound to the left-hand side of our claim, 
  the bound being 
  $\sum_{t \geq k} \exp(-\Omega(t)) = O(1)\cdot \exp(-\Omega(k))$. 
\end{proof}

We obtain another imporant corollary by setting $|x| = 0$ and $|y| = n$ in Theorem~\ref{thm:plain-main}. 
\begin{corollary}\label{coro:forkable-rare}%[cf. \cite{KRDO17}]
  Let $w \in \{0,1\}^n$ be a random variable satisfying the $\epsilon$-martingale condition. Then
  \[
    \Pr[\text{$w$ is forkable}] = \Pr[\mu(w) \geq 0] \leq \exp(-\Omega(n))
    \,.
  \]
\end{corollary}
Thus \emph{forkable strings are rare} 
where ``forkable'' is defined in Definition~\ref{def:forkable}.
This result 
significantly strengthens the $\exp(-\Omega(\sqrt{n}))$ 
bound obtained in Theorem 4.13 of~\cite{KRDO17}. 
The improvement comes in two respects: 
first, Corollary~\ref{cor:main} improves the exponent from $\sqrt{n}$ to $n$, 
and second, the characteristic string is allowed to be drawn 
from any distribution satisfying the $\epsilon$-martingale condition. 
For comparison, the characteristic string in Theorem 4.13 of~\cite{KRDO17} 
has the distribution $\mathcal{B}_\epsilon$, i.e., 
the bits were i.i.d.\ Bernoulli random variables 
with expectation $(1 - \epsilon)/2$.

\subsection{Two bounds for non-negative relative margin}\label{sec:bounds}
The main ingredients to proving Theorem~\ref{thm:plain-main} 
are two bounds on the event that for a characteristic string $xy$, 
the relative margin $\mu_x(y)$ is non-negative. 

\begin{bound}\label{bound:analytic}
  Let $x \in \{0,1\}^m$ and $y \in \{0,1\}^k$ be independent random
  variables, each chosen according to $\mathcal{B}_\epsilon$. Then
  \[
    \Pr[\mu_x(y) \geq 0] 
      \leq \exp({-\epsilon^3 (1 - O(\epsilon)) k/2})
    \,.
  \]
\end{bound}
% We are also interested in characteristic
% strings drawn from a distribution $\mathcal{W}$ 
% which satisfies $\epsilon$-martingale condition. 
% There are settings, 
% such as Genesis~\cite{DBLP:journals/iacr/BadertscherGKRZ18}, 
% where this flexibility is important.  

\begin{bound}\label{bound:geometric}
  Let $x \in \{0,1\}^m$ and $y \in \{0,1\}^k$ be random variables
  (jointly) satisfying the $\epsilon$-martingale condition with
  respect to the ordering $x_1, \ldots, x_m, y_1, \ldots, y_k$.  Let
  $x^\prime \in \{0,1\}^m$ and $y^\prime \in \{0,1\}^k$ be independent
  random variables, each chosen independently according to
  $\mathcal{B}_\epsilon$.  Then
  \[
    \Pr[\mu_x(y) \geq 0] \leq \Pr[\mu_{x^\prime}(y^\prime) \geq 0]
%    \,.
%  \]
%  As a result, 
%  \[
%    \Pr[\mu_x(y) \geq 0] 
      \leq \exp({-\epsilon^3 (1 - O(\epsilon)) k/2})
    \,.
  \]
\end{bound}

\paragraph{Proof of Theorem~\ref{thm:plain-main}.}
The equality is Fact~\ref{fact:margin-balance} 
and the inequality is Bound~\ref{bound:geometric}. $\qed$

% \hfill $\qed$ 

% \subsection{$\mathcal{B}$ stochastically dominates $\Distribution$; stationary distribution for reach}
\subsection{A stochastically dominant prefix distribution}\label{sec:dominance-rho-stationary}

Stochastic dominance plays an important role in the arguments
below. First of all, we observe that the distribution
$\mathcal{B}_\epsilon$ stochastically dominates any distribution
satisfying the $\epsilon$-martingale condition; this yields the first
inequality in Theorem~\ref{thm:main}. A more delicate application of
stochastic dominance is used in order to achieving bounds, such as
those of Section~\ref{sec:bounds}, that are independent of the length of
$x$. This follows from the fact that $\reach(B_{\epsilon})$ converges to a
particular, dominant distribution as its argument increases in length.

For notational convenience, we denote
the probability distribution associated with a random variable using
uppercase script letters; for example, the distribution of a random
variable $R$ is denoted by $\mathcal{R}$.  This usage should be clear
from the context.

\begin{definition}[Monotonicity and stochastic dominance]\label{def:dominance}
  Let $\Omega$ be a set endowed with a partial order $\leq$. A subset
  $A \subset \Omega$ is monotone if for all $x \leq y$, $x \in A$
  implies $y \in A$.  Let $X$ and $Y$ be random variables taking
  values in $\Omega$.
  % Let $\mathcal{X}$ and $\mathcal{Y}$ be distributions associated with 
  % $X$ and $Y$, respectively. 
  We say that $X$ \emph{stochastically dominates} $Y$, 
  written $Y \dominatedby X$, if 
  $
    \mathcal{X}(A) \geq \mathcal{Y}(A)
    % \,.
    $ for all monotone $A \subseteq \Omega$.  As a special case, when
    $\Omega = \R$, $Y \dominatedby X$ if
    $\Pr[X \geq \Lambda] \geq \Pr[Y \geq \Lambda]$ for every
    $\Lambda \in \R$.  We extend this notion to probability
    distributions in the natural way.
\end{definition}
% \begin{definition}[Stochastic dominance]\label{def:dominance} 
% Let $X$ and $Y$ be random
%   variables taking values in $\R$. We say that $X$ \emph{stochastically
%   dominates} $Y$, written $Y \dominatedby X$ if
%   \[
%     \Pr[X \geq \Lambda] \geq \Pr[Y \geq \Lambda]
%   \]
%   for every $\Lambda \in \R$.  We extend this notion to probability
%   distributions in the natural way.
% %  In addition, for two distributions $\mathcal{X}, \mathcal{Y}$,
% %  we say that $\mathcal{X}$ stochastically dominates $\mathcal{Y}$ 
% %  if and only if there are random variables $X \sim \mathcal{X}, Y \sim %\mathcal{Y}$ such that $Y \dominatedby X$;
% %  this is written as $\mathcal{Y} \dominatedby \mathcal{X}$.
% \end{definition}
Observe that for any non-decreasing function $u$ defined on $\Omega$,
$Y \dominatedby X$ implies $u(Y) \leq u(X)$. Finally, we note that for
real-valued random variables $X$, $Y$, and $Z$, if $Y \dominatedby X$
and $Z$ is independent of both $X$ and $Y$, then
$Z + Y \dominatedby Z + X$.

% Let $m \in \NN$ and suppose 
% $W = (W_1, \ldots, W_m) \in \{0,1\}^m$ satisfies the 
% $\epsilon$-martingale condition. 
% It turns out that $\rho(W)$ 
% is stochastically dominated by 
% the distribution of $\rho(B_1, \ldots, B_m)$, 
% where each $B_i \in \{0, 1\}$ is 
% an independent Bernoulli random variable with parameter $(1 - \epsilon)/2$.
% In addition, $\rho(B_1, \ldots, B_m)$ is stochastically dominated by 
% its limiting (stationary) distribution where we take $m \rightarrow \infty$.

%=======================================================
\begin{lemma}\label{lemma:rho-stationary}
  % Let $n \in \NN$ and consider a sequence of random variables
  % $W = (W_1, \ldots, W_n) \in \{0,1\}^n$ satisfying the
  % $\epsilon$-martingale condition. 
  Suppose $W = (W_1, \ldots, W_n) \in \{0,1\}^n$ satisfies the 
  $\epsilon$-martingale condition. 
  Let $\epsilon \in (0, 1)$ and $B = (B_1, \ldots, B_n) \in \{0,1\}^n$ 
  where each $B_i$ is independent with expectation $(1- \epsilon)/2$.
  Let $R_\infty \in \{0, 1, \ldots\}$ be a random variable 
  whose distribution $\StationaryRho$ is defined as 
    \begin{equation}
      \label{eq:stationary}
      \StationaryRho(k) 
        = \Pr[R_\infty = k] 
        \defeq \left(\frac{2\epsilon}{1+\epsilon}\right)\cdot \left(\frac{1-\epsilon}{1 + \epsilon}\right)^k
        \qquad \text{for $k = 0, 1, 2, \ldots$}\ 
      \,.
    \end{equation}
  Then $\rho(W) \dominatedby \rho(B) \dominatedby R_\infty$.
\end{lemma}

\begin{proof}
  We begin by observing that $B$ stochastically dominates $W$. As a
  matter of notation, for any fixed values
  $w_1, \ldots, w_k \in \{0,1\}^k$, let
  \[
    \theta[w_1, \ldots, w_k] = \Pr[ W_{k+1} = 1 \mid
    \text{$W_i = w_i$, for $i \leq k$}] \leq (1 - \epsilon)/2
  \]
  and $\theta[\varepsilon] = \Pr[W_1 = 1]$ 
  where $\varepsilon$ is the empty string. Then consider $n$ uniform and
  independent real numbers $(A_1, \ldots, A_n)$, each taking a value
  in the unit interval $[0,1]$; we use these random variables to construct a monotone
  coupling between $W$ and $B$. 
  Specifically, define $\beta: [0,1]^n \rightarrow \{0,1\}^n$
  by the rule $\beta(\alpha_1, \ldots, \alpha_n) = (b_1, \ldots, b_n)$
  where
  \[
    b_t = \begin{cases} 1 & \text{if $\alpha_t \leq (1-\epsilon)/2$},\\
      0 & \text{if $\alpha_t > (1 - \epsilon)/2$},
    \end{cases}
  \]
  and define
  $B = (B_1, \ldots, B_n) = \beta(A_1, \ldots, A_n)$; these
  $B_i$s are independent zero-one Bernoulli random variables with expectation
  $(1-\epsilon)/2$. Likewise define the function
  $\omega:[0,1]^n \rightarrow \{0,1\}^n$ so that
  $\omega(\alpha_1, \ldots, \alpha_n) = (w_1, \ldots, w_n)$
  where each $w_t$ is assigned by the iterative rule
  \[
    w_{t+1} = \begin{cases} 1 & \text{if $\alpha \leq \theta[w_1, \ldots, w_t]$},\\
      0 & \text{if $\alpha > \theta[w_1, \ldots, w_t]$},
    \end{cases}
  \]
  and observe that the probability law of
  $\omega(A_1, \ldots, A_n)$ is precisely that of
  $W = (W_1, \ldots, W_n)$. For convenience, we simply identify the
  random variable $W$ with $\omega(A_1, \ldots, A_n)$. Note
  that for any $\alpha = (\alpha_1, \ldots, \alpha_n)$ and for each
  $i$, the $i$th coordinates of $\beta(\alpha)$ and $\omega(\alpha)$ satisfy
  $\omega(\alpha)_i \leq \beta(\alpha)_i$ 
  % (which is to say that $W_i \leq B_i$). 
  % It follows immediately that
  % $\rho(\omega(\alpha)) \leq \rho(\beta(\alpha))$ with probability 1 and
  % hence $\rho(W) \dominatedby \rho(B)$. 
  % See~\cite[Lemma 22.5]{LevinPeres}. 
  (which is to say that $W_i \leq B_i$ with probability 1). 
  But this is equivalent to saying $W \dominatedby B$. 
  (See~\cite[Lemma 22.5]{LevinPeres}.) 
  Now consider the following partial order $\leq$ on the $n$-bit Boolean strings: 
  for $x,y \in \{0,1\}^n$, 
  we write $x \leq y$ if and only if $x_i = 1$ implies $y_i = 1, i \in [n]$.
  Since $\rho$ is non-decreasing with respect to this partial order, 
  we have  
  $\rho(\omega(\alpha)) \leq \rho(\beta(\alpha))$ with probability 1 and
  hence $\rho(W) \dominatedby \rho(B)$ as well.

  %  \paragraph{$R_\infty$ stochastically dominates $\rho(B)$.}
  To complete the proof, we now establish that
  $\rho(B) \dominatedby R_\infty$.  We remark that the random variables
  $\rho(B)$ (and $R_\infty$) have an immediate interpretation in terms
  of the Markov chain corresponding to a biased random walk on $\Z$
  with a ``reflecting boundary'' at -1. Specifically, consider the
  Markov chain on $\{0, 1, \ldots\}$ given by the transition diagram
  
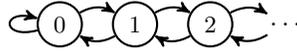
\begin{center}
    \begin{tikzpicture}[scale=1,>=stealth', auto, semithick,
      flat/.style={circle,draw=black,thick,text=black,font=\small}]
      \node[flat] (n0) at (0,0)  {$0$};
      \node[flat] (n1) at (1,0)  {$1$};
      \node[flat] (n2) at (2,0)  {$2$};
      \node[flat,white] (n3) at (3,0) {$ \ \ \ $};
      \node[] at (3,0) {$\ldots$};
      \draw[thick,->,bend left] (n0) to (n1);
      \draw[thick,->,bend left] (n1) to (n2);
      \draw[thick,->,bend left] (n2) to (n3);
      \draw[thick,->,loop left] (n0) to (n0);
      \draw[thick,->,bend left] (n1) to (n0);
      \draw[thick,->,bend left] (n2) to (n1);
      \draw[thick,->,bend left] (n3) to (n2);
    \end{tikzpicture}
  \end{center}
  where edges pointing right have probability $(1-\epsilon)/2$ and edges pointing left---including the loop at 0---have probability $(1+\epsilon)/2$. Examining the recursive description of $\rho(w)$, it is easy to confirm that the random variable $\rho(B_1, \ldots, B_n)$ is precisely given by the result of evolving the Markov chain above for $n$ steps with all probability initially placed at 0. It is further easy to confirm that the distribution given by~\eqref{eq:stationary} above is stationary for this chain.

  % The above Markov chain is \emph{irreducible} since every state is
  % reachable from every state.  In addition, it is easy to see that this
  % chain is \emph{aperiodic}---at time $n$, the walk visits every state
  % $s \in \{0, 1, \cdots, n\}$ with a strictly positive probability.
  % According to~\citep[Theorem 21.12]{LevinPeres}, an irreducible Markov
  % chain is positive recurrent if and only if there exists a distribution
  % $\pi$ on the state space satisfying $\pi P = \pi$, where $P$ is the
  % transition matrix of the chain.  We can take $\pi$ to be the
  % distribution given by~\eqref{eq:stationary}, implying that the chain
  % is positive recurrent.  Consequently, by ~\citep[Theorem
  % 21.14]{LevinPeres}, this irreducible, aperiodic, positive recurrent
  % Markov chain converges to the unique stationary distribution $\pi$.
  % It follows that the distributions $R_n$ limit to $R_\infty$.

  To establish stochastic dominance, it is convenient to work with the
  underlying distributions and consider walks of varying lengths: let
  $\DistRho_n: \Z \rightarrow \R$ denote the probability distribution given by
  $\rho(B_1, \ldots, B_n)$; likewise
  define $\DistRho_\infty$. For a distribution $\DistRho$ on $\Z$, we define $[\DistRho]_0$
  to denote the probability distribution obtained by shifting all
  probability mass on negative numbers to zero; that is, for $x \in \Z$,
  \[
    [\DistRho]_0(x) = \begin{cases} \DistRho(x) & \text{if $x > 0$},\\
      \sum_{t \leq 0} \DistRho(t) & \text{if $x = 0$},\\
      0 & \text{if $x < 0$.}
    \end{cases}
  \]
  We observe that if $A \dominatedby C$ then $[A]_0 \dominatedby [C]_0$ for any
  distributions $A$ and $C$ on $\Z$. It will also be convenient to
  introduce the shift operators: for a distribution
  $\DistRho: \Z \rightarrow \R$ and an integer $k$, we define $S^k\DistRho$ to be the
  distribution given by the rule $S^k\DistRho(x) = \DistRho(x-k)$. With these
  operators in place, we may write
  \[
    \DistRho_t = \left(\frac{1 - \epsilon}{2}\right) S^1 \DistRho_{t-1} +
      \left(\frac{1 + \epsilon}{2}\right) \left[S^{-1}\DistRho_{t-1} \right]_0\,,
  \]
  with the understanding that $\DistRho_0$ is the distribution placing unit probability at $0$. The proof now proceeds by induction. It is clear that $\DistRho_0 \dominatedby \DistRho_\infty$. Assuming that $\DistRho_n \dominatedby \DistRho_\infty$, we note that for any $k$
  \[
    S^k \DistRho_n \dominatedby S^k \DistRho_\infty \qquad \text{and, additionally, that
    }\qquad [S^{-1}\DistRho_n]_0 \dominatedby [S^{-1}\DistRho_{\infty}]_0\,.
  \]
  Finally, it is clear that stochastic dominance respects convex combinations, 
  in the sense that if $A_1 \dominatedby C_1$ and $A_2 \dominatedby C_2$ then 
  $\lambda A_1 + (1-\lambda) A_2 \dominatedby \lambda C_1 + (1-\lambda) C_2$ (for $0 \leq \lambda \leq 1$). We conclude that
  \[
    \DistRho_{t+1} = \left(\frac{1 - \epsilon}{2}\right) S^1 \DistRho_{t} +
      \left(\frac{1 + \epsilon}{2}\right) \left[S^{-1}\DistRho_{t} \right]_0 \dominatedby \left(\frac{1 - \epsilon}{2}\right) S^1 \DistRho_{\infty} +
      \left(\frac{1 + \epsilon}{2}\right) \left[S^{-1}\DistRho_{\infty} \right]_0 
      \,.
  \]
  By inspection, the right-hand side equals $\DistRho_{\infty}$, as desired. 
  Hence $\rho(B) \dominatedby R_\infty$.
  % $\qedhere$
\end{proof}

  \paragraph{Remark.} 
  In fact, the random variable $\rho(B)$
  actually converges to $R_\infty$ as $n \rightarrow \infty$. 
  This can be seen, for example, 
  by solving for the stationary distribution of the Markov chain in the proof above. 
  However, we will only require the dominance for our exposition. 
  Importantly, since $\mu_x(\varepsilon) = \rho(x)$, and 
  $\Pr[\mu_x(y) \geq 0]$ increases monotonically 
  with an increase in $\Pr[\mu_x(\varepsilon) \geq r]$ for any $r \geq 0$, 
  it suffices to take $|x| \rightarrow \infty$ 
  when reasoning about an upper bound on $\Pr[\mu_x(y) \geq 0]$. 

%=======================================================
%=======================================================
\subsection{Proof of Bound~\ref{bound:analytic}}\label{sec:gf-proof}

%\begin{proof}[of Bound~\ref{bound:analytic}]
% \begin{proof}
  Anticipating the proof, we make a few remarks about generating
  functions and stochastic dominance.  We reserve the term
  \emph{generating function} to refer to an ``ordinary'' generating
  function which represents a sequence $a_0, a_1, \ldots$ of
  non-negative real numbers by the formal power series
  $\gf{A}(Z) = \sum_{t = 0}^\infty a_t Z^t$. When
  $\gf{A}(1) = \sum_t a_t = 1$ we say that the generating function is
  a \emph{probability generating function}; in this case, the
  generating function $\gf{A}$ can naturally be associated with the
  integer-valued random variable $A$ for which $\Pr[A = k] = a_k$. If
  the probability generating functions $\gf{A}$ and $\gf{B}$ are
  associated with the random variables $A$ and $B$, it is easy to
  check that $\gf{A} \cdot \gf{B}$ is the generating function
  associated with the convolution $A + B$ (where $A$ and $B$ are
  assumed to be independent).  Translating the notion of stochastic
  dominance to the setting with generating functions, we say that the
  generating function $\gf{A}$ \emph{stochastically dominates}
  $\gf{B}$ if $\sum_{t \leq T} a_t \leq \sum_{t \leq T} b_t$ for all
  $T \geq 0$; we write $\gf{B} \dominatedby \gf{A}$ to denote this state of
  affairs. If $\gf{B}_1 \dominatedby \gf{A}_1$ and
  $\gf{B}_2 \dominatedby \gf{A}_2$ then
  $\gf{B}_1 \cdot \gf{B}_2 \dominatedby \gf{A}_1 \cdot \gf{A}_2$ and
  $\alpha \gf{B}_1 + \beta \gf{B}_2 \dominatedby \alpha \gf{A}_1 + \beta
  \gf{A}_2$ (for any $\alpha, \beta \geq 0$).  Moreover, if
  $\gf{B} \dominatedby \gf{A}$ then it can be checked that
  $\gf{B}(\gf{C}) \dominatedby \gf{A}(\gf{C})$ for any probability
  generating function $\gf{C}(Z)$, where we write $\gf{A}(\gf{C})$ to
  denote the composition $\gf{A}(\gf{C}(Z))$.

  Finally, we remark that
  if $\gf{A}(Z)$ is a generating function which converges as a
  function of a complex $Z$ for $|Z| < R$ for some non-negative $R$, 
  $R$ is called the \emph{radius of convergence} of $\gf{A}$.  
  It follows from \citep[Theorem 2.19]{WilfGF} that 
  $\lim_{k \rightarrow \infty} {a_k}R^k = 0$ and $|a_k| = O(R^{-k})$. 
	In addition, if $\gf{A}$ is a probability generating function associated with the
  random variable $A$ then it follows that
  $\Pr[A \geq T] = O(R^{-T})$.
  
  We define $p = (1 - \epsilon)/2$ and $q = 1 - p$ and 
  as in the proof of Bound~\ref{bound:geometric},
  consider the independent $\{0,1\}$-valued random variables
  $w_1, w_2, \ldots$ where $\Pr[w_t = 1] = p$. We also define the
  associated $\{\pm1\}$-valued random variables $W_t =
  (-1)^{1+w_t}$.

	Although our actual interest is in the random variable $\mu_x(y)$ 
	from~\eqref{eq:mu-relative-recursive} on a characteristic string $w=xy$, 
  we begin by analyzing the case when $|x|=0$. 

  %\vspace{-2ex}
	\paragraph{Case 1: $x$ is the empty string.}
  In this case, the random variable $\mu_x(y)$ is identical to $\mu(w)$ 
  from~\eqref{eq:mu-recursive} with $w = y$. 
	Our strategy is to study the probability generating
  function
  \[
    \gf{L}(Z) = \sum_{t = 0}^\infty \ell_t Z^t
  \]
  where $\ell_t = \Pr[\text{$t$ is the last time $\mu_t =
    0$}]$. Controlling the decay of the coefficients $\ell_t$ suffices
  to give a bound on the probability that $w_1\ldots w_k$ is forkable
  because
  \[
    \Pr[\text{$w_1 \ldots w_k$ is forkable}] \leq 1 - \sum_{t =
      0}^{k-1} \ell_t = \sum_{t = k}^\infty \ell_t\,.
  \]
  It seems challenging to give a closed-form algebraic expression for
  the generating function $\gf{L}$; our approach is to develop a
  closed-form expression for a probability generating function
  $\gf{\hat{L}} = \sum_t \hat{\ell}_t Z^t$ which stochastically
  dominates $\gf{L}$ and apply the analytic properties of this closed
  form to bound the partial sums $\sum_{t \geq k} \hat{\ell}_k$.
  Observe that if $\gf{L} \dominatedby \gf{\hat{L}}$ then the series
  $\gf{\hat{L}}$ gives rise to an upper bound on the probability that
  $w_1\ldots w_k$ is forkable as
  $\sum_{t=k}^\infty \ell_t \leq \sum_{t=k}^\infty \hat{\ell}_t$.

  The coupled random variables $\rho_t$ and $\mu_t$ are Markovian
  in the sense that values $(\rho_s, \mu_s)$ for $s \geq t$ are
  entirely determined by $(\rho_t, \mu_t)$ and the subsequent
  values $W_{t+1}, \ldots$ of the underlying variables $W_i$. We
  organize the sequence
  $(\rho_0, \mu_0), (\rho_1, \mu_1), \ldots$ into ``epochs''
  punctuated by those times $t$ for which $\rho_t = \mu_t =
  0$. With this in mind, we define $\gf{M}(Z) = \sum m_t Z^t$ to be
  the generating function for the first completion of such an epoch,
  corresponding to the least $t > 0$ for which
  $\rho_t = \mu_t = 0$. As we discuss below, $\gf{M}(Z)$ is not a
  probability generating function, but rather
  $\gf{M}(1) = 1 - \epsilon$. It follows that 
  \begin{align}\label{eq:L-def}
    \gf{L}(Z) 
      &= \left(
        1 + 
        (1 - \epsilon) \cdot \frac{\gf{M}(Z)}{\gf{M}(1)} + 
        \left( (1 - \epsilon) \cdot \frac{\gf{M}(Z)}{\gf{M}(1)} \right)^2 + 
        \cdots \right)\cdot \epsilon  \nonumber \\
      &= ( 1 + \gf{M}(Z) + \gf{M}(Z)^2 + \cdots ) \cdot \epsilon \nonumber \\
      &= \frac{\epsilon}{1 - \gf{M}(Z)}
      \,.
  \end{align}
  The expression above represents the following geometric process: 
  before the beginning of an epoch, 
  we ``ask'' whether the walk is ever going to come back to zero. 
  With probability $\epsilon$, the answer is ``no'' and we stop the process. 
  Otherwise, i.e., with probability $1 - \epsilon$, 
  we commence an epoch which is guaranteed to finish; 
  then we ask again.

  Below we develop an analytic expression for a generating function
  $\gf{\hat{M}}$ for which $\gf{M} \dominatedby \gf{\hat{M}}$ and define
  $\gf{\hat{L}} = \epsilon/(1 - \gf{\hat{M}}(Z))$. We then proceed as
  outlined above, noting that $\gf{L} \dominatedby \gf{\hat{L}}$ and using
  the asymptotics of $\gf{\hat{L}}$ to upper bound the probability
  that a string is forkable.

  In preparation for defining $\gf{\hat{M}}$, we set down two
  elementary generating functions for the ``descent'' and ``ascent''
  stopping times. Treating the random variables $W_1, \ldots$ as
  defining a (negatively) biased random walk, define $\gf{D}$ to be
  the generating function for the \emph{descent stopping time} of the
  walk; this is the first time the random walk, starting at 0, visits
  $-1$. The natural recursive formulation of the descent time yields a
  simple algebraic equation for the descent generating function,
  $\gf{D}(Z) = qZ + pZ \gf{D}(Z)^2$, and from this we may conclude
  \[
    \gf{D}(Z) = \frac{1 - \sqrt{1 - 4pqZ^2}}{2pZ}\,.
  \]
  We likewise consider the generating function $\gf{A}(Z)$ for the
  \emph{ascent stopping time}, associated with the first time the
  walk, starting at 0, visits 1: we have
  $\gf{A}(Z) = pZ + qZ \gf{A}(Z)^2$ and
  \[
    \gf{A}(Z) = \frac{1 - \sqrt{1 - 4pqZ^2}}{2qZ}\,.
  \]
  Note that while $\gf{D}$ is a probability generating function, the
  generating function $\gf{A}$ is not: according to the classical
  ``gambler's ruin'' analysis~\cite{Grinstead:1997ng}, the probability
  that a negatively-biased random walk starting at 0 ever rises to 1
  is exactly $p/q$; thus $\gf{A}(1) = p/q$.

  Returning to the generating function $\gf{M}$ above, we note that an
  epoch can have one of two ``shapes'': in the first case, the epoch
  is given by a walk for which $W_1 = 1$ followed by a descent (so
  that $\rho$ returns to zero); in the second case, the epoch is
  given by a walk for which $W_1 = -1$, followed by an ascent (so that
  $\mu$ returns to zero), followed by the eventual return of $\rho$
  to 0. Considering that when $\rho_t > 0$ it will return to zero
  in the future almost surely, it follows that the probability that
  such a biased random walk will complete an epoch is
  $p + q(p/q) = 2p = 1 - \epsilon$, as mentioned in the discussion
  of~\eqref{eq:L-def} above. One technical difficulty arising in a
  complete analysis of $\gf{M}$ concerns the second case discussed
  above: while the distribution of the smallest $t > 0$ for which
  $\mu_t = 0$ is proportional to $\gf{A}$ above, the distribution of
  the smallest subsequent time $t'$ for which $\rho_{t'} = 0$
  depends on the value $t$. More specifically, the distribution of the
  return time depends on the value of $\rho_t$. Considering that
  $\rho_t \leq t$, however, this conditional distribution (of the
  return time of $\rho$ to zero conditioned on $t$) is
  stochastically dominated by $\gf{D}^t$, the time to descend $t$
  steps. This yields the following generating function $\gf{\hat{M}}$
  which, as described, stochastically dominates $\gf{M}$:
  \[
    \gf{\hat{M}}(Z) = pZ\cdot \gf{D}(Z) + qZ \cdot \gf{D}(Z) \cdot
    \gf{A}(Z\cdot \gf{D}(Z))\,.
  \]
  
  It remains to establish a bound on the radius of convergence of
  $\gf{\hat{L}}$. Recall that if the radius of convergence of
  $\gf{\hat{L}}$ is $\exp(\delta)$ it follows that
  $\Pr[\text{$w_1 \ldots w_k$ is forkable}] = O(\exp(-\delta k))$. A
  sufficient condition for convergence of
  $\gf{\hat{L}}(z) = \epsilon/(1 - \gf{\hat{M}}(z))$ at $z$ is that
  that all generating functions appearing in the definition of
  $\gf{\hat{M}}$ converge at $z$ and that the resulting value
  $\gf{\hat{M}}(z) < 1$.
  
  The generating function $\gf{D}(z)$ (and $\gf{A}(z)$) converges when
  the discriminant $1 - 4pqz^2$ is positive; equivalently
  $|z| < 1/\sqrt{1 - \epsilon^2}$ or
  $|z| < 1 + \epsilon^2/2 + O(\epsilon^4)$. Considering
  $\gf{\hat{M}}$, it remains to determine when the second term,
  $qz D(z) \gf{A}(z \gf{D}(z))$, converges; this is likewise determined by
  positivity of the discriminant, which is to say that
  \[
    1 - (1 - \epsilon^2)\left(\frac{1 - \sqrt{1 - (1 - \epsilon^2)z^2}}{1 - \epsilon}\right)^2 > 0\,.
  \]
  Equivalently,
  \[
    |z| <  \sqrt{\frac{1}{1 + \epsilon}\left(\frac{2}{\sqrt{1 - \epsilon^2}} - \frac{1}{1+\epsilon}\right)} = 1 + \epsilon^3/2 + O(\epsilon^4) 
		\, .	
  \]
  Note that when the series $pz \cdot \gf{D}(z)$ converges, it
  converges to a value less than $1/2$; the same is true of
  $qz \cdot \gf{A}(z)$. It follows that for
  $|z| = 1 + \epsilon^3/2 + O(\epsilon^4)$, $|\gf{\hat{M}}(z)| < 1$
  and $\gf{\hat{L}}(z)$ converges, as desired. We conclude that
	\begin{align}
	  \Pr[\text{$w_1 \ldots w_k$ is forkable}] &= \exp(-\epsilon^3(1 + O(\epsilon))k/2)\,.
	\label{eq:prob_forkable_gf}
	\end{align}

          %\vspace{-2ex}
	\paragraph{Case 2: $x$ is non-empty.}
        The relative margin before $y$ begins is $\mu_x(\varepsilon)$.
        Recalling that $\mu_x(\varepsilon) = \rho(x)$ and conditioning on the event that $\rho(x) = r$, 
    let us define the random variables $\left\{ \tilde{\mu}_t \right\}$ for $t = 0, 1, 2, \cdots$ as follows: $\tilde{\mu}_0 = \rho(x)$ and
    \[
      %\, , \qquad \text{and} \qquad
      \Pr[\tilde{\mu}_t = s]\, =\, \Pr[\mu_x(y) = s \mid \rho(x) = r \text{ and } |y| = t ]
      \, .
    \]
    If the $\tilde{\mu}$ random walk makes the $r$th descent at some time $t < n$, then $\tilde{\mu}_t = 0$ and the remainder of the walk is 
	identical to an $(k-t)$-step $\mu$ random walk which we have already analyzed. 
	Hence we investigate the probability generating function
	\[
			\gf{B}_r(Z) = \gf{D}(Z)^r \gf{L}(Z) \quad \text{with coefficients} \quad
      b^{(r)}_t := \Pr[t \text{ is the last time } \tilde{\mu}_t = 0 \mid \tilde{\mu}_0 = r]	
  \]
  where $t = 0, 1, 2, \cdots$. Our interest lies in the quantity 
    \[
      b_t 
      := \Pr[t \text{ is the last time } \tilde{\mu}_t = 0] 
      = \sum_{r\geq 0}{  b^{(r)}_t \DistRho_m(r) } 
      \,,
     \]
  where the \emph{reach distribution} 
  $\DistRho_m : \Z \rightarrow [0,1]$ 
  associated with the random variable $\rho(x), |x| = m$ is defined as 
  \begin{align}\label{eq:dist-rho}
    % \DistRho_m(r) &= \Pr_{W \sim \mathcal{D}}[\rho(W) = r \Given W \text{ has length } m]
    \DistRho_m(r) &= \Pr_{x \SuchThat |x| = m}[\rho(x) = r]
    \, .
  \end{align}
     % from~\eqref{eq:dist-rho}.
     Let $\gf{R}_m(Z)$ be the probability generating function
     for the distribution $\DistRho_m$. 
     Using Lemma~\ref{lemma:rho-stationary} and Definition~\ref{def:dominance}, we deduce that
     $\gf{R}_m \dominatedby \gf{R}_\infty$ for every $m \geq 0$ since 
    $\DistRho_m \dominatedby \StationaryRho$.
     In addition, it is easy to check from~\eqref{eq:stationary} that
     the probability generating function for $\StationaryRho$ is in fact
     $\gf{R}_\infty(Z) = (1-\beta)/(1-\beta Z)$ where $\beta := (1-\epsilon)/(1+\epsilon)$. 
    Thus the generating function corresponding to the
     probabilities $\{b_t\}_{t=0}^\infty$ is
	\begin{align}
		\gf{B}(Z) 
		&= \sum_{t=0}^\infty{b_t Z^t} = \sum_{r=0}^\infty{\DistRho_m(r) \sum_{t=0}^\infty{b_t^{(r)} Z^t} } = \sum_{r=0}^\infty{\DistRho_m(r) \gf{B}_r(Z) } \nonumber \\
    &= \gf{L}(Z) \sum_{r=0}^\infty{\DistRho_m(r) \gf{D}(Z)^r}    \nonumber 
		= \gf{L}(Z)\  \gf{R}_m (\gf{D}(Z)) \nonumber 
		\dominatedby \gf{\hat{L}}(Z)\  \gf{R}_\infty (\gf{D}(Z))  \nonumber \\
    &= \frac{(1-\beta)\,\gf{\hat{L}}(Z) }{1 - \beta \gf{D}(Z)}
		\, .
	\label{eq:gf-mu-relative}
	\end{align}
  The dominance notation above follows because
  $\gf{L} \dominatedby \gf{\hat{L}}$ and $\gf{R}_m \dominatedby \gf{R}_\infty$.

  For $\gf{B}(Z)$ to converge, we need to check that $\gf{D}(Z)$
  should never converge to $1/\beta$.  One can easily check that
  the radius of convergence of $\gf{D}(Z)$---which is
  $\displaystyle 1/\sqrt{1-\epsilon^2}$---is strictly less than $1/\beta$ when
  $\epsilon > 0$.  We conclude that $\gf{B}(Z)$ converges if
  both $\gf{D}(Z)$ and $\gf{L}(Z)$ converge.  The radius of
  convergence of $\gf{B}(Z)$ would be the smaller of the radii
  of convergence of $\gf{D}(Z)$ and $\gf{L}(Z)$.  We already
  know from the previous analysis that $\gf{\hat{L}}(Z)$ has the
  smaller radius of the two; therefore, the bound
  in~\eqref{eq:prob_forkable_gf} applies to the relative margin $\mu_x(y)$
  for $|x|\geq 0$. 
  \hfill $\qed$  
  % $\qedhere$
  % $\qed$
% \end{proof}

%=======================================================
\subsection{Proof of Bound~\ref{bound:geometric}}\label{sec:martingale-proof-new}

Let $\epsilon \in (0, 1)$,  
$W \in \{0,1\}^m, W^\prime \in \{0,1\}^k$ 
where both $(W_1, \ldots, W_n)$ and $(W^\prime_1, \ldots, W^\prime_n)$ 
satisfy the $\epsilon$-martingale condition. 
Let $B \in \{0,1\}^m, B^\prime \in \{0,1\}^k$ where the components of $B, B^\prime$ 
are independent with expectation $(1 - \epsilon)/2$.
By Lemma~\ref{lemma:rho-stationary}, 

\begin{equation*}\label{eq:WB-dominance}\tag{$*$}
  W \dominatedby B\quad \text{and} \quad W^\prime \dominatedby B^\prime
  \,. 
\end{equation*}

Let us define the partial order $\leq$ on Boolean strings $\{0,1\}^k, k \in \NN$ 
as follows: 
$a \leq b$ if and only if
for all $i \in [k]$, $a_i = 1$ implies $b_i = 1$. 
Let $\mu : \{0,1\}^k \rightarrow \Z$ be the margin function 
from Lemma~\ref{lem:relative-margin}. 
Observe that for Boolean strings $a, a^\prime, b, b^\prime$ 
with $|a| = |a^\prime|$ and $|b| = |b^\prime|$, 
(i.) $b \leq b^\prime$ implies $\mu_a(b) \leq \mu_a(b^\prime)$ and 
(ii.) $a \leq a^\prime$ implies $\mu_a(b) \leq \mu_{a^\prime}(b)$. 
That is, 
\begin{equation*}\label{eq:mu-WB-dominance}\tag{$\dagger$}
  \text{$\mu_a(b)$ is non-decreasing in both $a$ and $b$}
  \,.   
\end{equation*}

Using~\eqref{eq:WB-dominance} and~\eqref{eq:mu-WB-dominance}, 
it follows that $\mu_W(W^\prime) \dominatedby \mu_{B}(B^\prime)$. 
% Recall that $P \dominatedby Q$ is true if, 
% for all monotone sets $E$, 
% by $\Pr[ P \in E ] \leq \Pr[Q \in E]$. 
% In particular, concerning the dominance $\mu_W(W^\prime) \dominatedby \mu_{B}(B^\prime)$, 
Writing $x = W$ and $y = W^\prime$, we have 
\begin{align*}
  \Pr[\mu_x(y) \geq 0]\, 
    = \Pr[\mu_W(W^\prime) \geq 0]\, 
    \leq \Pr[\mu_{B}(B^\prime) \geq 0]
\end{align*}
where the inequality comes from the definition of stochastic dominance. 
A bound on the right-hand side 
is obtained in Bound~\ref{bound:analytic}. 
\hfill $\qed$

In Appendix~\ref{sec:martingale-proof}, 
we present a weaker bound 
on $\Pr[\mu_x(y) \geq 0]$ where the sequence 
$x_1, \ldots, x_m, y_1, \ldots, y_k$ satisfies $\epsilon$-martingale conditions. 
The proof directly uses the properties of the martingale 
and Azuma's inequality  but 
it does not use a stochastic dominance argument. 
Although it gives a bound of $3 \exp\left( -\epsilon^4 (1 - O(\epsilon) ) k/64 \right)$, 
the reader might find the proof of independent interest.

\subsection{Proof of main theorems}\label{sec:thm-proofs}

\paragraph{Proof of Theorem~\ref{thm:main}.}

Let us start with the following observation.
It allows us to formulate the
$(s, k)$-settlement insecurity of a distribution $\Distribution$
directly in terms of the relative margin.

\begin{lemma}\label{lemma:settlement-margin}
  Let $s, k, T \in \NN$. 
  Let $\Distribution$ be any distribution on $\{0,1\}^T$. 
  Then
  \[
    \mathbf{S}^{s,k}[\Distribution] \leq
      \Pr_{w \sim \Distribution} \left[\parbox{60mm}{
          there is a decomposition $w = x y z$, 
          where $|x| = s - 1$ and $|y| \geq k + 1$, 
          so that $\mu_x(y) \geq 0$
      }\right]
    \,.
  \]
\end{lemma}
\begin{proof}
  Lemma~\ref{lem:main-forks} implies that 
  $\mathbf{S}^{s,k}[\Distribution]$ is no more than 
  the probability that slot $s$ is not $k$-settled 
  for the characteristic string $w$. 
  By Observation~\ref{obs:settlement-balanced-fork}, 
  this probability, in turn, is no more than 
  the probability that there exists an $x$-balanced fork 
  $F \Fork xy$
  where we write $w = xyz, |x| = s - 1, |y| \geq k + 1, |z| \geq 0$. 
  Finally, Fact~\ref{fact:margin-balance} states that 
  for any characteristic string $xy$, 
  the two events ``exists an $x$-balanced fork $F \Fork xy$'' 
  and ``$\mu_x(y)$ is non-negative'' have the same measure. 
  Hence the claim follows. 
\end{proof}

If the distribution $\mathcal{D}$ in the lemma above 
satisfies the $\epsilon$-martingale condition, 
the probability in this lemma is no more than the probability 
in the left-hand side of Corollary~\ref{cor:main}. 
Finally, by retracing the proof of Corollary~\ref{cor:main} 
using the explicit probability from Bound~\ref{bound:geometric}, 
we see that the bound in Corollary~\ref{cor:main} is 
$O(1) \cdot \exp\bigl(-\Omega(\epsilon^3 (1 - O(\epsilon))k)\bigr)$. 
Since $\mathcal{B}_\epsilon$ satisfies the $\epsilon$-martingale condition, 
we conclude that $\mathbf{S}^{s,k}[\mathcal{B}_\epsilon]$ is no more than 
this quantity as well.
% \[
%   \mathbf{S}^{s,k}[\mathcal{B}_\epsilon] 
%       \leq O(1) \cdot \exp\bigl(-\Omega(\epsilon^3 (1 - O(\epsilon))k)\bigr)
%     \,.
% \]

For any player playing the settlement game, 
the set of strings on which the player wins is monotone 
with respect to the partial order $\leq$ defined in Section~\ref{sec:martingale-proof-new}. 
To see why, note that if the adversary wins with a specific string $w$, 
he can certainly win with any string $w^\prime$ where $w \leq w^\prime$. 
As $\mathcal{B}_\epsilon$ stochastically dominates $\mathcal{W}$, it follows that 
$
  \mathbf{S}^{s,k}[\mathcal{W}] \leq \mathbf{S}^{s,k}[\mathcal{B}_\epsilon]
$.

\hfill$\qed$

\paragraph{Proof of Theorem~\ref{thm:main-CP}}
For the first inequality, observe that if $w$ violates $\kCP$, it must violate $\kSlotCP$ as well. 
It remains to prove the second inequality. 
Let $\Distribution$ be any distribution on $\{0,1\}^T$. 
We can apply Fact~\ref{fact:margin-balance} on the statement of Theorem~\ref{thm:cp-fork} 
to deduce that 
\begin{equation*}\label{eq:main-argument-1}
  \Pr_{w \sim \Distribution}[\text{$w$ violates $\kSlotCP$}] 
    \leq 
% \Pr_{\substack{w = x y z \\ |y| \geq k + 1} }[\text{exists a $k$-settlement violation} ] 
%    = 
    \Pr_{w \sim \Distribution}\left[\parbox{55mm}{
      there is a decomposition $w = xyz$, 
      where $|y| \geq k$, 
      so that $\mu_x(y) \geq 0$ 
    } \right] 
    \,.
\end{equation*}
By using a union bound over $|x|$, the above probability is at most 
\[
    \sum_{s = 1}^{T - k + 1} 
    \quad 
      \Pr_w\left[\parbox{60mm}{
        there is a decomposition $w = xyz$, 
        where $|x| = s - 1$ and $|y| \geq k$, 
        so that $\mu_x(y) \geq 0$ 
      }\right] 
    \,.
\]
Since $w$ satisfies the $\epsilon$-martingale condition, 
we can upper bound the probability inside the sum 
using Corollary~\ref{cor:main}. 
As we have seen in the proof of Theorem~\ref{thm:main}, 
the bound in Corollary~\ref{cor:main} is 
$
  O(1) \cdot \exp\bigl(-\Omega(\epsilon^3 (1 - O(\epsilon))k)\bigr)
  % \,.
$.
It follows that the sum above is at most $T \exp\bigl(-\Omega(\epsilon^3 (1 - O(\epsilon))k)\bigr)$.
\hfill $\qed$

It remains to prove the recursive formulation of the relative margin 
from Section~\ref{sec:recursion}; 
we tackle it in the next section.

%%% Local Variables:
%%% mode: latex
%%% TeX-master: "main"
%%% End:

\section{Proof of the relative margin recurrence}
\label{sec:margin-proof}

We set the stage by formally defining \emph{fork prefixes}.

\begin{definition}[Fork prefixes]
Let $w, x \in\{0,1\}^*$ so that $x\PrefixEq w$. 
Let $F, F'$ be two forks for $x$ and $w$, respectively. 
We say that $F$ is a \emph{prefix} of $F'$ if 
$F$ is a consistently labeled subgraph of $F'$. 
That is, all vertices and edges of $F$ also appear in $F'$ and 
the label of any vertex appearing in both $F$ and $F'$ is identical. 
We denote this relationship by $F\fprefix F'$.
\end{definition}

\noindent
When speaking about a tine that appears in both $F$ and $F'$, 
we place the fork in the subscript of relevant properties, e.g., writing $\reach_F$, etc.

Observe that for any Boolean strings $x$ and $w, x \PrefixEq w$, 
one can \emph{extend} (i.e., augment) a fork prefix $F \Fork x$ into a larger fork $F' \Fork w$ so that $F \ForkPrefix F'$. 
A \emph{conservative extension} is a minimal extension in that 
it consumes the least amount of reserve (cf. Definition~\ref{def:gap-reserve-reach}), 
leaving the remaining reserve to be used in future.
Extensions and, in particular, conservative extensions 
play a critical role in the exposition that follows. 
\begin{definition}[Conservative extension of closed forks]\label{def:extension}  

  Let $w$ be a Boolean string, $F$ a closed fork for $w$, 
  and let $s$ be an honest tine in $F$.
  Let $F'$ be a closed fork for $w0$ so that $F \ForkPrefix F'$ and 
  $F'$ contains an honest tine $\sigma, \ell(\sigma) = |w| + 1$. 
  We say that \emph{$F'$ is an extension of $F$} or, equivalently, 
  that \emph{$\sigma$ is an extension of $s$}, if $s \Prefix \sigma$. 
  If, in addition, 
  $\length(\sigma) = \height(F) + 1$, 
  we call this extension a \emph{conservative extension}.
\end{definition}
Clearly, $\sigma$ is the longest tine in $F^\prime$. 
Since $\sigma$ is honest, 
it follows that 
$\length(\sigma) 
\geq 1 + \height(F) 
= 1 + \length(s) + \gap(s)$.
The root-to-leaf path in $F^\prime$ 
that ends at $\sigma$ 
contains at least $\gap(s)$ adversarial vertices $u \in F'$ 
so that $\ell(u) \in [\ell(s) + 1, |w|]$ and 
$u \not \in F$. 
If $\sigma$ is a conservative extension, 
the number of such vertices is exactly $\gap(s)$ 
and, in particular, 
the height of $F'$ is exactly one more than the height of $F$.

The main ingredients to proving Lemma~\ref{lem:relative-margin} 
are a fork-building strategy for the string $xy$ 
and Propositions~\ref{prop:muxy0-lowerbound-adv} 
and~\ref{prop:muxy0-upperbound}. 
Specifically, 
recall equation~\eqref{eq:mu-relative-recursive}. 
The first proposition shows that the fork $F \Fork xy0$ 
built by the said strategy achieves 
$\mu_x(F) \geq \mu_x(y0)$ while 
the second proposition shows that this value, in fact, is the largest possible, 
i.e., $\mu_x(y0) \leq \mu_x(y0)$.
In addition, any fork-building strategy 
whose forks satisfy the premise of Proposition~\ref{prop:muxy0-lowerbound-adv} 
can be used to prove Lemma~\ref{lem:relative-margin}. 
% \begin{remark*}  
%   Since the statement of Proposition~\ref{prop:muxy0-upperbound} 
%   concerns only a specific value of the relative margin, 
%   its statement applies to any fork for $xy0$. 
%   In addition, observe that any fork-building strategy 
%   whose forks satisfy Proposition~\ref{prop:muxy0-lowerbound-adv} 
%   can be used to prove Lemma~\ref{lem:relative-margin}. 
%  \end{remark*}  

\subsection{A fork-building strategy to maximize \texorpdfstring{$x$-}{}relative margin}\label{sec:strategy-x}
Any fork $F \Fork xy$ contains two tines $t_x, t_\rho$ 
so that $\reach(t_\rho) = \rho(F), \reach_F(t_x) = \mu_x(F)$, 
and the tines $t_x, t_\rho$ are disjoint over the suffix $y$. 
We say that the tine-pair $(t_\rho, t_x)$ 
is a \emph{witness} to $\mu_x(F)$.

% \paragraph{The strategy.}
Let $x,y \in \{0,1\}^*$ and write $w = xy$. 
Recursively build closed forks $F_0, F_1, \ldots, F_{|w|}$ 
where $F_i \Fork w_1 \ldots w_i, i \geq 1$ and 
$F_0 \vdash \varepsilon$ is the trivial fork 
consisting of a single vertex 
corresponding to the genesis block. 
For $i = 0, 1, \ldots, |w| - 1$ in increasing order, do as follows.  
If $w_{i + 1} = 1$, set $F_{i+1} \leftarrow F_{i}$. 
If $w_{i+1}=0$, 
set $F_{i + 1} \Fork w0$ 
as a conservative extension of $F_i \Fork w$ so that 
$\sigma \in F_{i+1}, \ell(\sigma) = i + 1$ 
is a conservative extension of a tine $s \in F_i$ 
identified as follows. 
If $F_i$ contains no zero-reach tine, 
$s$ is the unique longest tine in $F_i$. 
Otherwise, 
first identify a maximal-reach tine $t_\rho \in F_i$ as follows:
if $i \geq |x| + 1$, $t_\rho$ is a maximal-reach tine in $F_i$ 
which belongs to a tine-pair witnessing $\mu_x(F_i)$; 
otherwise, $t_\rho$ can be an arbitrary maximal-reach tine in $F_i$. 
Finally, $s$ is the zero-reach tine in $F_i$ 
that diverges earliest from $t_\rho$. 
If there are multiple candidates for $s$ or $t_\rho$, break tie arbitrarily.

% Note that we have enough reserve to extend any zero-reach $F_i$-tine 
% to make it as long as the longest tine in $F_i$.
% However, if a zero-reach tine is not extended now, 
% it can never be extended later (see Claim~\ref{claim:nex}). 
% Finally, observe that there is always a unique longest (honest) tine in a closed fork.

%-------------------------- Relative margin lower bound via Adversary ------------
\begin{proposition}\label{prop:muxy0-lowerbound-adv}
  Let $x, y$ be arbitrary Boolean strings, $|y| \geq 1$ and $w = xy$. 
  Let $F \Fork w$ and $F^\prime \Fork w0$ be two closed forks 
  built by the strategy 
  above 
  % in Section~\ref{sec:strategy-x} 
  so that $F \ForkPrefix F^\prime$ and suppose, in addition,
  that $\rho(F) = \rho(xy)$ and $\mu_x(F) = \mu_x(y)$. 
  Then $\rho(F^\prime) = \rho(xy0)$ and 
  $\mu_x(F^\prime) \geq \mu_x(y0)$.
\end{proposition}

\subsection{Proof of Proposition~\ref{prop:muxy0-lowerbound-adv}}
  Before we proceed further, let us record two useful results related to conservative extensions 
  and closed fork prefixes. 

  \begin{claim}[A conservative extension has reach zero]\label{claim:ex}
    Consider closed forks $F\vdash w, F'\vdash w0$ 
    such that $F\fprefix F'$. 
    If a tine $t$ of $F'$ is a conservative extension 
    then $\reach_{F'}(t)=0$.
  \end{claim}
  \begin{proof}
    We have assumed that $t$ is a conservative extension, 
    so its terminal vertex must be the new honest node. 
    By definition, 
    $\reach_{F'}(t)=\text{reserve}_{F'}(t)-\text{gap}_{F'}(t)$. 
    Honest players will only place nodes 
    at a depth strictly greater than all other honest nodes, 
    so we infer that $t$ is the longest tine of $F'$, 
    and so $\text{gap}_{F'}(t)=0$. 
    Moreover, we observe that 
    there are no 1s occurring after this point 
    in the characteristic string, 
    and so $\text{reserve}_{F'}(t)=0$. 
    Plugging these values into 
    our definition of $\reach$ 
    we see that $\reach_{F'}(t)=0-0=0$. 
  \end{proof}

  %Intuitively, a tine that arises by extension in an honest slot must be the longest tine of the fork, because honest players will only extend a chain with maximum length. Moreover, there are no dishonest slots after the final honest slot, so the remaining reserve is 0. Therefore, reach is exactly 0.

  \begin{claim}[Reach of non-extended tines]\label{claim:nex}
    Consider a closed fork $F\vdash w$ and some closed fork $F'\vdash w0$ such that $F\fprefix F'$. If $t \in F$ then 
    $\reach_{F'}(t)\leq \reach_{F}(t) - 1$. 
    The inequality becomes and equality 
    if $F'$ is obtained via a conservative extension from $F$.
  \end{claim}
  \begin{proof}
    Definitionally, we know that $\reach_{F'}(t)=\text{reserve}_{F'}(t)-\text{gap}_{F'}(t).$ From $F$ to $F'$, the length of the longest tine increases by at least one, and the length of $t$ does not change, so we observe that $\gap_{F'}(t) \geq \gap_{F}(t) + 1$ 
    with equality only for conservative extensions. 
    The reserve of $t$ does not change, because there are no new 1s in the characteristic string. Therefore, 
    $
      \reach_{F'}(t)
      =\text{reserve}_{F'}(t)-\text{gap}_{F'}(t)
      \leq \text{reserve}_{F}(t)-\text{gap}_{F}(t) - 1
      =\reach_{F}(t) - 1
      % \,.
      % \qedhere
    $. 
    % $\qed$
  \end{proof}

  % Now we are ready do prove Proposition~\ref{prop:muxy0-lowerbound-adv}.
  % \begin{proof}[Proof of Proposition~\ref{prop:muxy0-lowerbound-adv}]
    Assume the premise of Proposition~\ref{prop:muxy0-lowerbound-adv}.
    That is, $F$ is a fork for $xy$ so that 
    $\rho(F) = \rho(xy), \mu_x(F) = \mu_x(y)$, 
    and the tine $t_\rho$ identified by the fork-building strategy in Section~\ref{sec:strategy-x} 
    belongs to an $F$-tine-pair $(t_\rho, t_x)$ that witnesses $\mu_x(F)$. 
    To recap, this means 
    $\reach_F(t_\rho) = \rho(F) = \rho(x)$, 
    $\reach_F(t_x) = \mu_x(F) = \mu_x(y)$,  
    and the tines $t_\rho, t_x$ are disjoint over $y$ 
    (i.e., $\ell(t_\rho \Intersect t_x) \leq |x|$). 
    In addition, since $\sigma \in F'$ is a conservative extension of $s$, 
    we have $\reach_{F'}(\sigma) = 0$.
    Finally, let $S$ be the set of all zero-reach tines in $F$.
    
    We will break this part of the proof into several cases 
    based on the relative reach and margin of the fork. 

    \paragraph{Case 1: $\rho(xy) > 0$ and $\mu_x(y)=0$.} 
    We wish to show that 
    $\rho(F') = \rho(xy0)$ and 
    $\mu_x(F') \geq 0$. 
    Since $\rho(F) > 0$, $s \neq t_\rho$ and therefore, 
    By~\eqref{eq:rho-recursive} and Claim~\ref{claim:nex}, 
    Thus 
    $\rho(F^\prime) 
    \geq \reach_{F'}(t_\rho) 
    = \reach_F(t_\rho) - 1 
    = \rho(xy) - 1 
    = \rho(xy0)
    $. Therefore, 
    $\rho(F^\prime) = \rho(xy0)$.

    Since $\mu_x(y)=0$, 
    $t_x$ is a candidate for being selected as $s$ and hence 
    $\ell(s \Intersect t_\rho) \leq \ell(t_x \Intersect t_\rho) \leq |x|$. 
    Thus $\sigma, t_\rho \in F'$ are disjoint over $y0$ 
    and, therefore, $\mu_x(F') \geq \reach_{F'}(\sigma) = 0$.

    \paragraph{Case 2: $\rho(xy)=0$.}
    We wish to show that 
    $\rho(F') = \rho(xy0)$ and 
    $\mu_x(F') \geq \mu_x(y) - 1$. 
    Since there is at least one zero-reach tine, $\reach_F(s) = 0$ 
    and, in addition, $t_\rho \in S, |S| \geq 1$.
    Since $\reach_{F'}(\sigma) = 0 = \rho(xy0)$ by~\eqref{eq:rho-recursive}, 
    $\sigma$ has the maximal reach in $F'$ and, 
    in particular, $\rho(F') = \rho(xy0)$.
    Depending on $S$ and $s$, there are three possibilities. 
      If $s = t_\rho$, 
      this means $S=\{t_\rho\}$, 
      $t_x$'s $F'$-reach is one less than 
      its $F$-reach, 
      and $\sigma, t_x$ are still disjoint over $y0$.   
      Hence $\mu_x(F') \geq \reach_F(t_x) - 1 = \mu_x(y)-1$.       
      If $s = t_x$, 
      then 
      $t_\rho$'s $F'$-reach is one less than 
      its $F$-reach 
      and $\sigma, t_\rho$ are disjoint over $y0$. 
      Hence $\mu_x(F') \geq \reach_F(t_\rho) - 1 
      = \rho(xy) -1 \geq \mu_x(y)-1$.
      Finally, suppose $s \neq t_\rho$ and $s \neq t_x$. 
      Then $\mu_x(y) = \reach_F(t_x) < 0$ and, in addition, 
      $s$ (and $\sigma$) must share an edge with $t_\rho$ somewhere over $y$ since otherwise, 
      we would have achieved $\mu_x(y)=0$. 
      As a result, $t_x$ and $\sigma$ must be disjoint over $y0$. 
      Hence $\mu_x(F') \geq \reach_{F'}(t_x) = 
      \reach_F(t_x) - 1 = \mu_x(y) - 1$.

    \paragraph{Case 3: $\rho(xy) > 0,\mu_x(y)\neq 0$.}
    We wish to show that 
    $\rho(F') = \rho(xy0)$ and 
    $\mu_x(F') \geq \mu_x(y) - 1$.
    In this case, $s \neq t_\rho$ and $s \neq t_x$ and therefore, 
    $\reach_{F'}(t_i)=\reach_{F}(t_i)-1$ for $i = {1,2}$. 
    The tines $t_\rho, t_x$ are still disjoint over $y0$. 
    In addition, $t_\rho$ will still have the maximal reach in $F'$ 
    since $\reach_{F'}(t_\rho) = \rho(xy) - 1 = \rho(xy0)$ by~\ref{eq:rho-recursive}. 
    Therefore, 
    $\rho(F') = \rho(xy0)$ and, in addition, 
    $\mu_x(F') \geq \reach_{F'}(t_x) = \reach_{F}(t_x) - 1 = \mu_x(y)-1$.

  % \end{proof}
  This complete the proof of Proposition~\ref{prop:muxy0-lowerbound-adv}.
  \hfill $\qed$

\subsection{Proof of Lemma~\ref{lem:relative-margin}}

% \begin{proof}[Proof of Lemma~\ref{lem:relative-margin}]
Let $F$ be a closed fork for the characteristic string $xy$. 
Let $t_\rho, t_x \in F$ be the two tines that witness $\mu_x(F)$, 
i.e., $\reach(t_\rho) = \rho(F), \reach_F(t_x) = \mu_x(F)$, 
and $t_\rho, t_x$ are disjoint over $y$. 
Let $\hat{t}$ be the longest tine in $F$.

In the base case, where $y=\varepsilon$, we observe that any two tines of $F$ are disjoint over $y$. Moreover, even a single tine $t_\rho$ is disjoint with itself over $\varepsilon$. Therefore, the relative margin $\mu_x(\varepsilon)$ must be greater than or equal to the reach of the tine $t$ that achieves $\text{reach}(t)=\rho(x)$. The relative margin must also be less than or equal to $\rho(x)$, because that is, by definition, the maximum reach over all tines in all forks $F\vdash w$. Putting these facts together, we have $\mu_x(\varepsilon)=\rho(x)$.

Moving beyond the base case, we will consider a pair of closed forks $F\vdash xy$ and $F'\vdash xyb$ 
such that $F \fprefix F'$, $x,y\in\{0,1\}^*$, $|y| \geq 1$, and $b \in\{0,1\}$. 
If $b=1$, we have set $F' = F$. The reach of each tine increases by 1 from $F$ to $F'$ 
since the gap has not changed but the reserve has increased by one. Therefore, $\mu_x(y1) = \mu_x(y)+1$, as desired.

If $b=0$, however, things are more nuanced. 
% Considering the equality in~\ref{eq:mu-relative-recursive}, 
% Proposition~\ref{prop:muxy0-lowerbound-adv} proves that the left-hand side is at least the right-hand side, 
% and Proposition~\ref{prop:muxy0-upperbound} proves that the left-hand side is at most the right-hand side. 
Consider the following proposition:

\begin{proposition}\label{prop:muxy0-upperbound}
  Let $x, y$ be arbitrary Boolean strings, $|y| \geq 1$, and $w = xy0$. 
  Then $\mu_x(y0) \leq 0$ if $\rho(xy) > \mu_x(y) = 0$, and 
  $\mu_x(y0) \leq \mu_x(y) - 1$ otherwise.
\end{proposition}
Recall that $\mu_x(F^\prime) \geq \mu_x(y0)$ by Proposition~\ref{prop:muxy0-lowerbound-adv}. 
Combining this with Proposition~\ref{prop:muxy0-upperbound} above, we conclude 
that $\mu_x(F^\prime) = \mu_x(y0)$ and, in addition, that 
the fork $F^\prime$ actually achieves the maximum reach and 
the maximum relative margin 
for the characteristic string $xy0$. 
It remains to prove Proposition~\ref{prop:muxy0-upperbound}.

\begin{proof}[Proof of Proposition~\ref{prop:muxy0-upperbound}]
  Suppose $F'\vdash xy0$ is a closed fork such that 
  $\rho(xy0)=\rho(F')$ and $\mu_x(y0)=\mu_x(F')$. 
  Let $t_\rho, t_x \in F'$ to be a pair of tines disjoint over $y$ in $F'$ such that $\reach_{F'}(t_\rho)=\rho(F')$ and $\reach_{F'}(t_x)=\mu_x(F')=\mu_x(y0)$. 
  Let $F\vdash xy$ be the unique closed fork such that $F\fprefix F'$.  
  Note that while $F'$ is an extension of $F$, 
  it is not necessarily a conservative extension.

  \paragraph{Case 1: $\rho(xy)>0$ and $\mu_x(y)=0$.} 
    We wish to show that $\mu_x(y0) \leq 0$.
    Suppose (toward a contradiction) that $\mu_x(y0) > 0$. 
    Then neither $t_\rho$ or $t_x$ is a conservative extension because, as we proved in Claim ~\ref{claim:ex}, conservative extensions have reach exactly 0. This means that $t_\rho$ and $t_x$ existed in $F$, and had strictly greater reach in $F$ than they do presently in $F'$ (by Claim ~\ref{claim:nex}). 
    Because $t_\rho$ and $t_x$ 
    % have been implicitly 
    are 
    disjoint over $y0$, they must also be disjoint over $y$; therefore the $\mu_x(F)$ must be at least $\min\{\reach_F(t_\rho),\reach_F(t_x)\}$. 
    Following this line of reasoning, we have 
    $0 
    = \mu_x(y) 
    \geq \min_{i \in \{1,2\}}\{\reach_F(t_i)\}
    > \min_{i \in \{1,2\}}\{\reach_{F'}(t_i)\}
    = \mu_x(F') = \mu_x(y0) > 0
    $, a contradiction, as desired.

  \paragraph{Case 2: $\rho(xy)=0$.}
    We wish to show that $\mu_x(y0) \leq \mu_x(y) - 1$
    or, equivalently, that $\mu_x(y0) < \mu_x(y)$. 
    First, we claim that $t_\rho$ must arise from an extension. 
    Suppose, toward a contradiction, that $t_\rho$ is not an extension, 
    i.e., $t_\rho \in F$. 
    The fact that $t_\rho$ achieves the maximum reach in $F'$ 
    implies that 
    $t_\rho$ has non-negative reach 
    since the longest honest tine always achieves reach 0. 
    Furthermore, 
    Claim ~\ref{claim:nex} states that 
    all tines other than the extended tine see their reach decrease. 
    Therefore, $t_\rho \in F$ must have had a strictly positive reach. 
    But this contradicts the central assumption of the case, i.e., 
    that $\rho(xy)=0$. 
    Therefore, we conclude that $t_\rho \in F', t_\rho \not \in F$, and, 
    since $F'$ differs from $F$ by a single extension, 
    $t_x \in F$.

    Let $s \in F$ be the tine-prefix of $t_\rho \in F'$ so that 
    $t_\rho$ is an extension of $s$. 
    Since $\reach_{F'}(t_\rho) = \rho(xy0) = 0$ by~\eqref{eq:rho-recursive}, 
    $\reach_F(s)$ must be at least 0. 
    Additionally, since $\rho(xy)=0$, $\reach_F(s) \leq 0$. 
    Together, these statements tell us that $\reach_F(s)=0$. 
    Restricting our view to $F$, we see that 
    $s$ and $t_x$ are disjoint over $y$ and 
    so it must be true that 
    $\min\{\reach_F(s),\reach_F(t_x)\} \leq \mu_x(y)$. 
    Because $\reach_F(s)=0$ and $\reach_F(t_x) \leq \rho(xy)=0$, we can simplify that statement to $\reach_F(t_x) \leq \mu_x(y)$. 
    Finally, since $t_x \in F$, 
    Claim ~\ref{claim:nex} tells us that 
    $\reach_{F'}(t_x) < \reach_F(t_x)$. 
    Taken together, these two inequalities show that 
    $\mu_x(y0) = \reach_{F'}(t_x) < \reach_F(t_x) \leq \mu_x(y)$.

  \paragraph{Case 3: $\rho(xy)>0,\mu_x(y)\neq0$.}
    We wish to show that $\mu_x(y0) \leq \mu_x(y) - 1$ 
    or, equivalently, that $\mu_x(y0) < \mu_x(y)$. 
    Note that by~\ref{eq:rho-recursive}, 
    $\rho(xy0) = \rho(xy) - 1 \geq 0$.
    We will break this case into two sub-cases. 
    \begin{description}[font=\normalfont\itshape\space]
      \item[If both $t_\rho, t_x \in F$.] 
      Then  $t_\rho, t_x \in F$ and, consequently, 
      $\min\{\reach_F(t_\rho),\reach_F(t_x)\} \leq \mu_x(y)$ since $t_\rho$ and $t_x$ must be disjoint over $y$. 
      Furthermore, by Claim~\ref{claim:nex}, 
      $\reach_{F'}(t_i)<\reach_F(t_i)$ for $i\in\{1,2\}$. 
      Therefore, 
      $\mu_x(y0) 
      = \reach_{F'}(t_x) 
      = \min\{\reach_{F'}(t_\rho),\reach_{F'}(t_x)\} 
      < \min\{\reach_{F}(t_\rho),\reach_{F}(t_x)\} 
      \leq\mu_x(y)$, as desired. 

      \item[If either $t_\rho \not \in F$ or $t_x \not \in F$.]
      It must be true that $\reach_{F'}(t_x)\leq 0$, because either $t_x$ is the extension (and therefore has reach exactly 0) or $t_\rho$ is the extension and we have $\reach_{F'}(t_x)=\mu_x(y0)\leq\rho(xy0)=\reach_{F'}(t_\rho)=0$. Recall that we have assumed $\mu_x(y)\neq0$. If $\mu_x(y)>0,$ we are done: certainly $\mu_x(y0)\leq0<\mu_x(y)$. If, however, $\mu_x(y)<0$, there is more work to do. 
      In this case, we claim that $t_x \in F$, i.e., $t_x$ did not arise from an extension. 
      To see why, consider the following: if $t_x$ arose from extension, then there must be some $s \in F$ 
      so that $s \Prefix t_x$ and $\reach_F(s) \geq 0$. Additionally, by our claim about non-extended tines, we see that 
      $\reach_F(t_\rho)>\reach_{F'}(t_\rho) = \rho(xy0) \geq 0$. 
      Therefore, 
      $\mu_x(y) \geq \min\{\reach_F(t_\rho), \reach_F(s)\} \geq 0$, 
      contradicting our assumption that $\mu_x(y) < 0$. 
      Thus $t_x \in F$. 

      The only remaining scenario is the one in which 
      $\mu_x(y)<0$ and $t_\rho$ arises from an extension 
      of some tine $s \in F, \reach_F(s) \geq 0$. 
      In this scenario, $t_x$ cannot have been the extension 
      (since there is only one). 
      By Claim~\ref{claim:nex}, 
      $\reach_F(t_x) > \reach_{F'}(t_x)$. 
      Using a now-familiar line of reasoning, note that 
      the two tines
      $t_x$ and $s$ are disjoint over $y$ 
      and, therefore, 
      $\mu_x(y) \geq \min\{\reach_F(s), \reach_F(t_x)\}$. 
      Since, 
      $\mu_x(y) < 0$ by assumption and $\reach_F(s) \geq 0$, 
      it follows that 
      $\mu_x(y) \geq \reach_F(t_x) > \reach_{F'}(t_x)=\mu_x(y0)$, as desired. \qedhere
    \end{description}
\end{proof}

This completes the proof of Lemma~\ref{lem:relative-margin}. 
\hfill $\qed$

% In fact, by maximizing relative margin, she can play the $(\Distribution,T;s,k)$-settlement game optimally (i.e., win whenever there exists a winning configuration for the given challenge characteristic string). For the settlement game to be winnable, there must exist some $\hat{y}$, a prefix of $y$, so that $|\hat{y}| \geq k + 1$ and $\mu_x(\hat{y}) \geq 0$. Recall from Fact~\ref{fact:margin-balance} that for a characteristic string $xy$, there is an $x$-balanced fork $F \vdash xy$ if and only if $\mu_x(y) \geq 0$.  Because in each slot, our adversary builds a fork that achieves the maximum value of relative margin for all possible decompositions, then at slot $|x\hat{y}|$ she will have built a fork $F\vdash x\hat{y}$ such that $\mu_x(F)\geq 0$. As we argued in the proof of Fact\ref{fact:margin-balance}, the definition of relative margin tells us that $F$ has two tines $t_\rho$ and $t_2$ with nonnegative reach that diverge prior to the start of $\hat{y}$. Consequently, she is able to append $\gap(t_i)$ adversarial vertices from our reserve to each $t_i$ so that they become maximum length, thus winning the $(\Distribution,T;s,k)$-settlement game.

% In the next section, we present a stronger strategy 
% for building a fork $F \Fork w$ which 
% achieves $\mu_x(F) = \mu_x(y)$ \emph{simultaneously for all 
% decompositions} $w = xy$. 

%%% Local Variables:
%%% mode: latex
%%% TeX-master: "main"
%%% End:

\section{Canonical forks and an optimal online adversary}
\label{sec:canonical-forks}

Let $w$ be a characteristic string, written $w = xy$, 
and recall the online fork-building strategy from Section~\ref{sec:strategy-x}. 
In Proposition~\ref{prop:muxy0-lowerbound-adv}, 
we showed that the fork produced by this strategy (for the string $w$) 
always contains a tine-pair $(t_\rho, t_x)$ that witnesses $\mu_x(y)$. 
In this section, we present an online fork-building strategy 
which produces a fork that \emph{simultaneously} contains, 
for every prefix $x \PrefixEq w$, 
a tine-pair that witnesses $\mu_x(y)$. 
These forks are called \emph{canonical forks}, defined below.
\begin{definition}[Canonical forks]
	Let $w_1 \ldots w_T \in \{0,1\}^T$. 
	For $n = 0, 1, \ldots, T$, a \emph{canonical fork $F_n$ for $w = w_1\ldots w_n$} 
	is inductively defined as follows. 
	If $n = 0$ then $F_0$ is the trivial fork for the empty string; 
	it consists of a single (honest) vertex and no edge. 
	If $n \geq 1$, the following holds: 
	$F_n$ is a closed fork so that $F_{n-1} \ForkPrefix F_n$. 
	$F_n$ contains an honest tine $\tau_\rho$ so that 
	$\reach(\tau_\rho) = \rho(F_n) = \rho(w)$. 
	For every decomposition $w = xy, x \Prefix w$, 
	% $\mu_x(F_n) = \mu_x(y)$ and, in addition, 
	$F_n$ contains two honest tines $\tau_x, \tau_{\rho x}$ 
	so that the tine-pair $(\tau_{\rho x}, \tau_x)$ witnesses $\mu_x(F_n) = \mu_x(y)$. 
	% It also contains a tine $\tau_w$ so that $\reach(\tau_w) = \mu_w(\varepsilon)$. 
	% The sequence $(F_n), n = 0, 1, \ldots, T$ is called a \emph{canonical sequence of forks for $w$}. 
	The (possibly non-distinct) designated tines $\tau_\rho, \tau_{\rho x}, \tau_x, x \Prefix w$ 
	are called the \emph{witness tines}.
\end{definition}
Note that if one's objective is to create a fork 
which contains many early-diverging tine-pairs witnessing large relative margins, 
a canonical fork is the best one can hope for.

\subsection{An online strategy for building canonical forks}
	Let $w$ be a characteristic string, 
	written as $w = xy$, and 
	let $F$ be a fork for $w$. 
	If the tines $t_1, t_2 \in F$ are disjoint over $y$, 
	we say \emph{$t_1$ and $t_2$ are $y$-disjoint}, or equivalently, 
	\emph{$t_1$ is $y$-disjoint with $t_2$}. 
	Note that this means $\ell(t_1 \Intersect t_2) \leq |x|$. 
	Let $\leq_\pi$ be the lexicographical ordering of the tines where 
	each tine is represented as the list of vertex labels appearing in the tine's root-to-leaf path.
	If two tines have the same vertex labels, 
	$\leq_\pi$ must break tie in an arbitrary but consistent way. 

	For a fixed fork, let $A, B$ be two sets of tines. 
	We define the \emph{early-divergence witness for $(A,B)$} as follows.  
	Let $C_{AB}$ be an ordered set of tine-pairs $(t'_a, t'_b), a' \in A, b' \in B$ 
	that minimize $\ell(t_a \Intersect t_b), t_a \in A, t_b \in B$. 
	The order of the elements in $C_{AB}$ is the following: 
	$(t_1, t_2) \leq (t'_1, t'_2)$ if and only if $t_1 \leq_\pi t'_1$ and $t_2 \leq_\pi t'_2$.
	The first element of $C_{AB}$ is called the early-divergence witness for $(A,B)$.
	% we say that a tine $t \in A$ \emph{diverges earliest with respect to $B$ with witness $t' \in B$} if 
	% $(t, t')$ attains the smallest value of $\ell(t_a, t_b)$ for all tines $t_a\in A, t_b\in B$ and, 
	% in addition, that $t, t'$ have the smallest $\leq_\pi$ ranks 
	% among all pairs that attain this minimum.
	% $t = \arg \min_{t_a \in A}\left\{ \min_{t_b \in B} \ell(t_a \Intersect t_b) \right\}$ 
	% and $t'$ 
	% has the smallest $\leq_\pi$-rank 
	% among all tines $t_b \in B$ 
	% so that $\ell(t \Intersect t_b)$ 
	% attains the minimum above. 

	The fork-building strategy $\Adversary^*$ presented in Figure~\ref{fig:adv-opt} 
	builds canonical forks in an online fashion, i.e., 
	it scans the characteristic string $w$ once, from left to right, 
	maintains a ``current fork,'' 
	and updates it after seeing each new symbol by only adding new vertices. 
	Since the final fork $F \Fork w$ is canonical, 
	it satisfies $\mu_x(F) = \mu_x(y)$ simultaeneously for all decompositions $w = xy$; 
	hence we call $\Adversary^*$ the \emph{optimal online adversary}.

	% Define $\Tines(F)$ as the set of all tines in $F$ and 
	% \[
	% 	\MaxReachTines_F(A) = \{t \in A \SuchThat \text{$\reach_F(t)$ is maximum over all $t \in A$}\}
	% 	\,.
	% \]
	% % When the fork $F$ is understood from the context, we omit the subscript $F$. 
	% We write $\MaxReachTines(F)$ to denote the set of tines with the largest reach in fork $F$. 

	% For $A, B \subseteq \Tines(F)$, define
	% \[
	% 	\EarliestDiverging_B(A) = \min_{\leq_\pi} \{a^* \SuchThat \text{$a^* \in A$ minimizes $\ell(a \Intersect b), a \in A, b \in B$}\}
	% 	\,.
	% \]
	% If $B$ has a single element $b$, we write $\EarliestDiverging_b(A)$.
	% % For any set $B \subseteq \Tines(F)$, define 
	% % \[
	% % 	\Disjoint_F(B,y) = \{ t \in F \SuchThat \text{$t$ is $y$-disjoint with every tine $b \in B$ }\}
	% % 	\,.
	% % \]
	% % When $B$ has a single element $b$, we overload the above notation and write $\Disjoint_F(b, y)$. 
	% Finally, for any tine $\tau \in F$, define 
	% \[
	% 	\Disjoint_F(\tau,y) = \{ t \in F \SuchThat \text{$t$ is $y$-disjoint with $\tau$ }\}
	% 	\,.
	% \]
	% % When $B$ has a single element $b$, we overload the above notation and write $\Disjoint_F(b, y)$. 

\begin{figure}[!h]
	\begin{center}
	  \fbox{
	    \begin{minipage}{.9 \textwidth}
	      \begin{center}
	        \textbf{The strategy $\Adversary^*$}
	      \end{center}
				Let $w = w_1 \ldots w_n \in \{0,1\}^n$ and $w_{n + 1} \in \{0, 1\}$. 
				If $n = 0$, set $F_0 \Fork \varepsilon$ as 
				the trivial fork comprising a single vertex. 
				Otherwise, for $n \geq 0$, 
				let $F_n$ be the closed fork 
				built recursively by $\Adversary^*$ for the string $w$. 
				If $w_{n + 1} = 1$, set $F_{n + 1} = F_n$.
				Otherwise, 
				the closed fork $F_{n + 1} \Fork w0$ 
				is the result of a 
				single conservative extension 
				of a tine $s \in F_n$ into a new honest tine 
				$\sigma \in F_{n+1}, \ell(\sigma) = n + 1$; 
				The tine $s$ can be identified as follows. 
				If $F_n$ contains no tine with reach zero, 
				$s$ is the unique longest tine in $F_n$. 
				Otherwise, 
				$s$ is the reach-zero tine 
				that diverges earliest with respect to 
				the set of maximal-reach tines in $F_n$. 
				If there are multiple candidates for $s$, 
				select the one with the smallest $\leq_\pi$-rank.
	      \begin{center}
	        \textbf{Designating the witness tines}
	      \end{center}
		% \paragraph{Designating the witness tines.}
			Writing $w^\prime = w w_{n + 1}, F = F_n$, and $F^\prime = F_{n + 1}$, 
			identify the tines 
			$
			\tau_\rho, \tau_w, 
			% , \tau_{w^\prime} 
			\tau_x, \tau_{\rho x} \in F^\prime, x \Prefix w
			$ as follows. 
			% If there are multiple candidates for any designated tine then 
			% select the one with the smallest $\leq_\pi$-rank. 
			Let $R$ (resp. $R'$) be the set of $F$-tines (resp. $F'$-tines) 
			with the maximal $F$-reach (resp. $F'$-reach). 
			Set $\tau_\rho$ as the element of $R'$ with smallest $\leq_\pi$-rank. 
			% Set $\tau_{w'}$ as the element of $R'$ that diverges earliest from $\tau_\rho$; 
			% if there are multiple candidates, select the one with the smallest $\leq_\pi$-rank.
			% Let $A$ be the set of $F$-tines that attain the reach $\max_{t \in F} \reach_{F'}(t)$. 
			% Set $\tau_{w}$ as the element in $A$ that diverges earliest with respect to $R'$ with witness $\tau_{\rho w} \in R'$; 
			% if there are multiple candidates, select the one with the smallest $\leq_\pi$-rank.
			Set $(\tau_{w}, \tau_{\rho w})$ as the early-divergence witness for $(R, R')$.
			For every decomposition $w = xy, |y| \geq 1, |x| \geq 0$, do as follows. 
			Let $B_x$ be the set of $F'$-tines that are $yw_{n+1}$-disjoint 
			with \emph{some} maximal-reach tine in $R'$. 
			Let $C_x \subseteq B_x$ contain the tines with the maximal $F'$-reach, 
			the maximum taken over $B_x$. 
			% that attain the reach $\max_{t \in B_x} \reach_{F'}(t)$. 
			% Set $\tau_{x}$ as the element in $C_x$ that diverges earliest with respect to $R'$; 
			% if there are multiple candidates, select the one with the smallest $\leq_\pi$-rank. 
			%%
			% Identify a tine-pair $(\tau_{x}, \tau_{\rho x})$ 
			% so that 
			% \[
			% 	(\tau_{x}, \tau_{\rho x}) = \arg\, \min_{(t,r') \SuchThat t \in C_x, r' \in R'} \ell(t \Intersect r') 
			% 	\,.
			% \]
			% If there are multiple pairs, break tie as follows: 
			% sort the pairs first by the $\leq_\pi$ rank of the first component and 
			% then the $\leq_\pi$ rank of the second component; 
			% take the first pair in this ordering. 
			Set $(\tau_{x}, \tau_{\rho x})$ as the early-divergence witness for $(C_x, R')$.

	    \end{minipage}
	  }
	\end{center}
	\caption{Optimal online adversary $\Adversary^*$}
	\label{fig:adv-opt}
 \end{figure}
	% \paragraph{The strategy $\Adversary^*$.}

	\begin{theorem}[$\Adversary^*$ builds canonical forks]\label{thm:canonical-fork}
		Let $w \in \{0,1\}^n$ and $b \in \{0,1\}$. 
		Let $F \Fork w$ and $F^\prime \Fork wb$ be two closed forks 
		built by the strategy $\Adversary^*$ 
		so that $F \ForkPrefix F^\prime$ and suppose, in addition, 
		that $F$ is canonical. 
		Then $F^\prime$ is canonical as well.
	\end{theorem}
	We remark that the fork-building strategy $\Adversary^*$ 
	would certainly satisfy Proposition~\ref{prop:muxy0-lowerbound-adv} and, therefore, 
	satisfy the recurrence relation~\eqref{eq:mu-relative-recursive} as well. 

\subsection{Winning the \texorpdfstring{$(\Distribution,T;s,k)$-}{}settlement game, optimally}\label{sec:adv-winning-settlement-game} 
	Consider the player in the $(\Distribution,T;s,k)$-settlement game 
	who, at the first step, samples 
	a characteristic string $w \sim \Distribution, w = w_1 w_2 \ldots w_T$. 
	Since the challenger is deterministic, 
	the game is completely determined by the characteristic string 
	and the choices of the player. 
	In particular, for a given prefix $x \Prefix w, |x| = s - 1$, 
	consider the decompositions $w = xyz$. 
	The player's chance of winning the game will be maximized if, 
	for every $y, |y| \geq k + 1$ (so that $n = |xy|\geq s + k$),
	the fork $F_n \Fork xy$ 
	contains a tine-pair $(\tau_{\rho x}, \tau_x)$ that witnesses $\mu_x(y)$. 
	In fact, 
	if $\mu_x(y) \geq 0$ for some $y$ then, 
	as shown in Fact~\ref{fact:margin-balance}, 
	the player wins the game by
	augmenting $F_n$ to an $x$-balanced fork $A_n \Fork xy$. 
	% In particular, he can do so 
	% by adding only $\gap(\tau_n)$ adversarial vertices to the tine $\tau_n \in F_n$. 
	% See the discussion in Section~\ref{sec:args-survey} and 
	% Observation~\ref{obs:settlement-balanced-fork} 
	% for the relationship between a balanced fork and the settlement violation.

	Note, in addition, that if $F_n$ is canonical, 
	the player can optimally play $(\Distribution, T; s, k)$-settlement games 
	\emph{simultaneously} for every $s \in [n - k]$. 
	That is, given a distribution $\Distribution$, 
	a canonical fork $F_n$ gives the player 
	the largest probability 
	of causing a settlement violation at as many slots $s \in [n - k]$ as possible, 
	at once.
	% Therefore, we shift our attention to characterizing an online strategy to build canonical forks.

\subsection{Proof of Theorem~\ref{thm:canonical-fork}}
	For convenience, let us record the following fact 
	which compacts Claims~\ref{claim:ex} and~\ref{claim:nex}.

	\begin{fact}\label{fact:reach-fork-ext}
		Let $F \Fork w$ and $F^\prime \Fork w0$ be closed forks so that 
		$F \ForkPrefix F^\prime$ and 
		$F^\prime$ differs from $F$ by a single conservative extension 
		$\sigma \in F^\prime, \ell(\sigma) = |w| + 1$.
		Then $\reach_{F^\prime}(t) = \reach_F(t) - 1$ for every $t \in F$ and,  
		in addition, $\reach_{F^\prime}(\sigma) = 0$.
	\end{fact}
	In the rest of the proof, we will frequently use 
	the above fact along with Lemma~\ref{lem:margin} and Lemma~\ref{lem:relative-margin}, 
	often without an explicit reference.

	By assumption, $F$ is a canonical fork. 
	Thus $\reach_F(t_\rho) = \rho(w)$ 
	% $\reach_F(t_w) = \mu_w(\varepsilon) = \rho(w)$ 
	and 
	for every prefix $x \Prefix w$, 
	$\reach_F(t_x) = \mu_x(y)$. 
	Let $w' = wb$ and let $\tau_\rho, \tau_w, \tau_{\rho w},
	% \tau_{w^\prime}, 
	\tau_x, \tau_{\rho x} \in F^\prime, x \Prefix w$ 
	be the purported witness tines in $F'$. 
	Note that $\tau_x$ must be 
	$yb$-disjoint with $\tau_{\rho x}$ by construction. 
	Similarly, $\tau_w$ must be 
	$w_{n+1}$-disjoint with $\tau_{\rho w}$ since 
	both cannot contain the unique vertex from slot $n + 1$. 
	It is evident from the construction that 
	$\rho(F') = \reach_{F'}(\tau_{\rho}) = \reach_{F'}(\tau_{\rho w}) = \reach_{F'}(\tau_{\rho x})$ 
	for $x \Prefix w$.
	Therefore, we wish to show that 
	$\reach_{F^\prime}(\tau_\rho) = \rho(w b)$, 
	$\reach_{F^\prime}(\tau_w) = \mu_{w}(b)$ 
	% $\reach_{F^\prime}(\tau_{w^\prime}) = \mu_{w^\prime}(\varepsilon)$, 
	and 
	$\reach_{F^\prime}(\tau_x) = \mu_x(y b)$ for $x \Prefix w$.	
	% It would be helpful to note that in any fork, 
	% the reach of a tine 
	% is no more than the reach of the last honest vertex on that tine. 

	\begin{description}[font=\normalfont\itshape\space]
		\item[If $b = 1$.]
			In this case, $F^\prime = F$ and $w^\prime = w 1$. 
			Examining the rule for assigning $\tau_\rho, \tau_x, \tau_{\rho x}$, and $\tau_w$, 
			we see that 
			$\tau_\rho = 
			% \tau_{w^\prime} 
			t_\rho$, 
			$\tau_w = t_\rho$, 
			$\tau_x = t_x$, and $\tau_{\rho x} = t_{\rho x}$ for all $x \Prefix w$. 
			Since $F^\prime = F$ and $b = 1$, 
			the $F^\prime$-reach of every $F$-tine is one plus its $F$-reach. 
			Thus for any $x, x \Prefix w$, writing $w^\prime = xy1$, we have 
			$\mu_x(y1) 
			= 1 + \mu_x(y) 
			= 1 + \reach_F(t_x) 
			= \reach_{F^\prime}(t_x)
			= \reach_{F^\prime}(\tau_x)$.
			Similarly, 
			$\rho(w1) 
			= 1 + \rho(w) 
			% = \rho(w^\prime)
			= \reach_{F^\prime}(t_\rho)
			= \reach_{F^\prime}(\tau_\rho)
			$. 
			By construction, 
			$\tau_w$ has the largest reach in $F$; 
			but this means $\reach_{F'}(\tau_w) = \reach_{F'}(t_\rho) = \rho(F') = \rho(w1)$ but, 
			on the other hand, $\mu_w(1) = 1 + \mu_w(\varepsilon) = 1 + \rho(w) = \rho(w1)$; 
			hence $\reach_{F^\prime}(\tau_w) = \mu_w(1)$. 
			% Finally, by examining the rule for assigning $\tau_{w'}$, 
			% $\reach_{F^\prime}(\tau_{w^\prime}) = \rho(F^\prime) = \rho(w^\prime) = \mu_{w^\prime}(\varepsilon)$. 

		\item[If $b = 0$.]
			The contingencies of this case 
			are covered by Propositions~\ref{prop:optadv-tau-rho},~\ref{prop:optadv-tau-mu-w}, 
			and~\ref{prop:optadv-tau-mu-x} below.

	\end{description}

	% To complete the proof, we only need to state and prove 
	% Propositions~\ref{prop:optadv-tau-rho},~\ref{prop:optadv-tau-mu-x}, and~\ref{prop:optadv-tau-mu-w}.

	% \begin{lemma}[Recursive description of relative margin]\label{lem:relative-margin}
	% \end{lemma}

	\begin{proposition}\label{prop:optadv-tau-rho}
		Assume the premise of Theorem~\ref{thm:canonical-fork} with $b = 0$.
		Then $F^\prime$ contains 
		% witness tines $\tau_\rho, \tau_{w0}$ 
		a witness tine $\tau_\rho$ 
		so that 
		$\reach_{F^\prime}(\tau_\rho) = \rho(w0)$.
		%  and 
		% $\reach_{F^\prime}(\tau_{w0}) = \mu_{w 0}(\varepsilon)$. 
	\end{proposition}
	\begin{proof}~
		Recall that the tine $\sigma \in F^\prime, \ell(\sigma) = |w| + 1$ 
		is a conservative extension to 
		a tine $s \in F, \reach_F(s) = 0$ 
		so that $\reach_{F^\prime}(\sigma) = 0$. 
		Also recall that $\mu_z(\varepsilon) = \rho(z)$ for any characteristic string $z$. 
		Finally, note that it suffices to show that 
		$\reach_{F^\prime}(\tau_\rho) \geq \rho(w0)$.
		%  and 
		% $\reach_{F^\prime}(\tau_{w0}) \geq \mu_{w 0}(\varepsilon)$. 

		% Consider the following contingencies based on $\rho(w)$.

		% \begin{description}[font=\normalfont\itshape\space]
			% \item[If $\rho(w) > 0$.]
				Suppose $\rho(w) > 0$. 
				Using Fact~\ref{fact:reach-fork-ext}, Lemma~\ref{lem:relative-margin}, 
				and examining the rule for assigning $\tau_\rho$, 
				we see that 
				$\reach_{F^\prime}(\tau_\rho) 
				\geq \reach_{F^\prime}(t_\rho) 
				= \reach_F(t_\rho) - 1 
				= \rho(w) - 1 
				 = \rho(w0)
				$. 
			% \item[If $\rho(w) = 0$.]
				On the other hand, if $\rho(w) = 0$ 
				then $\rho(w0)$ is zero as well. 
				It follows that 
				$\reach_{F^\prime}(\tau_\rho) 
				\geq \reach_{F^\prime}(\sigma) 
				= 0 = \rho(w0)
				$. 
		% \end{description}

		% Examining the rule for assigning $\tau_{w0}$, we have 
		% $\reach_{F^\prime}(\tau_{w0}) = \rho(F^\prime) 
		% = \rho(w0) = \mu_{w0}(\varepsilon)$. 	
	\end{proof}

	\begin{proposition}\label{prop:optadv-tau-mu-w}
		Assume the premise of Theorem~\ref{thm:canonical-fork} with $b = 0$.
		Then $F^\prime$ contains a tine-pair $(\tau_{\rho w}, \tau_w)$ 
		that witnesses $\mu_w(0)$.
		% so that 
		% $\reach_{F^\prime}(\tau_w) = \mu_w(0)$. 
	\end{proposition}
	\begin{proof}
		Recall that the tine $\sigma \in F^\prime, \ell(\sigma) = |w| + 1$ 
		is a conservative extension to 
		a tine $s \in F, \reach_F(s) = 0$ 
		so that $\reach_{F^\prime}(\sigma) = 0$. 
		In addition, since $F'$ contains a single vertex at slot $|w| + 1$, 
		$\tau_w$ and $\tau_{\rho w}$ are disjoint over the suffix $w_{n+1}$ 
		and, moreover, $\reach_{F'}(\tau_{\rho w}) = \rho(F') = \rho(w0)$ 
		by Proposition~\ref{prop:optadv-tau-rho}.
		Now consider the following contingencies based on $\rho(w)$.

		\begin{description}[font=\normalfont\itshape\space]
			\item[If $\rho(w) > 0$.]
				Thus $\mu_w(0) = \mu_w(\varepsilon) - 1 = \rho(w) - 1 = \rho(w0)$.
				There are two mutually exclusive scenarios 
				based on $\tau_{\rho w}$ and $\sigma$.
				If $\tau_{\rho w} = \sigma$ then, by construction, 
				$\tau_w \neq \sigma$ 
				(since $\ell(\tau_{\rho w}, \tau_w) \leq |w|$) and, 
				in addition, 
				% $\tau_w$ is the $F$-tine with the largest $F^\prime$-reach; 
				% by Fact~\ref{fact:reach-fork-ext}, 
				% $\tau_w$ must have the largest $F$-reach as well, i.e.,	
				$\reach_{F}(\tau_w) = \rho(w)$. 
				This implies 
				$\reach_{F^\prime}(\tau_w) = \reach_{F}(\tau_w) - 1 = \rho(w) - 1 = \mu_w(0)$. 
				On the other hand, if $\tau_{\rho w} \neq \sigma$ then $\tau_{\rho w} \in F$. 
				Since $\tau_w$ is the $F$-tine with the largest $F^\prime$-reach, 
				it follows that 
				$\reach_{F^\prime}(\tau_w) = \reach_{F^\prime}(\tau_{\rho w}) = \rho(w0) = \mu_w(0)$. 

			\item[If $\rho(w) = 0$.]
				Since $\rho(F) = \rho(w) = 0$, 
				Fact~\ref{fact:reach-fork-ext} tells us that 
				every $F$-tine must have a negative reach in $F^\prime$. 
				Since $\rho(F^\prime)$ is non-negative, 
				it must be the case that $\tau_{\rho w} = \sigma$. 
				We can reuse the argument from the subcase ``$\tau_{\rho w} = \sigma$'' 
				of the preceding case and conclude that $\reach_{F^\prime}(\tau_w) = \mu_w(0)$.
		\end{description}
	\end{proof}

	\begin{proposition}\label{prop:optadv-tau-mu-x}
		Assume the premise of Theorem~\ref{thm:canonical-fork} with $b = 0$.
		Let $x \Prefix w$ and write $w = xy$.
		Then $F^\prime$ contains a tine-pair $(\tau_{\rho x}, \tau_x)$ 
		that witnesses $\mu_x(y0)$.
		% $\reach_{F^\prime}(\tau_x) = \mu_x(y0)$. 
	\end{proposition}
	\begin{proof}~
		By construction, $\reach_{F^\prime}(\tau_x) = \mu_x(F')$ and, 
		by the definition of relative margin, 
		$\mu_x(F') \leq \mu_x(y0)$. 
		In light of~\eqref{eq:mu-relative-recursive}, 
		it suffices to show that 
		$\reach_{F^\prime}(\tau_x) \geq 0$ if $\rho(xy) > \mu_x(y) = 0$, 
		and $\reach_{F'}(\tau_x) \geq \mu_x(y) - 1$ otherwise.
		
		Let $R$ be the set of $F$-tines with the maximal $F$-reach and 
		let $R^\prime$ be the set of $F'$-tines with the maximal $F'$-reach; 
		thus $\tau_{\rho x} \in R'$.
		We know that $t_x$ is $y$-disjoint with $t_\rho$ in $F$. 
		Consider the following mutually exclusive cases.

		\begin{description}[font=\normalfont\itshape\space]
			\item[If $\rho(w) > 0$ and $\mu_x(y) = 0$.]
				In this case, $\mu_x(y0) = 0$ using Lemma~\ref{lem:relative-margin}. 
				Since $\reach_{F}(s) = 0 < \reach_{F}(t_{\rho x}) = \rho(w)$, 
				it follows that $s \neq t_{\rho x}$.
				In addition, observe that $t_{\rho x}$ must be in $R'$. 
				By our choice of $s$, 
				$\ell(s \Intersect t_{\rho x}) \leq \ell(t_x \Intersect t_{\rho x})$ 
				since $\reach_F(t_x) = \mu_x(y) = 0 = \reach_F(s)$.
				Since $t_x$ is $y$-disjoint with $t_{\rho x}$, so is $s$. 
				Recall that $\reach_{F^\prime}(\tau_x)$ is the largest among all tines 
				that are $y0$-disjoint with $\tau_{\rho x}$. 

				\begin{description}[font=\normalfont\itshape\space]
					\item[If $\tau_{\rho x} = t_{\rho x}$.]
						Thus $t_x$ is $y0$-disjoint with $\tau_{\rho x}$. 
						Since $\ell(\sigma) = |w| + 1$, 
						$\sigma$ must be $y0$-disjoint with $t_{\rho x} = \tau_{\rho x}$, 
						it follows that $\reach_{F^\prime}(\tau_x) \geq \reach_{F^\prime}(\sigma) = 0 = \mu_x(y0)$. 
						% Now, either $\ell(\tau_x) = n + 1$ or $\ell(\tau_x) \leq n$. 
						% In the former case, $\sigma = \tau_x$ 
						% and hence $\reach_{F^\prime}(\tau_x) = 0$. 
						% In the latter case, 
						% we have $\reach_{F^\prime}(\tau_x) \leq \reach_{F}(\tau_x) \leq \mu_x(y) = 0$. 
						% Here, the first inequality follows from Fact~\ref{fact:reach-fork-ext} 
						% and the second inquality follows since $\tau_x$ is $y$-disjoint with $\tau_\rho = t_\rho$. 
						% Hence $\reach_{F^\prime}(\tau_x) = 0 = \mu_x(y0)$.

					\item[If $\tau_{\rho x} \neq t_{\rho x}$.]
						This happens when $\rho(w) = 1, \rho(w0) = 0$, and $t_{\rho x}, \sigma \in R^\prime$. 
						Note that $|R^\prime| \geq 2$ since 
						both $\sigma, t_{\rho x} \in R^\prime$ but $\sigma \neq t_{\rho x}$. 
						If there are two $y0$-disjoint tines $r_1^\prime, r_2^\prime \in R^\prime$ 
						then $\reach_{F^\prime}(\tau_x) \geq 0 = \mu_x(y0)$. 
						% However, since $\rho(w0) = 0$, we actually have 
						% $\reach_{F^\prime}(\tau_x) = 0 = \mu_x(y0)$.
						Otherwise, all tines $r^\prime \in R^\prime$ share a vertex indexed by $y$. 
						Since $t_x$ is $y$-disjoint with $t_{\rho x}$, 
						$t_x$ must be $y$-disjoint (and thus $y0$-disjoint) with 
						every $r^\prime \in R^\prime$ as well. 
						Examining the rule for assigning $\tau_x$, 
						we conclude that $\tau_x = t_x$ and, therefore,  
						$\reach_{F^\prime}(\tau_x) = \reach_{F^\prime}(t_x) = \mu_x(y) = 0 = \mu_x(y0)$. 
			\end{description}

			\item[If $\rho(w) = 0$.]
				Let $x \Prefix w$ and note that 
				$\mu_x(y0) = \mu_x(y) - 1$. 
				Since $\rho(w) = 0$, $\reach_F(s) = 0$ 
				all $F$-tines will have a negative reach in $F^\prime$; 
				by Fact~\ref{fact:reach-fork-ext}, 
				$\sigma$ is the only tine in $F^\prime$ 
				with the maximal reach $\rho(F^\prime) = \rho(w0) = 0$, 
				i.e., $\tau_{\rho x} = \tau_\rho = \sigma$.
				In addition, we must also have $\reach_F(s) = 0$, i.e., $s \in R$; 
				we conclude that 
				$s$ has the smallest $\leq_\pi$ rank among all members of $R$ 
				and, therefore, $s = t_\rho$.
				It follows that 
				$\tau_x$ is $y0$-disjoint with $s = t_\rho$ and, in particular, $\tau_x \in F$. 
				Considering $t_x$, if it is $y$-disjoint with $t_\rho$ then 
				we must have $\tau_x = t_x$; 
				in this case, 
				$\reach_{F'}(\tau_x) = \reach_{F'}(t_x) = \reach_{F}(t_x) - 1 = \mu_x(y) - 1 = \mu_x(y0)$.
				Otherwise, $\ell(t_x \Intersect t_\rho) \geq |x| + 1$ 
				and there must be a tine $t_{\rho x} \in F$ that is $y$-disjoint with $t_x$ 
				(and hence, with $\tau_{\rho x}$). 
				Therefore, 
				$\reach_{F'}(\tau_x) \geq \reach_{F'}(t_{\rho x}) 
				\geq \reach_{F'}(t_x) = \reach_{F}(t_x) - 1 
				= \mu_x(y) - 1$. 
				Here, the first inequality follows from the construction of $\tau_x$ 
				and the second one follows since $t_{\rho x})$ has the maximal reach in $F$.

			\item[If $\rho(w) > 0$ and $\mu_x(y) \neq 0$.]
				% We know that $\reach_{F^\prime}(\tau_x) \leq \mu_x(y0) = \mu_x(y) - 1$. 
				% We wish to show that $\reach_{F^\prime}(\tau_x) \geq \mu_x(y0)$. 
				% % \textbf{What if $t_x = t_\rho$?}
				There can be two cases depending on whether $s$ has zero reach in $F$.

				\begin{description}[font=\normalfont\itshape\space]
					\item[If $\reach_F(s) = 0$.]
						Then $s \not \in \{t_{\rho x}, t_x\}$. 
						Observe that 
						$\reach_{F^\prime}(t_{\rho x}) = \reach_F(t_{\rho x}) - 1 = \rho(w) - 1 = \rho(w0)$. 
						It follwos that $t_{\rho x} \in R^\prime$. 
						Since $t_x$ is $y0$-disjoint with $t_{\rho x} \in R'$ and, in addition, that 
						$\tau_x$ has the largest reach among all tines 
						that are $y0$-disjoint with some member of $R^\prime$, 
						we conclude that 
						$\reach_{F^\prime}(\tau_x) \geq \reach_{F^\prime}(t_x) = \reach_{F}(t_x) - 1 = \mu_x(y) - 1 = \mu_x(y0)$. 

					\item[If $\reach_F(s) \geq 1$.]	
					% \item[If $\reach_F(s) \neq 0$ and $s = t_\mu_x$.]	
					% 	Thus $F$ contains no tine with reach zero and $s$ is the longest tine in $F$. 			
					% 	Suppose $s = t_x$. 
					% 	Then 
					% 	$\sigma$ will be $y0$-disjoint with $t_\rho$ in $F^\prime$ 
					% 	and, as before, 
					% 	$\reach_{F^\prime}(t_\rho) = \rho(w0)$. 
					% 	Since $\sigma$ and $t_\rho$ are $y0$-disjoint, 
					% 	a similar argument as in the preceding case would complete the claim.

					% \item[If $\reach_F(s) \neq 0, s = t_\rho \neq t_x$.]	
					% \item[If $\reach_F(s) \neq 0, s = t_\rho \neq t_x$, and $|R| \geq 2$.]	
					% 	There must be a tine $t^* \in R$ with the largest $F$-reach 
					% 	which is $y$-disjoint with $t_x$. 
					% 	Retracing our argument above, the claim will follow since 
					% 	these tines are $y0$-disjoint in $F^\prime$ as well. 

					% \item[If $\reach_F(s) \neq 0, s = t_\rho \neq t_x$, and $R = \{t_\rho\}$.]	
					% 	% If $\rho(w) = 1$ then $t_\rho \in R^\prime$ and, as a result, 
					% 	% $t_x$ will be $y0$-disjoint with $t_\rho$ in $F^\prime$ and the claim would follow. 
					% 	% Otherwise, suppose $\rho(w) \geq 2$. 
						In this case, $s$ is the longest tine in $F$.
						Considering fork $F^\prime$,
						if some tine $r^\prime \in R^\prime$ is $y0$-disjoint with $t_x$ 
						then $\reach_{F^\prime}(\tau_x) \geq \reach_{F'}(t_x) = \reach_{F}(t_x) - 1 = \mu_x(y) - 1 = \mu_x(y0)$. 
						Otherwise, $\ell(r^\prime \Intersect t_x) > |x|$ for every tine $r^\prime \in R^\prime$, 
						i.e., no maximal-reach $F^\prime$-tine is $y0$-disjoint with $t_x$. 
						Since $\ell(t_x, t_{\rho x}) \leq |x|$ by assumption 
						and $\tau_{\rho x} \in R^\prime$, it follows that 
						$\ell(\tau_{\rho x} \Intersect t_{\rho x}) \leq |x|$, i.e., 
						$t_{\rho x}$ is $y0$-disjoint with $\tau_{\rho x}$.
						Therefore, 
						$\reach_{F^\prime}(\tau_x) \geq \reach_{F^\prime}(t_{\rho x}) = \reach_{F}(t_{\rho x}) - 1 
						= \rho(w) - 1 \geq \mu_x(y) - 1 = \mu_x(y0)$. 
						Here, the second inequality is true since $\mu_x(y) \leq \rho(xy) = \rho(w)$.

				\end{description}

		\end{description}
		
	\end{proof}

	This completes the proof of Theorem~\ref{thm:canonical-fork}. 
	\hfill $\qed$

	In regards to the canonical fork $F \Fork w$ produced by the strategy $\Adversary^*$ 
	(see Figure~\ref{fig:adv-opt}), 
	it is possible to  
	maintain witness tines $\tau_\rho, \tau'_m \in F$, 
	for integers $m = -|w|, \ldots, |w|$, so that 
	for every prefix $x \Prefix w$, 
	the tine-pair $(\tau_\rho, \tau'_{\mu_x(y)})$ witnesses $\mu_x(y)$. 
	In particular, a single maxmimal-reach tine $\tau_\rho$ 
	appears in every witness tine-pair. We omit futher details.

\section*{Acknowledgments}
We are grateful to Shreyas Gandlur and Bruce Hajek (UIUC) 
for their suggestion about 
using the dominance argument in the proof of Bound~\ref{bound:geometric}. 

\bibliography{forks,abbrev0,crypto_crossref}

\appendix

\section{Exact settlement probabilities}
\label{sec:exact-prob}
Let $m, k \in \NN$  and $\epsilon \in (0,1]$. 
Let $w$ be a characteristic string of length $T = m + k$ such that 
the bits of $w$ are i.i.d.\ Bernoulli with expectation $\alpha = (1 - \epsilon)/2$. 
Write $w$ as $w = xy$ where $|x| = m, |y| = k$.
% where 
% each bit is independent with $\Pr[w_i = 1] = \alpha$.
The recursive definition of relative margin (cf. Lemma~\ref{lem:relative-margin}) 
implies an algorithm for computing the probability
$\Pr[\mu_x(y) \geq 0]$ in time $\Poly(m, k)$. In typical
circumstances, however, it is more interesting to establish an
explicit upper bound on $\Pr[\mu_x(y) \geq 0]$ where
$|x| \rightarrow \infty$; this corresponds to the case where the
distribution of the initial reach $\rho(x)$ is the dominant distribution
$\StationaryRho$ in Lemma~\ref{lemma:rho-stationary}. 
Due to dominance, $\StationaryRho(m)$ serves as an
upper bound on $\rho(x)$ for any finite $m = |x|$. 
For this purpose, one can implicitly
maintain a sequence of matrices $\left( M_t \right)$ for $t = 0, 1, 2, \cdots, k$
such that $M_0(r, r) = \StationaryRho(r)$ for all $0 \leq r \leq 2k$ and
the invariant
\[
  M_t(r, s) = \Pr_{y \sim \mathcal{B}(t, \alpha)}[\rho(xy) = r \text{ and }
  \mu_x(y) = s ]
\]
is satisfied for every integer $t \in [1, k]$,
$r \in [0, 2k]$, and $s \in [-2k, 2k]$. 
Here, $M(i,j)$ denotes the entry at the $i$th row and $j$th column of the matrix $M$.
%The matrix $M_0$ corresponds to the reach and margin after any prefix $x$. 
Observe that $M_t(r,s)$ can be computed solely from the neighboring cells of $M_{t-1}$, that is, 
from the values $M_{t-1}(r\pm 1, s \pm 1)$. 
Of course, only the transitions approved by the recursions in 
Lemma~\ref{lem:margin} and Lemma~\ref{lem:relative-margin} should be considered.

Finally, one can compute $\Pr[\mu_x(y) \geq 0]$ by summing $M_k(r,s)$ for
$r, s \geq 0$. Table~\ref{table:exact-probs} contains these
probabilities where $\alpha$ ranges from $0.05$ to $0.40$ and $k$
ranges from $50$ to $1000$. 
In addition, Figure~\ref{fig:exact-probs} shows the base-$10$ logarithm of
these probabilities. The
points corresponding to a fixed $\alpha$ appear to form a straight
line. This means the probability decays exponentially in $k$, or equivalently, that the exponent 
depends linearly on $k$, 
as stipulated by Bound~\ref{bound:analytic}. 

A \texttt{C++} implementation of the above algorithm is publicly available 
at~\href{https://github.com/saad0105050/forkable-strings-code}{https://github.com/saad0105050/forkable-strings-code}~\cite{PrForkableCode}.

\begin{table}[h]
	\centering
	\caption{
		Exact probabilities $\Pr[\mu_x(y) \geq 0]$ where 
		the bits of the characteristic string $xy$ are i.i.d.\ Bernoulli with expectation $\alpha$. 
		Each row of the table corresponds to a different $k = |y|$.
	} 
	\label{table:exact-probs}

	\begin{tabular}{|l||l|l|l|l|l|l|l|l|}
	\hline
	\multicolumn{1}{|c||}{\multirow{2}{*}{$k$}} & \multicolumn{8}{c|}{$\alpha$}                                             \\ \cline{2-9} 
	\multicolumn{1}{|c||}{}                   & 0.05     & 0.10      & 0.15     & 0.20      & 0.25     & 0.30      & 0.35     & 0.40      \\ 
	\hhline{|=#=|=|=|=|=|=|=|=|}
	50   & 5.37E-15  & 1.16E-09  & 1.02E-06  & 8.68E-05 & 1.96E-03 & 1.86E-02 & 9.36E-02 & 2.92E-01 \\ \hline
	100  & 1.23E-28  & 5.10E-18  & 3.52E-12  & 2.28E-08 & 1.03E-05 & 8.00E-04 & 1.72E-02 & 1.37E-01 \\ \hline
	150  & 2.83E-42  & 2.24E-26  & 1.22E-17  & 6.05E-12 & 5.54E-08 & 3.57E-05 & 3.30E-03 & 6.74E-02 \\ \hline
	200  & 6.49E-56  & 9.82E-35  & 4.21E-23  & 1.61E-15 & 2.98E-10 & 1.60E-06 & 6.40E-04 & 3.36E-02 \\ \hline
	250  & 1.49E-69  & 4.31E-43  & 1.46E-28  & 4.27E-19 & 1.61E-12 & 7.21E-08 & 1.25E-04 & 1.69E-02 \\ \hline
	300  & 3.42E-83  & 1.89E-51  & 5.05E-34  & 1.14E-22 & 8.67E-15 & 3.25E-09 & 2.44E-05 & 8.52E-03 \\ \hline
	350  & 7.84E-97  & 8.29E-60  & 1.75E-39  & 3.02E-26 & 4.67E-17 & 1.46E-10 & 4.78E-06 & 4.31E-03 \\ \hline
	400  & 1.80E-110 & 3.64E-68  & 6.06E-45  & 8.02E-30 & 2.52E-19 & 6.59E-12 & 9.37E-07 & 2.18E-03 \\ \hline
	450  & 4.13E-124 & 1.60E-76  & 2.10E-50  & 2.13E-33 & 1.36E-21 & 2.97E-13 & 1.84E-07 & 1.11E-03 \\ \hline
	500  & 9.47E-138 & 7.00E-85  & 7.26E-56  & 5.67E-37 & 7.32E-24 & 1.34E-14 & 3.60E-08 & 5.62E-04 \\ \hline
	550  & 2.17E-151 & 3.07E-93  & 2.51E-61  & 1.51E-40 & 3.95E-26 & 6.02E-16 & 7.05E-09 & 2.86E-04 \\ \hline
	600  & 4.98E-165 & 1.35E-101 & 8.70E-67  & 4.00E-44 & 2.13E-28 & 2.71E-17 & 1.38E-09 & 1.45E-04 \\ \hline
	650  & 1.14E-178 & 5.91E-110 & 3.01E-72  & 1.06E-47 & 1.15E-30 & 1.22E-18 & 2.71E-10 & 7.37E-05 \\ \hline
	700  & 2.62E-192 & 2.59E-118 & 1.04E-77  & 2.83E-51 & 6.19E-33 & 5.51E-20 & 5.31E-11 & 3.75E-05 \\ \hline
	750  & 6.02E-206 & 1.14E-126 & 3.61E-83  & 7.52E-55 & 3.33E-35 & 2.48E-21 & 1.04E-11 & 1.91E-05 \\ \hline
	800  & 1.38E-219 & 4.99E-135 & 1.25E-88  & 2.00E-58 & 1.80E-37 & 1.12E-22 & 2.04E-12 & 9.69E-06 \\ \hline
	850  & 3.17E-233 & 2.19E-143 & 4.33E-94  & 5.31E-62 & 9.69E-40 & 5.04E-24 & 4.00E-13 & 4.93E-06 \\ \hline
	900  & 7.27E-247 & 9.61E-152 & 1.50E-99  & 1.41E-65 & 5.23E-42 & 2.27E-25 & 7.84E-14 & 2.50E-06 \\ \hline
	950  & 1.67E-260 & 4.22E-160 & 5.19E-105 & 3.75E-69 & 2.82E-44 & 1.02E-26 & 1.54E-14 & 1.27E-06 \\ \hline
	1000 & 3.83E-274 & 1.85E-168 & 1.80E-110 & 9.98E-73 & 1.52E-46 & 4.61E-28 & 3.01E-15 & 6.48E-07 \\ \hline
	\end{tabular}

\end{table}

\begin{figure}[h!]

\centering
\begin{tikzpicture}
	\begin{axis}[
		xlabel=Length of $y$,
    ylabel=${\log_{10} \Pr[ \mu_x(y) \geq 0 ]}$,
    legend pos = outer north east,
    grid = both
    %,height = 7cm
    ]

%============= \alpha = 0.40    
	\addplot[color=blue,mark=triangle*] coordinates {
(50, -0.534881009996625)
(100, -0.861843437303938)
(150, -1.17110438398993)
(200, -1.47350880258986)
(250, -1.77251864117173)
(300, -2.06962896433787)
(350, -2.36559449154758)
(400, -2.66083495643829)
(450, -2.95559969289654)
(500, -3.25004400806221)
(550, -3.54426813580273)
(600, -3.83833855715903)
(650, -4.1323003235163)
(700, -4.42618449355664)
(750, -4.72001278023346)
(800, -5.01380053810684)
(850, -5.30755872931896)
(900, -5.60129524293208)
(950, -5.89501579493968)
(1000, -6.18872455077118)

	};
  
%============= \alpha = 0.35    
	\addplot[color=magenta,mark=o] coordinates {
(50, -1.02877436841824)
(100, -1.76405795532126)
(150, -2.48156081504761)
(200, -3.19351655688251)
(250, -3.90324719739329)
(300, -4.6119732073238)
(350, -5.32021396523135)
(400, -6.0282101484992)
(450, -6.73607948155516)
(500, -7.44388168223398)
(550, -8.15164782799613)
(600, -8.85939439475358)
(650, -9.56713023906372)
(700, -10.2748601722391)
(750, -10.9825868295165)
(800, -11.6903116636224)
(850, -12.3980354795589)
(900, -13.1057587252749)
(950, -13.8134816508858)
(1000, -14.521204396449)

	};

%============= \alpha = 0.30    
  	\addplot[color=darkgray,mark=square] coordinates {
(50, -1.73119438208806)
(100, -3.09673279902966)
(150, -4.44780375955355)
(200, -5.79553003490465)
(250, -7.14227457376805)
(300, -8.48869966525932)
(350, -9.83501495946096)
(400, -11.1812912676285)
(450, -12.5275534426098)
(500, -13.873810422115)
(550, -15.2200654730463)
(600, -16.5663198030557)
(650, -17.9125738621863)
(700, -19.258827819142)
(750, -20.6050817374433)
(800, -21.9513356410858)
(850, -23.2975895391646)
(900, -24.6438434351249)
(950, -25.9900973302771)
(1000, -27.3363512251228)

	};

%============= \alpha = 0.25    
	\addplot[color=purple,mark=diamond] coordinates {
(50, -2.7068067143968)
(100, -4.98669572210932)
(150, -7.25681899248793)
(200, -9.52547099796917)
(250, -11.793848036259)
(300, -14.0621691284583)
(350, -16.3304783495448)
(400, -18.5987849921223)
(450, -20.8670910667017)
(500, -23.1353970150404)
(550, -25.4037029351594)
(600, -27.6720088489442)
(650, -29.9403147613032)
(700, -32.2086206733433)
(750, -34.4769265853085)
(800, -36.7452324972576)
(850, -39.0135384092045)
(900, -41.281844321149)
(950, -43.5501502330935)
(1000, -45.8184561450381)

	};

%============= \alpha = 0.    
	\addplot[color=black,mark=x] coordinates {
(50, -4.06125884952046)
(100, -7.64248504871986)
(150, -11.2184152058828)
(200, -14.7939108371324)
(250, -18.3693634087138)
(300, -21.9448114175389)
(350, -25.520258928049)
(400, -29.095706383289)
(450, -32.6711538323473)
(500, -36.2466012807114)
(550, -39.822048728996)
(600, -43.3974961772738)
(650, -46.9729436255496)
(700, -50.5483910738248)
(750, -54.1238385221004)
(800, -57.6992859703754)
(850, -61.2747334186509)
(900, -64.8501808669272)
(950, -68.4256283152027)
(1000, -72.0010757634778)

	};
  
%============= \alpha = 0.15    
	\addplot[color=black,mark=diamond*] coordinates {
(50, -5.99101535930985)
(100, -11.4538646368115)
(150, -16.9145865673508)
(200, -22.3752357117723)
(250, -27.8358819361379)
(300, -33.2965280374903)
(350, -38.7571741335604)
(400, -44.2178202294009)
(450, -49.6784663252318)
(500, -55.1391124210625)
(550, -60.5997585168932)
(600, -66.0604046127235)
(650, -71.5210507085538)
(700, -76.9816968043839)
(750, -82.4423429002152)
(800, -87.9029889960452)
(850, -93.3636350918761)
(900, -98.8242811877078)
(950, -104.284927283538)
(1000, -109.745573379367)

};

%============= \alpha = 0.10    
	\addplot[color=red,mark=*] coordinates {
(50, -8.93400016060089)
(100, -17.292337808078)
(150, -25.6501551609344)
(200, -34.007967886374)
(250, -42.3657805652942)
(300, -50.7235932437328)
(350, -59.0814059221663)
(400, -67.4392186006004)
(450, -75.7970312790341)
(500, -84.1548439574669)
(550, -92.5126566359009)
(600, -100.870469314334)
(650, -109.228281992768)
(700, -117.586094671201)
(750, -125.943907349634)
(800, -134.301720028069)
(850, -142.659532706502)
(900, -151.017345384936)
(950, -159.375158063369)
(1000, -167.732970741802)

};
%============= \alpha = 0.05    
	\addplot[color=blue,mark=*] coordinates {
(50, -14.2699842004544)
(100, -27.9093384792564)
(150, -41.5486518624197)
(200, -55.1879652167995)
(250, -68.8272785711567)
(300, -82.4665919255127)
(350, -96.10590527987)
(400, -109.745218634226)
(450, -123.384531988584)
(500, -137.023845342941)
(550, -150.663158697297)
(600, -164.302472051655)
(650, -177.94178540601)
(700, -191.581098760368)
(750, -205.220412114725)
(800, -218.859725469083)
(850, -232.499038823439)
(900, -246.138352177796)
(950, -259.777665532154)
(1000, -273.41697888651)

	};
  \legend{
  $\alpha = 0.40$, $\alpha = 0.35$, 
   $\alpha = 0.30$,$\alpha = 0.25$, $\alpha = 0.20$, 
   $\alpha = 0.15$,$\alpha = 0.10$,  $\alpha = 0.05$
  }
	\end{axis}%
\end{tikzpicture}%
\caption{The probabilities from Table~\ref{table:exact-probs}  
drawn in the base-$10$ logarithmic scale. 
}
\label{fig:exact-probs}
\end{figure}
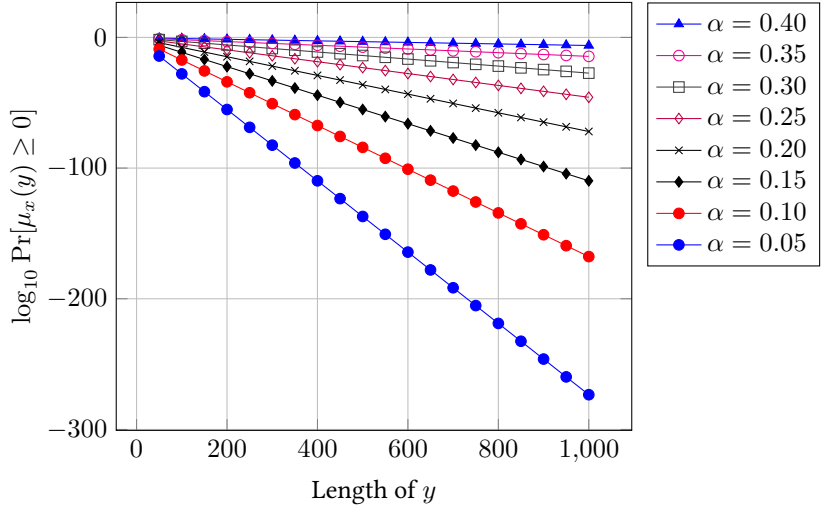

%%% Local Variables:
%%% mode: latex
%%% TeX-master: "main"
%%% End:

%\section{Stationary distribution for $\rho(x)$}
%\label{sec:dominance-app}
%\input{dominance-app}

\section{A forkability bound for strings satisfying the \texorpdfstring{$\epsilon$}{epsilon}-martingale condition}
\label{sec:martingale-proof}
%=======================================================
% \subsection{Proof of Bound~\ref{bound:geometric}}\label{sec:martingale-proof}
Below we present a bound (Bound~\ref{bound:geometric-original}) on the probability that 
a characteristic string satisfying the $\epsilon$-martingale condition has a non-negative relative margin. 
We remark that the bound below is weaker than Bound~\ref{bound:geometric}. 
Before we proceed, recall 
the following standard large deviation bound for supermartingales.
\begin{theorem}[Azuma's inequality (Azuma; Hoeffding). See {\cite[4.16]{Motwani:1995:RA:211390}} for a discussion]\label{thm:azuma}
  Let $X_0, \ldots, X_n$ be a sequence of real-valued random variables
  so that, for all $t$,
  $\Exp[X_{t+1} \mid X_0, \ldots, X_{t}] \leq X_t$ and
  $|X_{t+1} - X_t| \leq c$ for some constant $c$. Then  
  $
    \Pr[X_n - X_0 \geq \Lambda] \leq
    \exp\left(-{\Lambda^2}/{2nc^2}\right)
    %\,.
  $ 
  for every $\Lambda \geq 0$.
\end{theorem}

\begin{bound}\label{bound:geometric-original}
  Let $x \in \{0,1\}^m$ and $y \in \{0,1\}^k$ be random variables,
  satisfying the $\epsilon$-martingale condition (with respect to the ordering $x_1, \ldots, x_m, y_1, \ldots, y_k$). Then
  \[
    \Pr[\mu_x(y) \geq 0] \leq
    % \exp({-2\epsilon^4 (1 - O(\epsilon))n})
    3 \exp\left( -\epsilon^4 (1 - O(\epsilon) ) k/64 \right)  
    \, .
  \]
\end{bound}

%Now we are ready to prove Bound~\ref{bound:geometric-original}.

%\begin{proof}[of Bound~\ref{bound:geometric}]
\begin{proof}
  Let $w_1, w_2, \ldots$ be random variables obeying the
  $\epsilon$-martingale condition.  Specifically,
  $\Pr[w_t = 1 \mid E] \leq (1 - \epsilon)/2$ conditioned on any event
  $E$ expressed in the variables $w_1, \ldots, w_{t-1}$.  For
  convenience, define the associated $\{\pm1\}$-valued random
  variables $W_t = (-1)^{1+w_t}$ and observe that
  $\Exp[W_t] \leq -\epsilon$.

%\vspace{-2ex}
\paragraph{If $x$ is empty.}
Observe that in this case, the relative margin $\mu_x(y)$ reduces to 
the non-relative margin $\mu(y)$ from Lemma~\ref{lem:margin}. 
Since the sequence $y_1, y_2, \ldots$ in the statement of the claim 
is identical to the sequence $w_1, w_2, \ldots$ defined above, 
we focus on the reach and margin of the latter sequence. 
Specifically, define $\rho_t = \rho(w_1 \ldots w_t)$ and
$\mu_t = \mu(w_1 \ldots w_t)$ to be the two random variables from
Lemma~\ref{lem:margin} acting on the string $w=w_1 \ldots w_t$. The
analysis will rely on the ancillary random variables
$\overline{\mu}_t = \min(0,\mu_t)$.  Observe that $\Pr[\text{$w$
  forkable}] = \Pr[\mu(w) \geq 0] = \Pr[\overline{\mu}_k = 0]$, so we
may focus on the event that $\overline{\mu}_k = 0$. As an additional
preparatory step, define the constant
$\alpha = (1+\epsilon)/(2\epsilon) \geq 1$ and define the random
variables $\Phi_t \in \mathbb{R}$ by the inner product
  \[
    \Phi_t = (\rho_t, \overline{\mu}_t) \cdot
    \left(\begin{array}{c} 1\\ \alpha\end{array}\right) = \rho_t +
    \alpha \overline{\mu}_t\,.
  \]
  The $\Phi_t$ will act as a ``potential function'' in the analysis:
  we will establish that $\Phi_k < 0$ with high probability and,
  considering that
  $\alpha\overline{\mu}_k \leq \rho_k + \alpha \overline{\mu}_k =
  \Phi_k$, this implies $\overline{\mu}_k < 0$, as desired.
  
  Let $\Delta_t = \Phi_t - \Phi_{t-1}$; we claim that---conditioned
  on any fixed value $(\rho, \mu)$ for $(\rho_t, \mu_t)$---the
  random variable $\Delta_{t+1} \in [-(1 + \alpha),1+ \alpha]$ has
  expectation no more than $-\epsilon$. The analysis has four cases,
  depending on the various regimes of $\rho$ and $\mu$ from Lemma~\ref{lem:margin}. 
  When $\rho > 0$ and $\mu < 0$,
  $\rho_{t+1} = \rho + W_{t+1}$ and
  $\overline{\mu}_{t+1} = \overline{\mu} + W_{t+1}$, where
  $\overline{\mu} = \max(0,\mu)$; then
  $\Delta_{t+1} = (1 + \alpha)W_{t+1}$ and
  $\Exp[\Delta_{t+1} ] \leq -(1 + \alpha)\epsilon \leq -\epsilon$. When
  $\rho > 0$ and $\mu \geq 0$, $\rho_{t+1} = \rho + W_{t+1}$
  but $\overline{\mu}_{t+1} = \overline{\mu}$ so that
  $\Delta_{t+1} = W_{t+1}$ and $\Exp[\Delta_{t+1} ] \leq -\epsilon$. Similarly, when
  $\rho = 0$ and $\mu < 0$,
  $\overline{\mu}_{t+1} = \overline{\mu} + W_{t+1}$ while
  $\rho_{t+1} = \rho + \max(0, W_{t+1})$; we may compute
  \[
    \Exp[\Delta_{t+1} ] \leq \frac{1 - \epsilon}{2}(1 + \alpha) - \frac{1 +
      \epsilon}{2}\alpha = \frac{1 - \epsilon}{2} - \epsilon\alpha =
    \frac{1 - \epsilon}{2} - \epsilon\left(\frac{1}{\epsilon} \cdot
      \frac{1 + \epsilon}{2}\right) = -\epsilon\,.
  \]
  Finally, when $\rho = \mu = 0$ exactly one of the two random
  variables $\rho_{t+1}$ and $\overline{\mu}_{t+1}$ differs from
  zero: if $W_{t+1} = 1$ then
  $(\rho_{t+1}, \overline{\mu}_{t+1}) = (1,0)$; likewise, if
  $W_{t+1} = -1$ then
  $(\rho_{t+1}, \overline{\mu}_{t+1}) = (0,-1)$. It follows that
  \[
    \Exp[\Delta_{t+1} ] \leq \frac{1 - \epsilon}{2} - \frac{1 +
      \epsilon}{2}\alpha \leq -\epsilon\,.
  \]

  \noindent
  Thus $
  \Exp[\Phi_k] = \Exp \sum_{t=1}^k \Delta_t  
  \leq -\epsilon k
  $. 
  We wish to apply Azuma's inequality to conclude that
  $\Pr[\Phi_k \geq 0]$ is exponentially small. For this purpose, we
  transform the random variables $\Phi_t$ to a related supermartingale by
  shifting them: specifically, define
  $\tilde{\Phi}_t = \Phi_t + \epsilon t$ and
  $\tilde{\Delta}_t = \Delta_t + \epsilon$ so that
  $\tilde{\Phi}_t = \sum_i^t \tilde{\Delta}_t$. Then
  \[
    \Exp[\tilde{\Phi}_{t+1} \mid \tilde{\Phi}_1, \ldots,
    \tilde{\Phi}_{t}] = \Exp[\tilde{\Phi}_{t+1} \mid W_1, \ldots,
    W_{t}]\leq \tilde{\Phi}_t\,,
    \qquad
    \tilde{\Delta}_t \in [-(1 + \alpha) + \epsilon, 1+ \alpha +
    \epsilon]\,,
  \]
  and $\tilde{\Phi}_k = \Phi_k + \epsilon k$. It follows
  from Azuma's inequality that
  \begin{align}\label{eq:azuma-bound}
    \Pr[\text{$w$ forkable}] 
    &= \Pr[\overline{\mu}_k = 0] \leq \Pr[\Phi_k \geq 0] = \Pr[\tilde{\Phi}_k \geq \epsilon k] 
    \nonumber \\ 
    &\leq \exp\left(-\frac{\epsilon^2 k^2}{2k (1 + \alpha + \epsilon)^2}\right)
       = \exp\left(-\left(\frac{2 \epsilon^2}{1 + 3 \epsilon + 2\epsilon^2}\right)^2 \cdot \frac{k}{2}\right) \nonumber \\
    &\leq \exp\left(-\frac{2\epsilon^4}{1 + 35\epsilon} \cdot k\right)
    \,.
    %                  \qedhere
  \end{align}

\newcommand{\muxr}{\mu_x^{(r)}}
\newcommand{\Snr}{S_k^{(r)}}
\newcommand{\Sr}{S^{(r)}}
\newcommand{\Srstar}{S^{(r^*)}}
\newcommand{\event}[1]{\mathsf{#1}}
\newcommand{\notevent}[1]{\overline{\event{#1}}}

%\vspace{-2ex}
\paragraph{If $x$ is not empty.} 
In this case, we go back to study the sequences $x$ and $y$ as in the statement of the claim.
Recall the reach distribution (i.e., the distribution of the random variable $\rho(x)$) 
$\DistRho_m : \Z \rightarrow [0,1]$ from~\eqref{eq:dist-rho}. 
Since $x = (x_1, \ldots, x_m)$ satisfies the $\epsilon$-martingale condition, 
Lemma~\ref{lemma:rho-stationary} states that $\DistRho_m \dominatedby \StationaryRho$.
We reserve the symbol $\muxr$ for the relative margin 
random walk $\mu_x$ which starts at a non-negative initial position $r$. 
Thus $\rho(x) = \mu_x(\epsilon) = r$, and
\begin{align}\label{eq:azuma-generic}
\Pr[\mu_x(y) \geq 0] 
&= \sum_{r \geq 0}{\DistRho_m(r) \Pr[\muxr(y) \geq 0]} 
\leq \sum_{r \geq 0}{\StationaryRho(r) \Pr[\muxr(y) \geq 0]} 
\, 
\end{align}
since the sequence $( \, \Pr[\muxr(y) \geq 0] \, )_{r=0}^\infty$ is non-decreasing and $\DistRho_m \dominatedby \StationaryRho$. Fix a ``large enough'' positive integer $r^*$ whose value will be assigned later in the analysis. 
Let us define the following events:
 \begin{itemize}
  %\item Event $\event{A}_r$:~when $r > r^*$. 
  \item Event $\event{B}_r$:~it occurs when $r \in [0, r^*]$ and the $\muxr$ walk is strictly positive on every prefix of $y$ with length at most $k/2$; and 
  \item Event $\event{C}_{r,s}$:~it occurs when $r \in [0, r^*]$ and 
  $\hat{y}$ is the smallest prefix of $y$ of length $s \in [r, k/2]$ 
  such that $\muxr(\hat{y}) = 0$. 
  We say that $\hat{y}$ is a witnesses to the event $\event{C}_{r, s}$.
\end{itemize}
%Note that these two events cannot happen simultaneously. 
The right-hand side of~\eqref{eq:azuma-generic} can be written as
\begin{align*}
     &\quad \sum_{r>r^*}{\StationaryRho(r) \Pr[\muxr(y) \geq 0]} 
		+ \sum_{r \leq r^*}{\StationaryRho(r) \Pr[\event{B}_r] \cdot \Pr\left[\muxr(y) \geq 0 \mid \event{B}_r\right]} \\
    &\quad+ \sum_{r \leq r^*}{\StationaryRho(r) \sum_{s = r}^{k/2}{\Pr[\event{C}_{r,s}] \cdot \Pr[\muxr(y) \geq 0 \mid \event{C}_{r,s}]} }
    \, .
\end{align*} 
We observe that the probabilities $\Pr[\muxr(y) \geq 0]$ and $\Pr[\muxr(y) \geq 0 \mid \event{B}_r]$ are at most one. 
In addition, recall that for two non-negative sequences $(a_i), (b_i)$ of equal lengths, 
we have $\sum{a_i b_i} \leq \max b_i$ if $\sum{a_i} \leq 1$. 
Thus~\eqref{eq:azuma-generic} can be simplified as
\begin{align}\label{eq:three-terms}
\Pr[\mu_x(y) \geq 0] 
 &\leq 
    \sum_{r > r^*}{\StationaryRho(r)} 
  + \sum_{r \leq r^*}{\StationaryRho(r) \Pr[\event{B}_r]} \nonumber \\
  &\quad+ \sum_{r \leq r^*}{\StationaryRho(r)\, \max_{r \leq s \leq k/2}{\Pr[\muxr(y) \geq 0 \mid \event{C}_{r,s}]} }
  \nonumber \\
 &\leq    
      \sum_{r > r^*}{\StationaryRho(r)}            
  + 
      \max_{r \leq r^*}{\Pr[\event{B}_r]}          
  + 
      \max_{\substack{r \leq r^* \\ r \leq s \leq k/2}}{\Pr[\muxr(y) \geq 0 \mid \event{C}_{r,s}]}   
\, .
\end{align}

\emph{The first term in~\eqref{eq:three-terms} } is the right-tail of the distribution $\StationaryRho$. 
Using Lemma~\ref{lemma:rho-stationary}, 
this quantity is at most $\beta^{r^*}$ where $\beta := (1-\epsilon)/(1+\epsilon)$. 
Furthermore, it can be easily checked that the above quantity is at most $\exp(-5 \epsilon/3)$.

\emph{The second term in~\eqref{eq:three-terms} } concerns the event
$\event{B}_r$ and calls for more care.  Define
\[
  \Snr := \sum_{t=0}^k {W_t}
\]
where $W_0 = r$ and the random variables $W_t$ are defined at the
outset of this proof for $t \geq 1$.  We know that the $\muxr$ walk
starts with $\rho(x) = \mu(x) = r \geq 0$.  Since $\event{B}_r$ holds,
both the margin $\mu_x(\hat{y})$ and the reach $\rho(x\hat{y})$ remain
non-negative for all prefixes $\hat{y}$ of length
$t = 1, 2, \cdots, k/2$.  These two facts imply that the random
variable $\muxr(\hat{y})$ is identical to the sum $\Sr_t$ for all
prefixes $\hat{y}$ of length $t = 1, 2, \cdots, k/2$.

To be precise,
\[
  \Pr[\event{B}_r] = \Pr[\Sr_t \geq 0 \quad \text{for all } t \leq k/2]\,.
\]
The latter probability is at most $\Pr[\Sr_{k/2} \geq 0]$ because the
event $\Sr_{k/2} \geq 0$ does not constrain the intermediate sums
$\Sr_t$ for $t < k/2$.  Since $\Pr[\Sr_{k/2} \geq 0]$ increases
monotonically in $r$, we conclude that the second term
in~\eqref{eq:three-terms} is at most $\Pr[\Srstar_{k/2} \geq 0]$.  Now
we are free to shift our focus from the relative margin walk to the
sum of a martingale sequence.

For notational clarity, let us write $S := \Srstar_{k/2}$. 
Since the sequence $(w_t)$ obeys the $\epsilon$-martingale condition, 
$\Exp S$ is at most $M := r^* - k\epsilon/2$. 
Let us set $r^* = W_0 = k\epsilon/4$. Then $\Exp S$ is at most $-k\epsilon/4$ and Azuma's inequality gives us
\[
\Pr[S \geq 0] 
= \Pr[(S - \Exp S) \geq k\epsilon/4] 
\leq \exp\left( - \frac{(k\epsilon/4)^2}{2(k/2)\cdot 2^2}\right) 
= \exp\left( -\frac{k \epsilon^2}{64} \right)
\, .
\]
This is an upper bound on the second term in~\eqref{eq:three-terms}.

\emph{The third term in~\eqref{eq:three-terms}} concerns the event $\event{C}_{r,s}$ and it can be bounded using 
our existing analysis of the $|x|=0$ case. 
Specifically, suppose $y = \hat{y} z$ where
$\hat{y}$ is a witness to the event $\event{C}_{r,s}$. 
Since the $\muxr$ walk remains non-negative over the entire string $\hat{y}$, 
it follows that $\rho(x\hat{y}) = \mu(x\hat{y}) = 0$ 
and as a consequence, the $\mu_{x\hat{y}}$ walk on $z$ is identical to 
the $\mu$ walk on $z$. 
Our analysis in the $|x| = 0$ case suggests that 
$\Pr[\mu(z) \geq 0]$ is at most $A(k-s, \epsilon)$ 
where $|z| = k - s$ and $A(k, \epsilon)$ is the bound in~\eqref{eq:azuma-bound}. 
Since $A(\cdot,\epsilon)$ decreases monotonically in the first argument, 
$A(k-s, \epsilon)$ is at most $A(k/2, \epsilon)$. 
However, since the last quantity is independent of $r$, 
the third term in~\eqref{eq:three-terms} is at most 
$A(k/2, \epsilon) = \exp\left( -k \epsilon^4/(1+35\epsilon) \right)$.

Returning to~\eqref{eq:three-terms} and using $r^* = k\epsilon/4$, we get
\begin{align*}
\Pr[\mu_x(y) \geq 0] 
 &\leq    \exp\left(-\frac{5 \epsilon}{3} \cdot \frac{k\epsilon}{4} \right)  
        + \exp\left(-\frac{2\epsilon^4}{1 + 35\epsilon} \cdot \frac{n}{2}\right)
        + \exp\left( -\frac{k \epsilon^2}{64} \right)
\, .
\end{align*}
It is easy to check that the above quantity is at most
$
  3 \exp\left( - k \epsilon^4/(64 + 35 \epsilon) \right) 
= 3 \exp\left( - \epsilon^4 (1 - O(\epsilon) ) k/64 \right)
%\,.
$.

% $\qed$
\end{proof}

% \section{Figures}
% \label{sec:figures}
% \input{appendix_figures}

\end{document}